\documentclass[authoryear, sn-chicago]{sn-jnl}%
\usepackage{multirow}
\usepackage{amsmath,amssymb,amsfonts}%
\usepackage{amsthm}%
\usepackage{mathrsfs}%
\usepackage[title]{appendix}%
\usepackage{xcolor}%
\usepackage{textcomp}%
\usepackage{manyfoot}%
\usepackage{booktabs}%
\usepackage{algorithm}%
\usepackage{algorithmicx}%
\usepackage{algpseudocode}%
\usepackage{listings}%
\usepackage{enumitem}
\usepackage{longtable}
\usepackage{color}
\usepackage{natbib}
\usepackage{tabularx}   
\usepackage[bf]{caption}
\usepackage{caption}
\usepackage{lineno}
\usepackage{hyperref}
\usepackage{cleveref}
\usepackage{bibunits} 
\usepackage[T1]{fontenc}
\usepackage[utf8]{inputenc}
\usepackage{babel}
\usepackage[export]{adjustbox}
\usepackage{subfig}
\usepackage{array}
\usepackage{setspace}
\usepackage{cuted}
\usepackage{algcompatible}
\usepackage{scrextend}
\usepackage{hyperref}
\usepackage{xtab}
\usepackage{footnote}
\usepackage{varwidth}
\usepackage{multicol}
\usepackage{arydshln}
\usepackage{caption}
\usepackage{mathtools}
\usepackage{algpseudocode}
\usepackage{hhline}
\usepackage{pgfplots}
\usepackage[font={small}]{caption}
\renewenvironment{table}[1][]%
{\tableorg[#1]%
\tablebodyfont%
\renewcommand\footnotetext[2][]{{\removelastskip\vskip3pt%
\let\tablebodyfont\tablefootnotefont%
\hskip0pt\if!##1!\else{\smash{$^{##1}$}}\fi##2\par}}%
}{\endtableorg}

\usepackage{float}

\usepackage{comment}

\usepackage{import}

\def\P{\mathbb P}
\def\E{\mathbb E}
\def\R{\mathbb R}
\def\Var{\mathbb Var}
\def\Cov{\mathbb Cov}
\def\du{du}
\def\P{\mathbb P}
\def\R{\mathbb R}
\def\1{\mathbf 1}
\def\x{\mathbf x}
\def\z{\mathbf z}
\def\y{\mathbf y}
\def\E{\mathbb E}
\def\W{\mathbb W}
\def\w{\mathbf w}

\DeclarePairedDelimiter\abs{\lvert}{\rvert}
\DeclarePairedDelimiter\norm{\lVert}{\rVert}

\DeclareFontEncoding{LS1}{}{}
\DeclareFontEncoding{LS2}{}{\noaccents@}
\DeclareFontSubstitution{LS1}{stix}{m}{n}
\DeclareFontSubstitution{LS2}{stix}{m}{n}
\DeclareMathAlphabet\mathscr{LS1}{stixscr}{m}{n}
\SetMathAlphabet\mathscr{bold}{LS1}{stixscr}{b}{n}
\DeclareMathAlphabet\mathcal{LS2}{stixcal}{m}{n}
\SetMathAlphabet\mathcal{bold}{LS2}{stixcal}{b}{n}

\DeclareFontFamily{U}{BOONDOX-calo}{\skewchar\font=45 }
\DeclareFontShape{U}{BOONDOX-calo}{m}{n}{ <-> s*[1.05] BOONDOX-r-calo}{}
\DeclareFontShape{U}{BOONDOX-calo}{b}{n}{ <-> s*[1.05] BOONDOX-b-calo}{}
\DeclareMathAlphabet{\mathcalboondox}{U}{BOONDOX-calo}{m}{n}
\SetMathAlphabet{\mathcalboondox}{bold}{U}{BOONDOX-calo}{b}{n}
\DeclareMathAlphabet{\mathbcalboondox}{U}{BOONDOX-calo}{b}{n}

\DeclareMathAlphabet\mathbfcal{OMS}{cmsy}{b}{n}
\DeclareMathAlphabet\mathbfscr{OMS}{cmsy}{b}{n}

\newcounter{algsubstate}

\algrenewcommand\algorithmicforall{\textbf{foreach}}
\algrenewcommand\algorithmicindent{.8em}

\newtheorem*{theorem}{Theorem}

\newcommand{\comm}[1]{}

\makeatletter
\newcommand{\inlineitem}[1][]{
\ifnum\enit@type=\tw@
{\descriptionlabel{#1}}
\hspace{\labelsep}
\else
\ifnum\enit@type=\z@
\refstepcounter{\@listctr}\fi
\quad\@itemlabel\hspace{\labelsep}
\fi}

\makeatother

\theoremstyle{thmstyletwo}%

\theoremstyle{thmstylethree}%

\setcitestyle{aysep={}}

\hypersetup{colorlinks = true, linkcolor = blue, filecolor = blue, urlcolor = blue}

\begin{document}

\title[Article Title]{Modelling Home Range and Intraspecific Spatial Interaction in Wild Animal Populations}

\author*[1,2]{\fnm{Fekadu L.} \sur{Bayisa}}\email{fbayisa@uoguelph.ca; ORCID: https://orcid.org/0000-0003-1421-6825}

\author[3]{\fnm{Christopher L.} \sur{Seals}}\email{cls0035@auburn.edu}

\author[3]{\fnm{Hannah J.} \sur{Leeper}}\email{hzl0122@auburn.edu}

\author[3]{\fnm{Todd D.} \sur{Steury}}\email{tds0009@auburn.edu}

\author[1]{\fnm{Elvan} \sur{Ceyhan}}\email{ezc0066@auburn.edu}

\affil*[1]{\orgdiv{Department of Mathematics and  Statistics}, \orgname{Auburn University}, \orgaddress{\street{221 Parker Hall}, \city{Auburn}, \postcode{36849-5319}, \state{AL}, \country{USA}}}

\affil[2]{\orgdiv{Department of Mathematics and  Statistics}, \orgname{University of Guelph}, \orgaddress{\street{50 Stone Road East}, \city{Guelph}, \postcode{N1G 2W1}, \state{ON}, \country{Canada}}}

\affil[3]{\orgdiv{College of Forestry, Wildlife and Environment}, \orgname{Auburn University}, \orgaddress{\street{602 Duncan Dr}, \city{Auburn}, \postcode{36849}, \state{AL}, \country{USA}}}

\abstract{Interactions among individuals from the same-species of wild animals are an important component of population dynamics. An interaction can be either static (based on overlap of space use) or dynamic (based on movement). The goal of this work is to determine the level of static interactions between individuals from the same-species of wild animals using 95\% and 50\% home ranges, as well as to model their movement interactions, which could include attraction, avoidance (or repulsion), or lack of interaction, in order to gain new insights and improve our understanding of ecological processes. Home range estimation methods (minimum convex polygon, kernel density estimator, and autocorrelated kernel density estimator), inhomogeneous multitype (or cross-type) summary statistics, and envelope testing methods (pointwise and global envelope tests) were proposed to study the nature of the same-species wild-animal spatial interactions.  This study provides comprehensive, self-contained methodological details for investigating spatial interactions between individuals of the same species in wildlife populations.  Using GPS collar data, we applied the methods to quantify both static and dynamic interactions between black bears in southern Alabama, USA. In general, our findings suggest that the black bears in our dataset showed no significant preference to live together or apart, i.e., there was no significant deviation from independence toward association or avoidance (i.e., segregation) between the bears.  This can be loosely interpreted to mean that a black bear is generally indifferent to the presence of other black bears living or wandering nearby.}  

\keywords{ Inhomogeneous multitype point process, cross-type summary function,   minimum convex polygon,  kernel density estimator, autocorrelated kernel density estimator, pointwise and global envelope tests}

\maketitle

\section{Introduction}\label{sec1}

Biodiversity sustains life on Earth, providing essential resources such as food, shelter, medicine, and recreation. However, it faces threats from habitat loss, climate change, overexploitation of resources, and other human activities. Habitat loss and fragmentation, identified as the most significant threats to biodiversity worldwide \cite{allendorf2017genetics}, underscore the need for understanding how species use spaces during range expansions. Such insights can guide land management strategies, aid in designing conservation areas \citep{bocedi2014mechanistic}, and identify suitable habitats for future population growth \citep{mladenoff1999predicting}. Besides, they can help mitigate human-wildlife conflicts \citep{wilton2014distribution} and improve knowledge of movement patterns of species and social interactions. Modelling space-use dynamics offers a framework for organizing management actions effectively \citep{fortin2020quantitative, wysong2020space}. To this end, home range estimation and understanding spatial species interactions are crucial for effective species management and conservation.   \\

This work is aimed at investigating the use of space by wild animal species and the types of spatial interactions that may exist between pairs of the same species of wild animals. We illustrate the proposed methodology with a black bear population in Alabama, USA. By employing aggregated telemetry locations of wild animals, we can estimate their total space use, often referred to as the home range. The swift adoption of telemetry and home-range estimators utilizing telemetry data has led to a substantial body of literature on wild animal home range estimation 
\citep{powell2012home}.\\

According to \cite{burt1943territoriality}, the home range is the area covered by an individual wild animal during its regular activities of foraging, mating, and caring for its young. Occasional sallies outside this area, likely exploratory, are not considered part of the home range. While this definition forms the basis for the meaning of a home range, the definition does not offer guidance on quantifying occasional sallies or identifying the area from which they occur. Utilizing observed data is essential to understanding a wild animal's cognitive map of its home range.  With some level of predictability, home range estimators of wild animals should delimit where wild animals can be found. These estimators  quantify the likelihood that the wild animal is in different places. Moreover, home range estimators  can be used to assess the importance of different places to the wild animal \citep{powell2012movements}. \\

The Minimum Convex Polygon (MCP), a widely used home range estimator \citep{horne2009habitat}, is the smallest convex polygon that covers all known wild animal locations \citep{hayne1949calculation}. However, it neglects internal structures and central tendencies within home ranges, which are crucial for understanding wildlife ecology \citep{powell2012movements}. The utilization distribution, derived from relocation data, indicates the intensity of space use. Most studies adopt a 95\% probability threshold for home range estimation \citep{powell2012movements}, excluding the 5\% most isolated locations to reduce bias. Advances in data collection have enhanced home range estimation and conservation efforts \citep{morato2016space}. The most commonly used statistical tools for estimating wild animal home ranges are MCP and kernel density estimation (KDE) \citep{fleming2015rigorous}, though both have limitations: MCP is non-probabilistic, and KDE assumes independent and identically distributed (iid) location processes. For autocorrelated data, particularly from GPS tracking, KDE may underestimate home ranges. Continuous-time stochastic models address autocorrelation \citep{calabrese2016ctmm}, enabling variogram modeling and improving home range estimation, even with correlated data \citep{fleming2014non}. \\

The conditional distribution of encounters is proposed to characterize the long-term location probabilities of encounters for animal movement within home ranges. An estimator for the spatial distribution of encounter events is developed, directly building upon one of the most widely used analyses in movement ecology, namely home-range estimation \citep{noonan2021estimating}.  According to \cite{braunstein2020black}, animal movement dynamics can be influenced by resource selection, which, in turn, is governed by the movement limitations of the animal.   In stark contrast, most real animals exhibit non-uniform space use within spatially restricted home ranges \citep{bowen1982home,burt1943territoriality, fleming2014fine, kie2010home, martinez2020range, moorcroft2006mechanistic,  noonan2019comprehensive, tucker2019large, powell2000animal}, and encounters between individuals do not occur uniformly in space, but are instead concentrated at territorial boundaries  \citep{bermejo2004home, ellwood2017active, nievergelt1998group,   wilson2012construction}, in/around heavily used habitats and/or habitat features \citep{weckel2006jaguar, whittington2011caribou} or at key resources \citep{de2010spatial, price2013habitat}. We therefore base our work on recent analytical work by \cite{martinez2020range} incorporating non-uniform movement within home ranges into encounter theory.\\

This work is also aimed at exploring the nature of mutual interactions among same-species wild animals. Environmental factors like the amount of sunlight, precipitation, terrain type, soil properties, nutrient availability, and spatial distribution of nutrients can affect the way wild animals interact with each other. Complex spatial  interactions can also result from biological mechanisms like competition, reproduction, and mortality. As a result, due to the limitations of the human eye in discerning complex interaction patterns beyond basic trends and second-order aspects  \citep{ripley1976second}, ecologists have turned to spatial statistics \citep{diggle1983statistical, gelfand2010handbook,  illian2008statistical, van2019theory} to quantify and test hypotheses related to interaction patterns  \citep{gelfand2019handbook, wiegand2013handbook}. Functional statistics in exploratory data analysis for spatial point patterns offer a means to capture diverse aspects of the underlying pattern, including tendencies to seek out or avoid other individuals within and between species \citep{de2022testing}.  In shared habitats, wild animals can partition their space through attraction, avoidance (or repulsion), or exhibit no interaction. Statistical methods can discern the specific type of spatial interaction between pairs of the same-species wild animals.\\

We propose using inhomogeneous multitype summary statistics and envelope testing to explore spatial interactions among same-species wild animal pairs. For intraspecific interactions, cross-versions of inhomogeneous summary statistics, such as the inhomogeneous cross-type $L$-function, which is a scaled version of the $K$-function, and the cross-type $J$-function, can be applied \citep{moller2003statistical}. Monte Carlo tests help evaluate the null hypothesis that animal pairs do not exhibit spatial interactions, comparing observed summary functions with pointwise envelopes from null model simulations. To avoid multiple testing issues, it is crucial to fix the spatial lag beforehand \citep{shaffer1995multiple}. To address potential misinterpretations of pointwise envelopes \citep{baddeley2015spatial}, we use a numerical index-based approach, including the \emph{maximum absolute deviation (MAD)} and the \emph{Diggle-Cressie-Loosmore-Ford (DCLF)} tests \citep{loosmore2006statistical}, for simultaneous hypothesis testing across various spatial lags. \\

The main contributions of our work are as follows: we formulated existing and state-of-the-art methods for estimating home ranges of wild animals, applied them to relocation data for home range estimation, and introduced novel inhomogeneous cross-type summary functions from spatial statistics to analyze spatial interactions among pairs of the same species. This application represents a unique use of spatial statistics in ecology, establishing a foundation for future research on home range estimation and spatial interactions, particularly within our specified domain. The methodology employed in this article is adaptable for other wild animal species beyond our application domain or can be suitably modified for various plant species. The article is structured as follows: Section \ref{StatisticalFrameworkToModelHomeRanges} presents statistical methods, Section \ref{ApplicationToBlackBear} examines a black bear case study, Section \ref{Discussion} covers implications and summarizes the work, and Section \ref{FutureWork} suggests future directions.

\section{Statistical Methods}\label{StatisticalFrameworkToModelHomeRanges}
\subsection{Modelling of Home Ranges}
This section provides an overview of the statistical methods used to estimate the home ranges of wild animal species. Towards this end, the relocation of wild animal species is assumed to be a stochastic process. Let \[\w = \left\lbrace \z_{t_{i}} =\left(x_{t_{i}}, y_{t_{i}}\right)^{T}\mid t_{i} \ge 0, \; i = 1, 2, \cdots, n\right\rbrace  \subset \W,\] denote the relocation data of a wild animal in a study region $ \W\subset \mathbb{R}^{2}$ at times $t_{i}$, $i = 1, 2, \cdots, n$. Here, $x_{t_{i}}$ and $y_{t_{i}}$ represent the longitude and latitude coordinates of a wild animal tracked by GPS at times $t_{i}$.

\subsubsection{Minimum convex polygon estimation}
The minimum  convex polygon is a widely used technique for estimating home ranges. It generates the smallest convex polygon covering observed relocation data $\w$, providing a simple yet popular method \citep{mohr1947table}. This estimation lacks a probabilistic model, relying on peripheral points, making it susceptible to outliers which potentially influence home range estimation irrespective of the distribution of interior points.

\subsubsection{Kernel density estimation} \label{KDEsub}
The home range of wild animal species can be described in terms of a probabilistic model. \cite{worton1989kernel}  proposed a nonparametric approach,  namely \emph{a kernel density estimator},  to estimate the intensity $p\left(\mathbf{z}\right)$  of home range use at a location $\mathbf{z}$  given by 
\begin{align}\label{Mo'a2021}
\hat{p}_{\boldsymbol\Lambda}\left(\mathbf{z}\right)  = \frac{1}{n} \sum_{i = 1}^{n}\abs{\boldsymbol\Lambda}^{-\frac{1}{2}}\mathcal{K}\left(\boldsymbol\Lambda^{-\frac{1}{2}}\left(\mathbf{z} - \mathbf{z}_{t_{i}}\right) \right),
\end{align}
where $\mathcal{K}:\R^{2} \rightarrow \left[0, \infty\right) $ is a kernel function in  $\R^{2}$, and $\boldsymbol\Lambda$ is a symmetric and positive definite $2 \times 2$ bandwidth matrix. The home ranges of wild animal species can be defined as the $c$-level set $\left\lbrace \mathbf{z}: p\left(\mathbf{z}\right) \ge c\right\rbrace$ of the utilization density $p\left(\mathbf{z}\right)$ with a probability content of $100(1-\alpha)\%$, $\alpha \in \left(0, 1\right)$. That is, $1-\alpha = \int_{\left\lbrace \mathbf{z}:  p\left( \mathbf{z}\right) \ge c \right\rbrace}p\left(\mathbf{z}\right)d\mathbf{z}$.  The kernel home range estimator, based on Equation \eqref{Mo'a2021}, is the $\hat{c}$-level set $\left\lbrace \mathbf{z}: \hat{p}_{\boldsymbol\Lambda}\left(\mathbf{z}\right) \ge \hat{c}\right\rbrace$ of the kernel density estimator $\hat{p}_{\boldsymbol\Lambda}\left(\mathbf{z}\right)$. If $\alpha = 0.05$, then $\hat{c}$ is chosen to achieve a specific probability content, such as $0.95 = \int_{\left\lbrace\mathbf{z}:   \hat{p}_{\boldsymbol\Lambda}\left( \mathbf{z}\right) \ge \hat{c}\right\rbrace}\hat{p}_{\boldsymbol\Lambda}\left(\mathbf{z}\right)d\mathbf{z}.$\\

The choice of kernel function has minimal impact on the accuracy of the kernel density estimator compared to the influence of bandwidth \citep{wand60jones}. \cite{worton1989kernel} utilized a constrained bandwidth matrix $\lambda \mathbf{I}$, where $\mathbf{I}$ is the identity matrix, dependent on a single smoothing parameter $\lambda > 0$. The probability of a wild animal being in an infinitesimal region centered at $\mathbf{z}$ in the study area $\W$ can be approximated by
\begin{align}\label{Mo'a2021b}
p\left(\mathbf{z}\right)  \approx \varphi\left( \mathbf{z}, \boldsymbol\mu, \boldsymbol\Sigma\right) 
= \left[2\pi \det\left(\boldsymbol\Sigma\right)\right]^{-\frac{1}{2}}\exp\left\lbrace -\frac{1}{2}\left(\mathbf{z} - \boldsymbol\mu\right)^{T} \boldsymbol\Sigma^{-1}\left( \mathbf{z} - \boldsymbol\mu\right)\right\rbrace, 
\end{align}
where $\varphi\left( \mathbf{z}, \boldsymbol\mu, \boldsymbol\Sigma\right)$ is the probability density function of a multivariate normal distribution with mean $\boldsymbol\mu$ and covariance matrix $\boldsymbol\Sigma$. The parameters $\boldsymbol\mu$ and $\boldsymbol\Sigma$ can be estimated from the relocation data. The mean integrated squared error $MISE\left( \lambda\right) = \E\left[ \int \left( \hat{p}_{\lambda}\left(\mathbf{z}\right) - p\left(\mathbf{z}\right)\right)^{2}d\mathbf{z}\right]$ of the reference density function in Equation \eqref{Mo'a2021b} and the kernel density estimator in Equation \eqref{Mo'a2021} can be used for selecting $\lambda$ through either the 'ad hoc' method or least-squares cross-validation, see the details in \cite{van2020infill} and the reference therein. 

\subsubsection{Autocorrelated kernel density estimation}
The  autocorrelation that may exist in wild animal tracking data, used for home range estimation, can violate the assumption of \emph{iid} data in kernel density estimation. When we refer to autocorrelation in wild animal relocation data, we mean the statistical correlation between an individual's current and past locations, persisting into the future (i.e., temporal autocorrelation). Applying kernel density estimation to autocorrelated data leads to underestimated home ranges \citep{fleming2014fine}. Underestimation results from the fact that a set of \emph{iid} observations carries more information about the home range than an equivalent number of highly autocorrelated observations \citep{fleming2015rigorous}. Thus, assuming \emph{iid} data in home range estimation overestimates information and underestimates home ranges. Higher resolution in tracking wild animal movement can increase autocorrelation in relocation data. Employing such relocation data can degrade the accuracy of the kernel density estimator for home range analysis.   To mitigate the limitations of kernel density estimation, we employ autocorrelated kernel density estimation in our dataset. We assume that the relocation data represent a sample from a nonstationary autocorrelated continuous movement process. Introduced by \cite{fleming2015rigorous}, this method enhances home range estimation accuracy by effectively incorporating the information content of autocorrelated data. Autocorrelated kernel density estimation incorporates movement effects through the autocorrelation function, derived from a fitted movement model or directly estimated from the data \citep{fleming2014fine}. When autocorrelation approaches zero and relocation data become independent, autocorrelated kernel density estimation converges to standard kernel density estimation. Therefore, autocorrelated kernel density estimation serves as a generalization of standard kernel density estimation \citep{fleming2015rigorous}. Wild animal species movement can be represented as a nonstationary process, allowing the mean and autocorrelation functions to change over time. Utilizing a time-varying autocorrelation function enhances the modelling of spatial dependence in wild animal species relocation data. The relocation data of wild animal species can be associated with specific time points $t_{1}$, $t_{2}$, $\cdots$, $t_{n}$. The probability of a wild animal  being in an infinitesimal region centered at location $\mathbf{z}$ at time $t_{i}$ can be approximated using the probability density function  $\mathcal{N}\left( \mathbf{z}, \boldsymbol\mu_{t_{i}}, \boldsymbol\Sigma_{t_{i}}\right)$ of a multivariate normal distribution given by
\begin{align}\label{Mo'a2022}
p\left(\mathbf{z}, t_{i}\right)  \approx \mathcal{N}\left( \mathbf{z}, \boldsymbol\mu_{t_{i}}, \boldsymbol\Sigma_{t_{i}}\right)
=  \frac{\exp\left\lbrace -\frac{1}{2}\left(\mathbf{z} - \boldsymbol\mu_{t_{i}}\right)^{T} \boldsymbol\Sigma_{t_{i}} ^{-1}\left( \mathbf{z} - \boldsymbol\mu_{t_{i}}\right)\right\rbrace}{\left[2\pi \det\left( \boldsymbol\Sigma_{t_{i}}\right)\right]^{\frac{1}{2}}}.
\end{align} 
For $i, k = 1, 2, \cdots, n$, \[\boldsymbol\mu_{t_{i}} = \E\left[\mathbf{z}_{t_{i}}\right],\;\; \boldsymbol\Sigma_{t_{i}, t_{k}} = \E\left[\left(\mathbf{z}_{t_{i}} - \boldsymbol\mu_{t_{i}}\right)\left(\mathbf{z}_{ t_{k}} - \boldsymbol\mu_{ t_{k}}\right)^{T}\right],\] represent the mean and the autocorrelation function. Here, $\boldsymbol\Sigma_{t_{i}}$ represents the covariance structure when  the time points $t_{i}$ and $ t_{k}$ in $\boldsymbol\Sigma_{t_{i}, t_{k}}$ are the same.  Based on Equation \eqref{Mo'a2022}, the probability of a wild animal being in an infinitesimal region centered at location $\mathbf{z}$ can be approximated as a time-averaged density given by 
\begin{align}\label{Mo'a2023}
\P\left(\mathbf{z}\right)  \approx
\frac{1}{n}\displaystyle\sum_{i = 1}^{n} \left[2\pi \det\left( \boldsymbol\Sigma_{t_{i}}\right)\right]^{-\frac{1}{2}}\exp\left\lbrace -\frac{1}{2}\left(\mathbf{z} - \boldsymbol\mu_{t_{i}}\right)^{T} \boldsymbol\Sigma_{t_{i}}^{-1}\left( \mathbf{z} - \boldsymbol\mu_{t_{i}}\right)\right\rbrace.
\end{align}
The parameters $\boldsymbol\mu_{t_{i}}$ and $\boldsymbol\Sigma_{t_{i}}$ can be obtained from the spatio-temporal data, and the bandwidth $\lambda$ can be estimated using the mean integrated squared error of the kernel density estimator in Equation \eqref{Mo'a2021} with reference to the density function in Equation \eqref{Mo'a2023}.

\subsection{Stochastic movement processes}
The relocation data set $\w$ for wild animal species can be viewed as a realization of a stochastic process, which is a sequence of time-indexed random variables $\mathbf{z}\left(t\right)$  that can be correlated in time. The mean location $\boldsymbol\mu\left( t\right) $ of  a nonstationary stochastic movement process $\mathbf{z}\left(t\right)$ can reveal shifts in mean location over time. In our case, such shifts could correspond to movement behaviours such as movement within home ranges (or territories). In such cases, we may obtain a complete analysis by considering the autocorrelation function or, equivalently, the \emph{semivariance function} of the stochastic movement process.  The semivariance function measures spatial distance variability between locations of the wild animal. \\

Most time-series methods for estimating semivariance functions assume stationarity, implying that the statistical properties of a stochastic process remain constant over time. However, ecological systems exhibiting daily, seasonal, or annual cycles violate this assumption, requiring a nonstationary approach when analyzing wild animal movement data. In nonstationary processes involving wild animal movement, the semivariance $\gamma(t_1, t_2)$ between locations $\mathbf{z}(t_1)$ and $\mathbf{z}(t_2)$  at times $t_{1}$ and $t_{2}$ can be influenced not only by the lag $\tau = t_{2}-t_{1}$ but also by the absolute times. To address this, we calculate the average time $\bar{t} = \left( t_{1} + t_{2}\right)/2$ for each pair of wild animal locations.  In the context of individual wild animal movement analysis, reliability in semivariance estimates is confined to the lag range $t_{d} < \tau \ll T$, where $t_{d}$ is the sampling time step, and $T$ is the sampling duration. According to \cite{fleming2014fine}, avoiding direct estimates of the mean and variance provides an unbiased estimator of the semivariance function. That is to say that the  semivariance function has unbiased estimators, while the  autocorrelation function does not. Using the described nonstationary approach, the method-of-moments estimator for the semivariance function in evenly sampled data can be expressed as follows.
\begin{align}\label{Semivariance}
\hat{\gamma}\left(\tau\right)  = \displaystyle\frac{1}{2n\left(\tau\right) }\sum_{\bar{t}}\left[ \mathbf{z}\left(\frac{\bar{t} + \tau}{2}\right)-\mathbf{z}\left(\frac{\bar{t} - \tau}{2}\right)\right]^{2}.
\end{align}
Here, $n\left(\tau\right)$ represents the number of wild animal location pairs separated by lag $\tau$, and $\hat{\gamma}\left(\tau\right)$ is obtained by summing over the time average value $\bar{t}$ for the lag $\tau$. For more details, refer to the Online Resources and the reference therein. With increasing lag, fewer wild animal locations are available for semivariance estimation. Thus, more reliable semivariance estimates can be obtained from shorter lags in evenly sampled relocation data. That is, fine-scale features at smaller lags are as important as, or more important than, larger-scale features at larger lags in driving more reliable semivariance estimates. \\

Empirical semivariance, obtained from Equation \eqref{Semivariance}, can be plotted against the time lag between relocations, offering an unbiased estimation of the autocorrelation structure in the relocation data. For a range-resident wild animal, the semivariance of its relocation data should eventually reach an asymptote proportional to its home range. If the semivariance does not approach an asymptote with increasing time lag, the relocation data may be unsuitable for home range analysis \citep{calabrese2016ctmm}.\\

Closer-in-time wild animal locations exhibit greater similarity than distant ones. Directional persistence in wild animal motion results in autocorrelated velocities, indicating that the direction and speed of a wild  animal at one point in time correlate with those at other points. \cite{calabrese2016ctmm} suggest that position autocorrelation, velocity autocorrelation, and range residency are useful for classifying continuous-time stochastic processes (or movement models). The \emph{iid} process models wild animal movement, assuming uncorrelated locations and velocities, a simplification in traditional home range estimation. In contrast, the \emph{Brownian process} lacks velocity autocorrelation, limiting its ability to capture diverse movement patterns \citep{turchin1998quantitative}. On the other hand, the Ornstein-Uhlenbeck process, with mean-reverting behaviour and attraction to the mean, suits data without directional persistence but with confined space use. The \emph{Ornstein-Uhlenbeck with foraging process} is effective for analyzing wild animal relocation data with correlated velocities and limited space use. It is also applicable to diverse datasets with fine sampling for exploring velocity autocorrelation and prolonged range residence. For further details, refer to the Online Resource and the references therein. Wild animal relocation data can also be isotropic or anisotropic; with isotropic processes lacking directional dependence and anisotropic processes exhibiting variation based on the direction of interest.\\

Empirical semivariance plots derived from Equation \eqref{Semivariance} offer insights into wild animal movement behaviour. They assist in evaluating theoretical semivariance models for \emph{iid}, Ornstein-Uhlenbeck, and Ornstein-Uhlenbeck with foraging processes. These theoretical models are fitted to empirical data using maximum likelihood, and the fitted models are compared using the Akaike information criterion (AIC) to determine the best model. The selected model can then be applied for estimating home ranges through autocorrelated kernel density estimation.

\subsection{Spatial interaction modelling}
This section outlines modelling wild animal relocation data as multitype (or marked) point patterns, where the multitype spatial point patterns (or relocation data of wild animals) encompass both the spatial locations (or points) of the wild animals and their identities.
\subsubsection{Multitype spatial point pattern } 
When we say an {\em event} in a given geographical (or spatial) region $\W\subseteq\R^2$,  we mean a wild animal location (GPS position) in $\W$  during a given time period.  The total collection of events be referred to as the {\em wild animal relocation data set}, which can be thought of as a collection of locations $\{\mathbf{z}_1, \mathbf{z}_2, \cdots, \mathbf{z}_n\}\subseteq \W$, $n\geq0$, and  a {\em mark} $m_j$, which is attached to each location $\mathbf{z}_j$, $j =1, 2, \cdots, n$.  In advance, we do not know $n$, i.e.,~the number of wild animal locations within $\W$  during the time period in question. Such data set $\mathbfcal{Z} = \{(\mathbf{z}_{1}, m_{1}), (\mathbf{z}_{2}, m_{2}), \cdots,(\mathbf{z}_{n},m_{n})\}\subseteq \W\times \mathcal{M}$, where $\mathcal{M}$ is the set of marks, is most naturally classified as a {\em marked point pattern} \citep{baddeley2015spatial}.  When the mark space is discrete, say,  $\mathcal{M} =  \left\lbrace 1, 2, \cdots, k\right\rbrace$, $k > 1$, we say that  $\mathbfcal{Z}$ is a \emph{multitype point pattern} and we may split $\mathbfcal{Z}$ into the marginal (or purely spatial) point processes $\mathbfcal{Z}_{i} = \left\lbrace  \mathbf{z}_{j}: (\mathbf{z}_{j},  m_{j}) \in \mathbfcal{Z},  m_{j} = i\right\rbrace$, $i = 1, 2, \cdots, k.$  This collection may formally be represented by the vector $\left(\mathbfcal{Z}_{1},  \mathbfcal{Z}_{2}, \cdots, \mathbfcal{Z}_{k}\right)$, which may be referred to as a \emph{multivariate point process}.  In a recent study, \cite{bayisa2023regularised} used this approach to model real-world problem.

\subsubsection{Summary functions for pairs of types} \label{Summary2080}
We assume that the spatial point process $\mathbfcal{Z}_{i}$ corresponding to a wild  animal species (or type) $i$ is a nonstationary process. When we say a nonstationary process, we mean the probability distribution of the point process is not invariant under translation. Intuitively, we mean that the relocation of the wild animal in $\W\subseteq\R^{2}$ does not look the same from every angle in terms of relocation density and intra-relocation interactions.\\

\cite{moller2003statistical} proposed  an inhomogeneous  cross-type $K$-function $K^{ij}_{inhom}\left( r\right)$ to describe the proportional of points of type (or wild animal species) $j$ seen within distance $r \ge 0$ of a typical point of type $i$. It is obtained by weighting the expected number of type $j$ points  lying within a distance $r \ge 0 $ of a typical type $i$  point by the intensity of type $j$  points. That is,
\begin{align}\label{Kfun}
K^{ij}_{inhom}\left(r\right) = \E\left[\sum_{\mathbf{z}\in\mathbfcal{Z}_{j}}\frac{1_{\left\lbrace \norm{\mathbf{z}-\mathbf{v}}\le r\right\rbrace}}{\lambda\left(\mathbf{z} \right)}\;\big|\; \mathbf{v}\in\mathbfcal{Z}_{i}\right],
\end{align} 
where the distance $r \ge 0$ and $1_{\left\lbrace \cdot\right\rbrace} $ is an indicator function.  The estimator of the $K^{ij}_{inhom}\left(r\right)$ that involves edge correction can be found in the Online Resource and the reference therein.  Inhomogeneous  cross-type $L$-function can be obtained from Equation \eqref{Kfun} using the standard relationship $L^{ij}_{inhom}\left(r\right) = 	\sqrt{K^{ij}_{inhom}\left(r\right)/\pi}$. Notice that when type $j$ points are mostly found in the $r$-neighbourhood of type $i$ points, $K^{ij}_{inhom}(r)$ takes values greater than $\pi r^2$, which is the value of $K^{ij}_{inhom}(r)$ when the spatial points of components (or wild animal species) $i$ and $j$ are independent (or when there is no spatial interaction between types $i$ and $j$). This is equivalent to saying that when type $j$ points  are mostly found in the $r$-neighbourhood of  type $i$ points, $L^{ij}_{inhom}\left(r\right)$ takes values greater than $r$, which is the value of $L^{ij}_{inhom}\left(r\right)$ in the absence of spatial interaction (or independent components).  The transformation from $K^{ij}_{inhom}\left(r\right)$ to $L^{ij}_{inhom}\left(r\right)$  approximately stabilizes the variance of the estimate \citep{baddeley2015spatial}. \\

The empty-space function  $F_{j}\left(r\right)$  for type $j$ points  can be defined as the cumulative distribution function of the distance from an arbitrary location, say, $\mathbf{v}$, to the nearest type $j$ point.  It  represents the probability of finding a type $j$ point within a distance $r \ge 0$ of an arbitrary point $\mathbf{v}$ \citep{van1999indices}. Let $B\left(\mathbf{0}, r \right)$ denote a closed ball with radius $r$ centered at the origin $\mathbf{0}$ and  $\mathbfcal{Z}$ be an arbitrary spatial point process. One minus the empty-space function  can be given by 
\begin{align}\label{Mo'aa2022D}
1-F\left(r\right)   = 	\P\left(\mathbfcal{Z}\cap B\left(\mathbf{0}, r \right) = \emptyset\right)  =  \E\left[\displaystyle\prod_{\mathbf{z}\in \mathbfcal{Z}}1_{\displaystyle\left\lbrace \mathbf{z}\;\notin\; B\left(\mathbf{0}, r \right)\right\rbrace } \right], 
\end{align} 
which is the probability that there is no point of the point process $\mathbfcal{Z}$ within the ball $B\left(\mathbf{0}, r \right)$.\\

The cross-type nearest-neighbour function $G_{ij}$ represents the cumulative distribution function of the distance from a type $i$ point to the nearest type $j$ point.   It can be understood as the probability that the distance from an arbitrary type  $i$ point to the nearest type $j$  point is at most $r$ \citep{van1999indices}. The cross-type $G_{ij}$ measures the association between types (or same-species wild animals) $i$ and $j$.  For an arbitrary spatial point process $\mathbfcal{Z}$, the nearest-neighbour distance distribution function $G\left(r\right)$ can also be expressed as follows.
\begin{align}
1-G\left(r\right)   &= 	\P\left(\mathbfcal{Z}\cap B\left(\mathbf{0}, r\right) \backslash \left\lbrace \mathbf{0}\right\rbrace = \emptyset\;\big|\;\mathbf{0}\in \mathbfcal{Z}\right),\label{Mo'aa2023A}\\&  =  \E\left[\displaystyle\prod_{\mathbf{z}\in \mathbfcal{Z}\backslash\left\lbrace \mathbf{0}\right\rbrace }1_{\displaystyle\left\lbrace \mathbf{z}\;\notin \;B\left(\mathbf{0}, r \right)\;\big|\;\mathbf{0}\;\in\; \mathbfcal{Z}\right\rbrace }\right], \label{Mo'aa2023B}
\end{align} 
which represents the conditional probability that there is no additional point within the ball $B\left(\mathbf{0}, r \right)$,  given that $\mathbfcal{Z}$ already has a point at the origin $\mathbf{0}$. The probability-generating functional can be used to re-express the summary functions, see the details in \cite{daley2007introduction}. A probability-generating functional $\mathscr{D}$  for a spatial point process $\mathbfcal{Z}$ on $\W$  can be defined by
\begin{align}
\mathscr{D}_{\mathbfcal{Z}}\left(u\left(\boldsymbol\xi\right) \right)  = 	\E\left[\displaystyle\prod_{\boldsymbol\xi\in \mathbfcal{Z}}u\left(\boldsymbol\xi\right)\right],
\end{align} 
where $u: \W \rightarrow \left[0, 1\right]$ is a bounded nonnegative measurable function on the space $\W$  such that $ 0 \le u\left(\boldsymbol\xi\right) \le 1$ for any point $\boldsymbol\xi\in \W$. Let $u\left(\mathbf{z} \right)  = 1_{\left\lbrace \mathbf{z}\;\in\; B\left(\mathbf{0}, r\right)\right\rbrace }$.  Based on Equation \eqref{Mo'aa2022D}, it follows that $1-F\left(r\right)$ is equal to $\mathscr{D}_{\mathbfcal{Z}}\left(1 - u\left(\mathbf{z}\right) \right)$. Let $\P^{!\mathbf{0}}$ denote the reduced Palm distribution, see the details in \cite{palm1943intensitatsschwankungen}. Then $\P^{!\mathbf{0}}$, intuitively speaking, is the conditional probability that an event, say $A$, will occur given that $\mathbf{0}\in \mathbfcal{Z}$, i.e.,  given that there is a point of the process $ \mathbfcal{Z}$ at the origin $\mathbf{0}$. Using the reduced palm distribution and based on Equation \eqref{Mo'aa2023A}, we have that $1-G\left(r\right) =  \P^{\mathbf{!0}}$. Moreover, based on Equation \eqref{Mo'aa2023B}, the generating functional  corresponding to the reduced Palm distribution $\P^{\mathbf{!0}}$ is $\mathscr{D}^{\mathbf{!0}}$. Using these notations, we have that $1-G\left(r\right)  =  \mathscr{D}^{\mathbf{!0}}_{\mathbfcal{Z}}\left(1 - u\left(\mathbf{z}\right)\right)$. If $1-F\left(r\right) \neq 0$ , the $J$-function for the homogeneous point process can be given by 
\begin{align}
J\left(r\right)  = 	\frac{1-G\left(r\right)}{1-F\left(r\right)}.
\end{align} 
Spatial point patterns can be characterized as clustered, regular, or random.  In clustered point patterns, there are small distances between neighbouring points and large empty spaces between clusters. Consequently, $1-G\left(r\right)$  is smaller than $1-F\left(r\right)$ in clustered point patterns. That is, $J$ takes values smaller than 1  for clustered patterns. On the other hand,  $1-G\left(r\right)$  is larger than $1-F\left(r\right)$ in regular point patterns. Thus, $J$ takes values greater than 1 for regular point patterns. For Poisson point processes,  $\mathscr{D}^{\mathbf{!0}}$ is the same as $\mathscr{D}$, and that leads to $J = 1$ for $r\ge 0$.\\

\cite{van2011aj} extended the summary functions ($F$, $G$, and $J$) to an inhomogeneous point process. Using the probability-generating functionals, 
\begin{align}
1-F_{inhom}\left(r\right) =    \mathscr{D}_{\mathbfcal{Z}}\left(1 - \tilde{u}_{r}\left(\mathbf{z}\right) \right) \;\;\text{and}\;\;  1-G_{inhom}\left(r\right) = \mathscr{D}^{\mathbf{!0}}_{\mathbfcal{Z}}\left(1 - \tilde{u}_{r}\left(\mathbf{z}\right)\right),
\end{align} 
where 
\begin{align}\label{Infimum}
\tilde{u}_{r}\left(\mathbf{z}\right) = \tilde{\lambda}  \frac{1_{\displaystyle\left\lbrace \mathbf{z}\;\in\; B\left(\mathbf{0}, r \right)\right\rbrace }}{\lambda\left( \mathbf{z}\right)},\;\; 0 <  \tilde{\lambda}  =  \displaystyle\inf_{\mathbf{z}}\left\lbrace \lambda\left( \mathbf{z}\right)\right\rbrace,
\end{align} 
and $\lambda\left( \mathbf{z}\right)$ is the spatial intensity that is used to take the inhomogeneity of observing points of  $\mathbfcal{Z}$ over $\W$ into account. It is worth noting that in Equation \eqref{Infimum}, the infimum is used to ensure that the function $\tilde{u}_{r}\left(\mathbf{z}\right)$ takes values in $\left[ 0, 1\right]$, as required by the definition of generating functionals \citep{daley2007introduction}.  It follows that 
\begin{align}
J_{inhom}\left(r\right)  = 	\frac{1-G_{inhom}\left(r\right)}{1-F_{inhom}\left(r\right)}, \;\; F_{inhom}\left(r\right) \ne 1.
\end{align} 
The values of $J_{inhom}\left(r\right)$ are statistically interpreted in the same way as in the homogeneous point process case. That is, $J_{inhom}\left(r\right) = 1$ indicates the absence of interaction, and values greater than one indicate regular patterns, while values less than one indicate clustering. The details of the summary functions can be seen in \cite{van1999indices}, \cite{van2011aj}, and  \cite{baddeley2015spatial}.\\

In studying intra-specific interactions, such as the interaction between pairs of same-species wild animals, we employ cross versions of summary statistics that describe associations between points of different types. \cite{van2011aj} introduced an inhomogeneous cross-type $J$-function, which \cite{cronie2016summary} thoroughly investigated. Let $\mathbfcal{Z}_{i}$ be the marginal point process that consists of the points in  $\mathbfcal{Z}$ having label $i$. Then
\begin{align}\label{Ginhomcross}
1-G^{ij}_{inhom}\left(r\right) = \E\left[\displaystyle\prod_{\mathbf{z}\in\mathbfcal{Z}_{j}} \left(1 - \tilde{u}_{j, r}\left(\mathbf{z}\right)\right)\;\big|\;\mathbf{0}\;\in\; \mathbfcal{Z}_{i}\right],
\end{align} 
where 
\begin{align}\label{Infimumm}
\tilde{u}_{j, r}\left(\mathbf{z}\right) = \tilde{\lambda}_{j}  \frac{1_{\displaystyle\left\lbrace \mathbf{z}\;\in\; B\left(\mathbf{0}, r \right)\right\rbrace }}{\lambda_{j}\left( \mathbf{z}\right)},\;\; 0 < \tilde{\lambda}_{j} = \inf_{\mathbf{z}}\left\lbrace \lambda_{j}\left( \mathbf{z}\right)  \right\rbrace.
\end{align} 
Using Equation \eqref{Ginhomcross} and  under appropriate technical conditions \citep{cronie2016summary}, the inhomogeneous cross-type $J$-function for the inhomogeneous multitype  point process $ \mathbfcal{Z}$ can be given by 
\begin{align}
J^{ij}_{inhom}\left(r\right)  = 	\frac{1-G^{ij}_{inhom}\left(r\right)}{1-F^{j}_{inhom}\left(r\right)}, 
\end{align} 
where $F^{j}_{inhom}$ is the empty space function for the inhomogeneous point process $\mathbfcal{Z}_{j}$,  $F^{j}_{inhom}\left(r \right)  \ne 1$,  $G^{ij}_{inhom}$ is the nearest-neighbour distance distribution function from a point of type $i$ to the nearest one of type $j$, and $J^{ij}_{inhom}\left(r\right)$ compares the distribution of distances from a point of type $i$ to the nearest one of type $j$ to those from an arbitrarily chosen point to the nearest one of type $j$ of the point process $\mathbfcal{Z}_{j}$, see the details in \cite{cronie2016summary}.   The inhomogeneous  cross-type	$J^{ij}_{inhom}\left(r\right)$ can be used  to assess the independence of point patterns of types $i$  and $j$. Note that  $J^{ij}_{inhom}\left(r\right)  = 1$ when the wild animals $i$ and $j$ are independent. Values  of $J^{ij}_{inhom}\left(r\right)$ less than 1 suggest that the wild animal $j$ tends to cluster around those of wild animal $i$, while values of $J^{ij}_{inhom}\left(r\right)$ greater than 1 indicate that the wild animal $j$ tends to avoid the wild animal $i$ \citep{de2022testing}.
\subsection{Modelling the intensity function of a marked point process}
The spatially varying intensity functions must be estimated in order to compute the inhomogeneous cross-type summary functions. To accomplish this, we model the spatially varying intensity function parametrically.  Let $\mathbf{Z}$  denote a realization of the marked spatial point process  $\mathbfcal{Z} = \{(\mathbf{z}_{1}, m_{1}),\cdots, (\mathbf{z}_{n},m_{n})\}\subseteq \W\times \mathcal{M}$, where $\mathbf{z}_{i} \in \W$ is a spatial location, $m_{i} \in \mathcal{M}$ is the corresponding mark (or type), and $i = 1, 2, \cdots, n$.  Let  $(\mathbf{z}, m) \in \mathbf{Z}$ and $\lambda\left( (\mathbf{z}, m), \mathbf{Z}\right)$ denote the conditional intensity of a marked point process. We interpret $\lambda\left((\mathbf{z}, m), \mathbf{Z}\right)d\mathbf{z}$ as the conditional probability of finding a point of type $m$ in an infinitesimal neighbourhood of the point $\mathbf{z}$, given that the rest of the point process coincides with $\mathbf{Z}$ \citep{baddeley2015spatial}. One can model the logarithm of the conditional intensity function $\lambda\left( (\mathbf{z}, m), \mathbf{Z}\right)$ at location $\mathbf{z} = \left(x, y\right) $ and mark  $m$ as follows.
\begin{align}\label{LinearModel}
\log\left(\lambda\left( (\mathbf{z}, m), \mathbf{Z}\right)\right) = \alpha_{m} + \beta_{m}x +  \gamma_{m}y.
\end{align} 
In this model, we consider marks as a factor, so that the trend (or model) has a separate constant value for each level of marks. Furthermore, the intensity is modelled in a loglinear form in the location with different intercepts for each wild animal species. The model can be estimated using the method of maximum likelihood via the approach of  \cite{berman1992approximating}, which  \cite{baddeley2000practical} extended to multitype point patterns. \\

Edge correction issues may arise when the process $\mathbfcal{Z}$ is unbounded and the data $\mathbf{Z}$ is a partially observed realization of $\mathbfcal{Z}$. Some points from $\mathbfcal{Z}$ may fall on the edge or outside $\W \times \mathcal{M}$. As a result, the conditional intensity $\lambda\left( (\mathbf{z}, m), \mathbf{Z}\right)$ of $\mathbfcal{Z}$ may not be fully captured, leading to systematic error in parameter estimation. To address edge effects, we use the border method or reduced sample estimator, which is fast to compute and applicable to windows of arbitrary shape \citep{Ripley1988BD}. 

\subsection{Monte Carlo envelope tests}\label{MCET2080}
Monte Carlo envelope tests are commonly employed to assess whether estimated summary functions show significant deviations from their values under the null hypothesis \citep{besag1977simple, ripley1979tests, ripley1981spatial, marriott1979barnard}. In our case, the summary functions are the inhomogeneous cross-types  $L$-and $J$-functions.  The simulations for the Monte Carlo tests are generated from the null hypothesis of independence of components. The simulations are obtained by shifting the pattern and intensity function of one type (or mark) with respect to the other, which leaves the marginal structures unchanged and affects only the interactions between types (or marks). To account for the possibility of points being moved outside the plot due to this shift, a torus correction can be implemented \citep{baddeley2015spatial, cronie2016summary, de2022testing}.  The null hypothesis explicitly states that spatial point patterns of different types are realizations of independent point processes. The procedure to test the null hypothesis is as follows. Let  $T_{obs}\left(r\right)$ be the summary function estimator for the observed point pattern at  spatial distance $r$. Let $S$ be the number of spatial point pattern simulations generated from the null hypothesis, and the estimators $T_{k}\left(r\right), \; k = 2, 3, \cdots, S+1,$ be the summary function estimators for the simulated spatial point patterns.  This can be restated as   $T_{1}\left(r\right)$  =  $T_{obs}\left(r\right)$ , $T_{2}\left(r\right)$,  $T_{3}\left(r\right)$, $\cdots$, $T_{S+1}\left(r\right)$ represent the summary function estimators for the observed spatial point pattern (the first) and the simulated spatial point patterns (the last $S$ estimators).   It is clear that $T_{k}\left(r\right)$ and  $T_{obs}\left(r\right)$ are functions of $r$, and we must obtain the estimates in order to perform the test.  The values of $r$ are assumed to be within a given range, say $\left[0, r_{max}\right]$. The ideal $r_{max}$ is determined by the range of spatial interaction. However, we usually choose $r_{max}$ based on the size of the study region, the summary function at hand, and problem-specific knowledge \citep{de2022testing}.
\subsubsection{Pointwise envelope test}
Let $r$ be a fixed distance in $\left[0, r_{max}\right]$. At the specific distance $r$ and under the null hypothesis,  the random variables $T_{obs}\left(r\right)$ , $T_{2}\left(r\right)$,  $T_{3}\left(r\right)$, $\cdots$, $T_{S+1}\left(r\right)$ are \emph{iid}.  If $T_{obs}\left(r\right)$  is randomly chosen from $\left\lbrace T_{obs}\left(r\right) , T_{2}\left(r\right),  T_{3}\left(r\right),  \cdots,  T_{S+1}\left(r\right)  \right\rbrace$,  the probability that it is the largest  or the $k$-$th$ largest for any $k\in\left\lbrace1, 2, \cdots, S+1\right\rbrace$ is $\frac{1}{S+1}.$  The two-sided Monte Carlo test with significance level $\alpha = \frac{2}{S+1}$ rejects the null hypothesis if the estimate of $T_{obs}\left(r\right)$  lies outside the range of the estimates of  $T_{k}\left(r\right), k = 2, 3, \cdots, S+1.$ That is to say,  we reject the null hypothesis when the estimate of $T_{obs}\left(r\right)$ lies outside the interval given by 
\begin{align}\label{Envelop2080}
\left( L_{env} \left( r\right) = \displaystyle\min_{k\in\left\lbrace 2, \cdots, S+1\right\rbrace}\left\lbrace T_{k}\left(r\right)\right\rbrace , \;\; U_{env}\left(r\right) =  \displaystyle\max_{k\in\left\lbrace 2, \cdots, S+1\right\rbrace}\left\lbrace T_{k}\left(r\right)\right\rbrace\right),
\end{align} 
with a significance level of $\alpha = \frac{2}{S+1}$. It can be seen in Equation \eqref{Envelop2080} that the test is a pointwise envelope  as $L_{env} \left( r\right) $ and $ U_{env}\left(r\right) $ depend on $r$. The pointwise envelope is typically plotted over the entire interval $\left[0, r_{max}\right]$ of $r$ values. It is used as a diagnostic tool to identify the ranges where the data deviates from the assumed model \citep{besag1977simple, ripley1977modelling}.\\

We would like to emphasize that pointwise envelopes are not confidence bands for the true value of the function.  Instead, they define the critical points based on a Monte Carlo simulation in a pointwise fashion \citep{ripley1981spatial}. The pointwise envelope test is constructed by choosing a fixed value of $r$ and rejecting the null hypothesis if the observed function value lies outside the envelope at this value of  $r$.
\subsubsection{Global envelope test}
When the plot of the observed data does not fall entirely within the pointwise envelopes, the null hypothesis cannot be simply rejected \citep{loosmore2006statistical}. The reason for this is that considering all r at the same time causes a multiple testing problem \citep{shaffer1995multiple}. A global envelope test with a properly chosen range of $r$ values can be used to address this multiple testing problem. It rejects the null hypothesis when 
\begin{align}\label{GEnvelop208}
G_{env}\left(T_{obs} \right)  = 
\begin{cases}
1, & \text{if}\;\; \exists \;r\in \;\left[0, r_{max}\right]\text{such that }\; T_{obs}\left(r\right) \;\notin\;\left(T_{low} \left(r \right),\; T_{up}\left( r\right) \right) \\
0, & \text{otherwise}
\end{cases}
\end{align} 
is equal to one. The lower and upper values $T_{low} \left(r \right)$ and  $T_{up}\left( r\right)$ in Equation \eqref{GEnvelop208} are determined in such a way that the  significance level of the hypothesis testing is adjusted for multiple testing.  Thus, the global envelope test  rejects the null hypothesis if the observed function $T_{obs}\left(\cdot\right)$ is not completely inside the envelope.\\

One of the global envelope tests is the \emph{maximum absolute deviation (MAD)} between the estimator of the summary function and the theoretical (or expected value) of the estimator under the null hypothesis \citep{diggle1979parameter}. The \emph{MAD} test is based on the measure given by 
\begin{align}\label{MGEnvelop208}
V_{k}  = \displaystyle\max_{r\;\in\;\left[0, r_{max}\right]\;}\abs[\bigg]{T_{k} \left(r \right)  - \E\left[T_{k} \left(r \right) \right]},
\end{align} 
where $k = 1, 2, \cdots,  S+1$, $T_{1}\left(r\right) =  T_{obs}\left(r\right)$, and  the expected value $\E\left[T\left(r \right) \right]$ is computed under the null hypothesis.  The test involves computing $V_{1}$ (equivalent to $V_{obs}$) for the observed point pattern and similarly calculating $V_{k}$ for each of the simulated spatial point patterns ($k = 2, 3, \cdots, S+1$) using Equation \eqref{MGEnvelop208}.  For two-sided hypothesis testing, we reject the null hypothesis if $V_{obs}$ is larger than $V_{max} = \max\left\lbrace V_{k}: \;k = 2, 3, \cdots,  S+1\right\rbrace$ with \emph{MAD} test significance level $\alpha = \frac{2}{S+1}$. The global envelope test based on Equation \eqref{GEnvelop208} rejects the null hypothesis at the significance level $\alpha$ when $T_{obs}\left( r\right)$  does not lie entirely within the global envelope given by
\begin{align*}
\left(  \E\left[T\left(r \right) \right] - V_{max},  \; \E\left[T\left(r \right) \right] + V_{max}\right).
\end{align*} 
Apart from \emph{MAD}-based hypothesis testing, another approach involves hypothesis testing using the \emph{integrated squared error}, denoted as $Q$. The \emph{integrated squared error} is calculated as the integral of the squared difference between the estimator $T(r)$ and its expected value $\E[T(r)]$ over the range $\left[0, r_{max}\right]$, expressed as:
\begin{align*}
Q  = \displaystyle\int_{\left[0, r_{max}\right]}\left( T\left(r \right)  - \E\left[T\left(r \right) \right]\right) ^{2}dr.
\end{align*} 
A Monte Carlo test based on $Q$ is known as the \emph{Diggle-Cressie-Loosmore-Ford (DCLF)} test \citep{loosmore2006statistical}. The power of the aforementioned tests can be maximized when the interval length of $\left[0,  r_{max}\right]$ is larger than the range of the spatial interaction.  The \emph{MAD} test  is insensitive to the choice of  $\left[0,  r_{max}\right]$.   Even though the \emph{DCLF} test is quite sensitive to the choice of  $\left[0,  r_{max}\right]$,   it is typically  more powerful than the \emph{MAD} test \citep{baddeley2015spatial}.   \cite{baddeley2014tests} recommend using the \emph{DCLF} test provided that the range of spatial interaction is known approximately.  If there is no information about the range of spatial interaction,  it is advisable to use the \emph{MAD} test, and in this case, we need to choose the interval length of  $\left[0,  r_{max}\right]$ to be as large as practicable \citep{baddeley2015spatial}.\\

The Lotwick-Silverman test  can also be used to determine whether the spatial point patterns of the wild animals are independent of each other \citep{lotwick1982methods}. To do so, we employ the  inhomogeneous cross-type $L$- and $J$-functions presented in this section. Remember that  $L_{inhom}^{ij}\left( r\right) = r$ when the same-species wild animals of  types $i$ and $j$ are independent. Its larger values suggest a positive association, while smaller values indicate a negative association between the same-species wild animals. Similarly, $J_{inhom}^{ij}\left( r\right)  = 1$ when the same-species wild animals of types $i$ and $j$ are independent. Its smaller values suggest that wild animals of type $j$ tend to cluster around wild animals of type $i$, while  values larger than 1 suggest that wild animals of type $j$ tend to avoid wild animals of type $i$.

\section{Application to black bear relocation data}\label{ApplicationToBlackBear}
This section introduces the study area and real data, outlines the simulation of spatial point patterns under the null hypothesis, and provides exploratory analysis along with the main study results.

\subsection{Study area and real data set}
We used relocation data collected from 12 GPS-collared black bears, 2015–2017, in Washington County (in the cities:  Chatom, Mount Vernon, Fruitdale, and Wagarville) and Mobile County of Alabama.  Methods used to collect data is described in \cite{leeper2021resource}. The median sampling interval for the GPS-based tracking of black bear relocation was one hour. The sampling duration of the black bear relocation tracking ranges from 197.041 to 414.334 days. Each bear was relocated between 1,955 and 8,043 times, resulting in 64,471 data points from the 12 bears.  Each relocation dataset includes the initial location of the black bear as well as the relocation of the black bear every hour on average during its movement. Figure \ref{Or213} presents the spatial relocation data for 12 black bears in two southern Alabama counties of the United States (left). It also depicts the initial (blue) and final (red) locations of a single black bear in Mount Vernon City of Mobile County (right).
\begin{figure}[h]
\centering
\includegraphics[width=0.4\textwidth, height=0.45\linewidth]{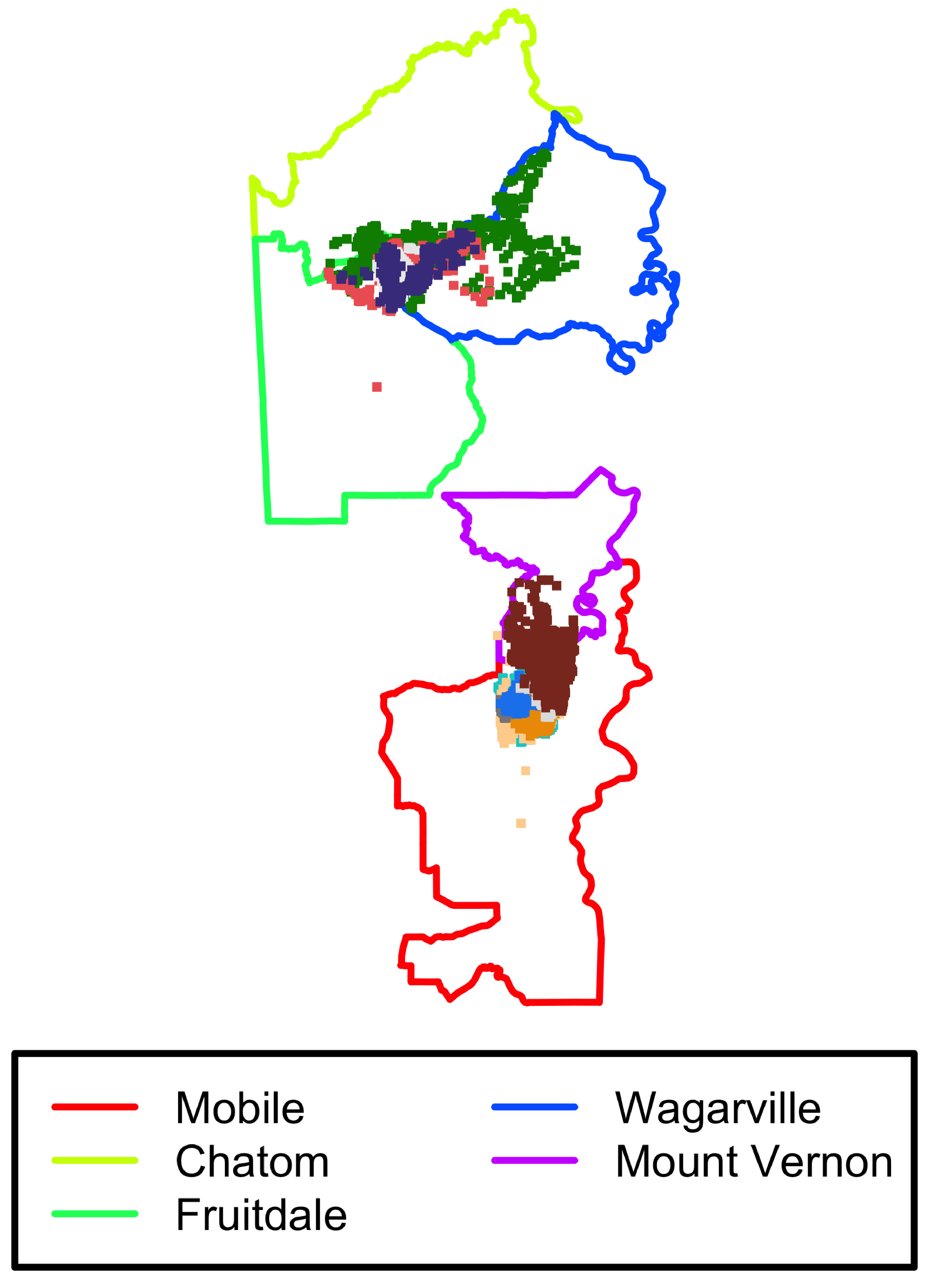}
\includegraphics[width=0.45\textwidth, height=0.4\linewidth]{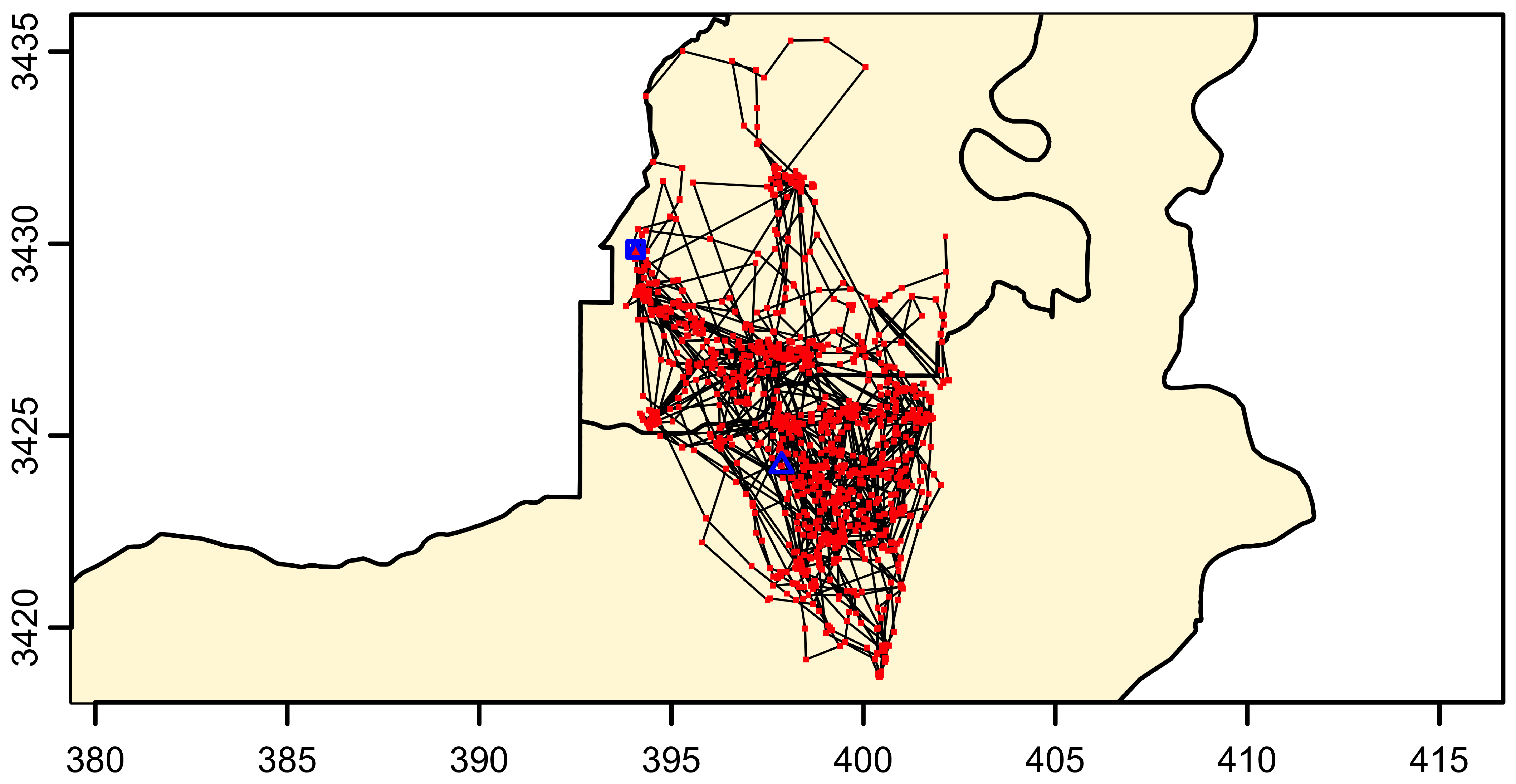}
\caption{Spatial relocation data for the twelve black bears in two southern counties of Alabama in the United States (left), as well as the initial location (triangular) and the final location (rectangular) of the movements of a black bear in Mobile County and Mount Vernon City (right).  The scale of the tick labels in the right-hand side relocation plot is in kilometers (km)}
\label{Or213}
\end{figure}
\subsection{Monte Carlo simulations}
Here, we briefly outline the simulation of spatial point patterns for the inferential methods discussed in Section \ref{MCET2080}. To test the independence of components (or the independence of the spatial point patterns) of pairs of black bears, we generated a reasonable number of spatial point patterns for each type by randomly shifting the points of their corresponding observed spatial point pattern. In other words, we divided the observed multitype spatial point pattern into subpatterns consisting of points belonging to each type, and then we randomly displaced the points within these subpatterns to generate the desired number of simulated point patterns for each type.  In the case of a stationary spatial point process, the shifted version of the observed spatial point pattern is statistically equivalent to the original. However, in the case of a nonstationary spatial point process, it is necessary to also adjust the intensity of the point process along with the shift of the point pattern. Under the null hypothesis of independence of components, we employed different shifts to each type of spatial point patterns, and shifting the pattern and intensity function of one black bear relative to the other preserves the marginal structures while exclusively affecting the inter-species interactions \citep{baddeley2015spatial}. 
\subsection{Exploratory data analysis}
The empirical semivariance plots in Figure \ref{Bearsve} depict average squared distances traveled by twelve black bears at different time lags. If these plots flatten out at larger lags, indicating range residency, then the black bear relocation data sets are suitable for home range analysis \citep{calabrese2016ctmm}. It can be seen from the semivariance plots that there are high variabilities in the relocations of black bears at larger lags. The black bears seem to leave the areas where they mostly spend time and then return after a while. However, we believe that further investigation into the sources of these variabilities is needed. \\

We also fitted the theoretical semivariance functions to the empirical semivariance estimates for subsequent statistical analyses.  The theoretical semivariance functions for \emph{iid},  \emph{OU}, and \emph{OUF} processes, as detailed in \cite{fleming2014fine},  are fitted to the empirical semivariance estimates using maximum likelihood estimation. The \emph{AIC} is employed to select the optimal model that offers a superior fit to the empirical semivariance estimates. Table \ref{table:SVM} shows the estimated theoretical semivariance functions for black bears 493, 498, and 505.  The table  presents $\Delta AIC$ values, which indicate the difference between the \emph{AIC} values of the candidate models and the minimum \emph{AIC}. Consequently, smaller $\Delta AIC$ values signify a superior model, with zero representing the optimal model.\\
\begin{table}[h]
\caption{The estimated theoretical semivariance models for black bears 493, 498, and 505}
\label{table:SVM}
\centering
\begin{tabular}{|c||l|c|}
\hline\hline
\multicolumn{1}{|c|}{Black bear} & \multicolumn{1}{|c|}{Model}& \multicolumn{1}{|c|}{$\Delta$AIC}\\\hline\hline
\multirow{3}{1em}{493} & \emph{OU} anisotropic  & 0.000 \\
& \emph{OUF} anisotropic & 2.000\\ 
& \emph{OU} & 74.070\\\hline
\multirow{3}{1em}{498} & \emph{OUF} anisotropic & 0.000\\   
&\emph{OU}  anisotropic    &93.626\\
& \emph{OUF} &  437.672\\\hline
\multirow{3}{1em}{505} & \emph{iid} anisotropic  &  0.000\\  
&\emph{OU}  anisotropic &    2.000 \\
& \emph{OUF} anisotropic &   4.000 \\\hline\hline
\end{tabular}
\end{table}

According to the \emph{AIC}, the semivariance analysis of the black bear relocation data set revealed three movement behaviours. The semivariance model for the  \emph{OUF}  process fits the empirical semivariances of black bears 498 and 501 well, while the semivariance model for the \emph{OU} process better describes the empirical semivariances of black bears  487, 488, 491, 493, 495, 503,  505 and 506. In contrast, the semivariance model for the \emph{iid} process better fits the movement behaviour of black bears  490 and  504. The best models selected for all the black bears are anisotropic models, which indicate that the black bears have a directional preference in their movements. This anisotropic property in black bear movements may be due to geographical factors, food availability, etc. In general, the empirical semivariance models for black bears 487, 488, 490, 491, 493, 495, 498,  501, and  506  tend to flatten out, and  the theoretical semivariance models also capture the movement behaviour of these bears across the lags, as shown in Figure \ref{Bearsve}. 
\begin{figure}[h]
\centering
\includegraphics[width=0.2\textwidth, height=0.15\linewidth]{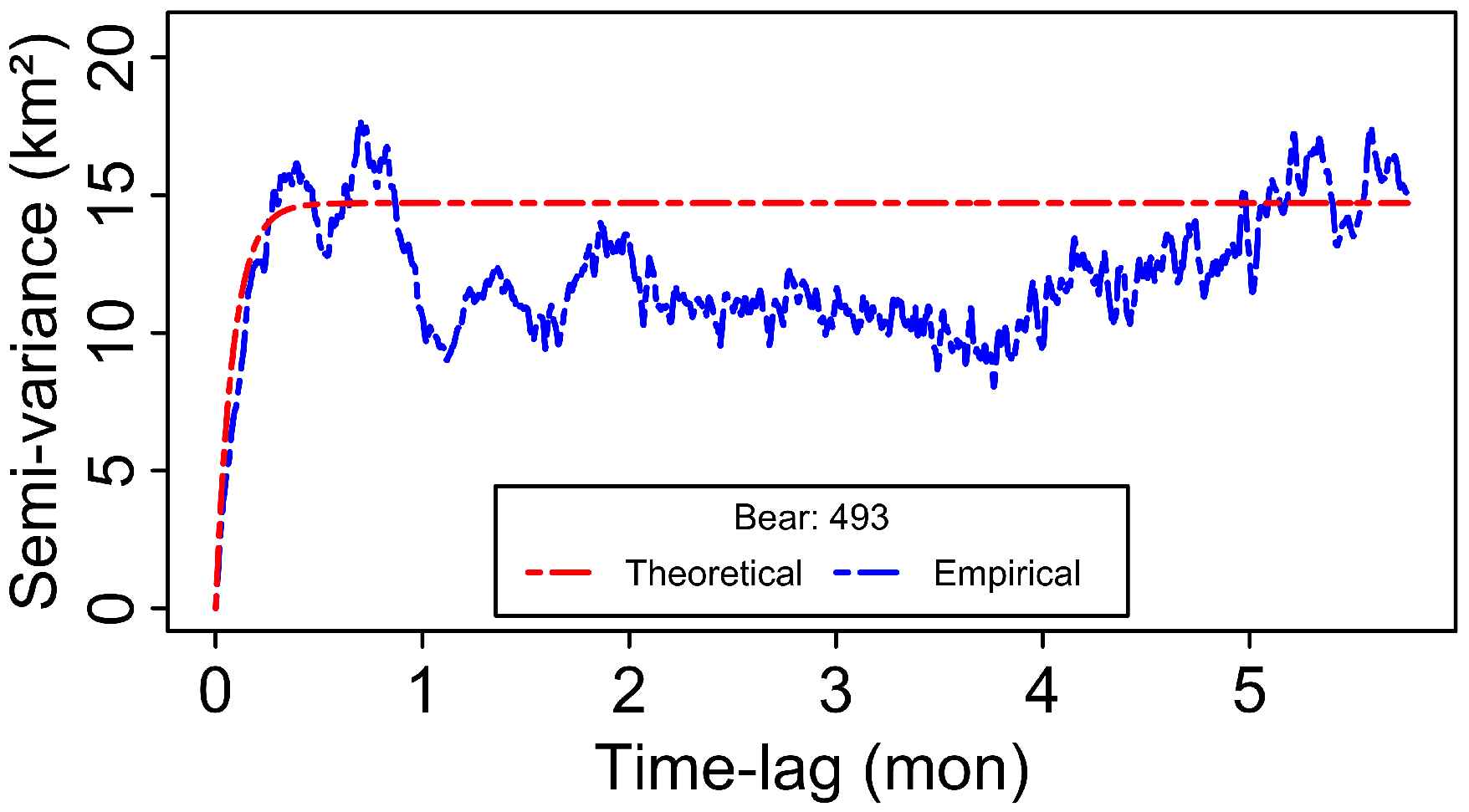}
\includegraphics[width=0.2\textwidth, height=0.15\linewidth]{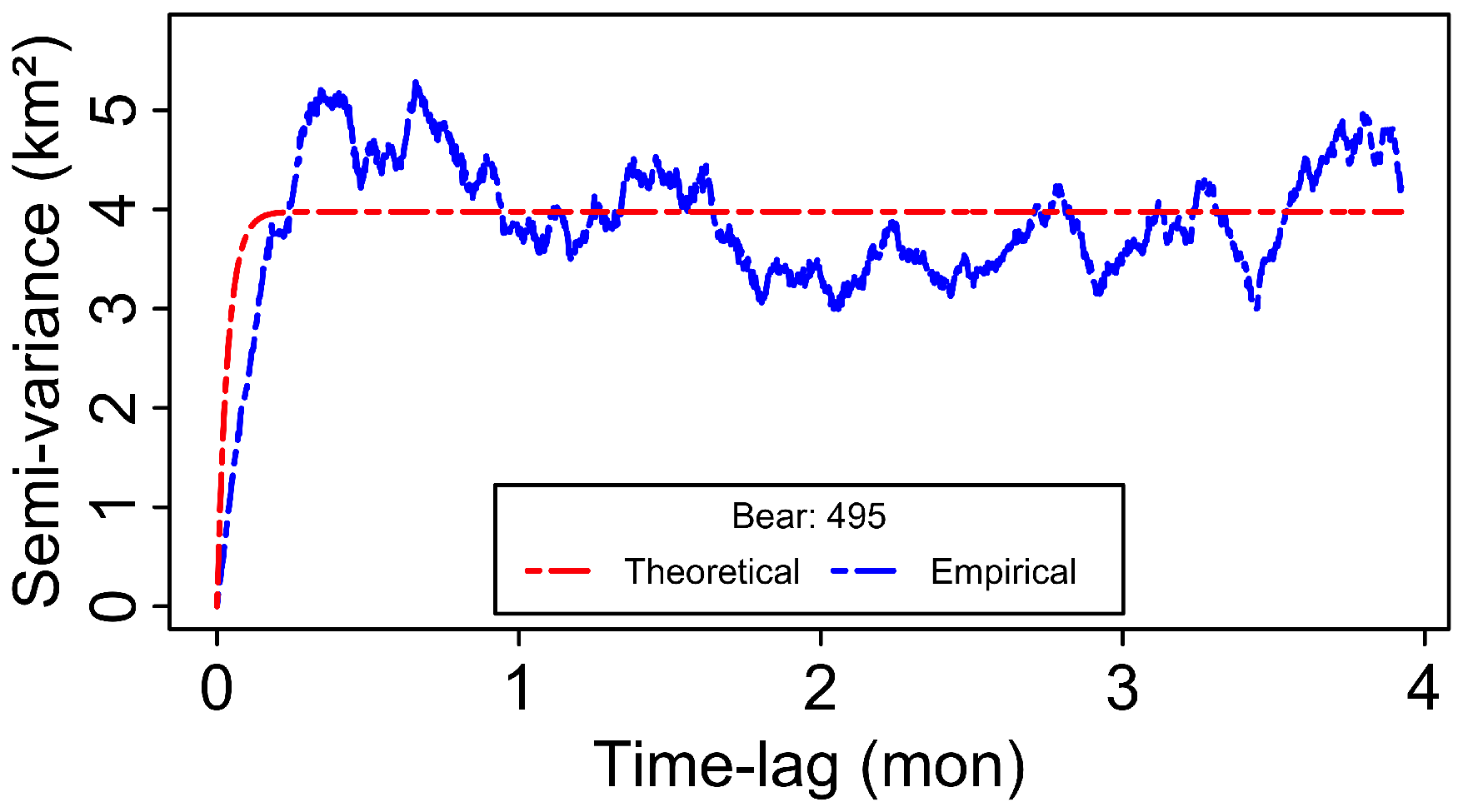}
\includegraphics[width=0.2\textwidth, height=0.15\linewidth]{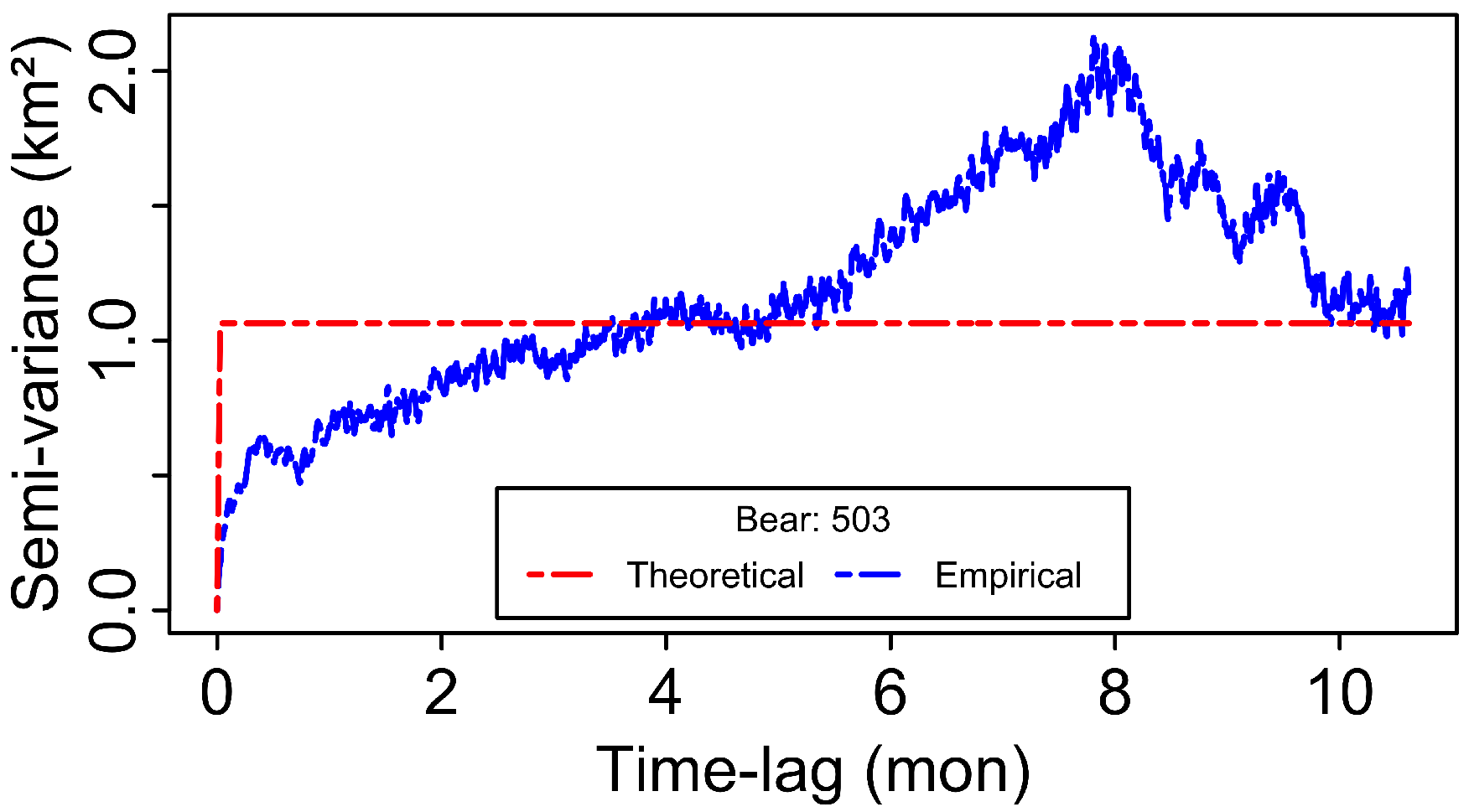}
\includegraphics[width=0.2\textwidth, height=0.15\linewidth]{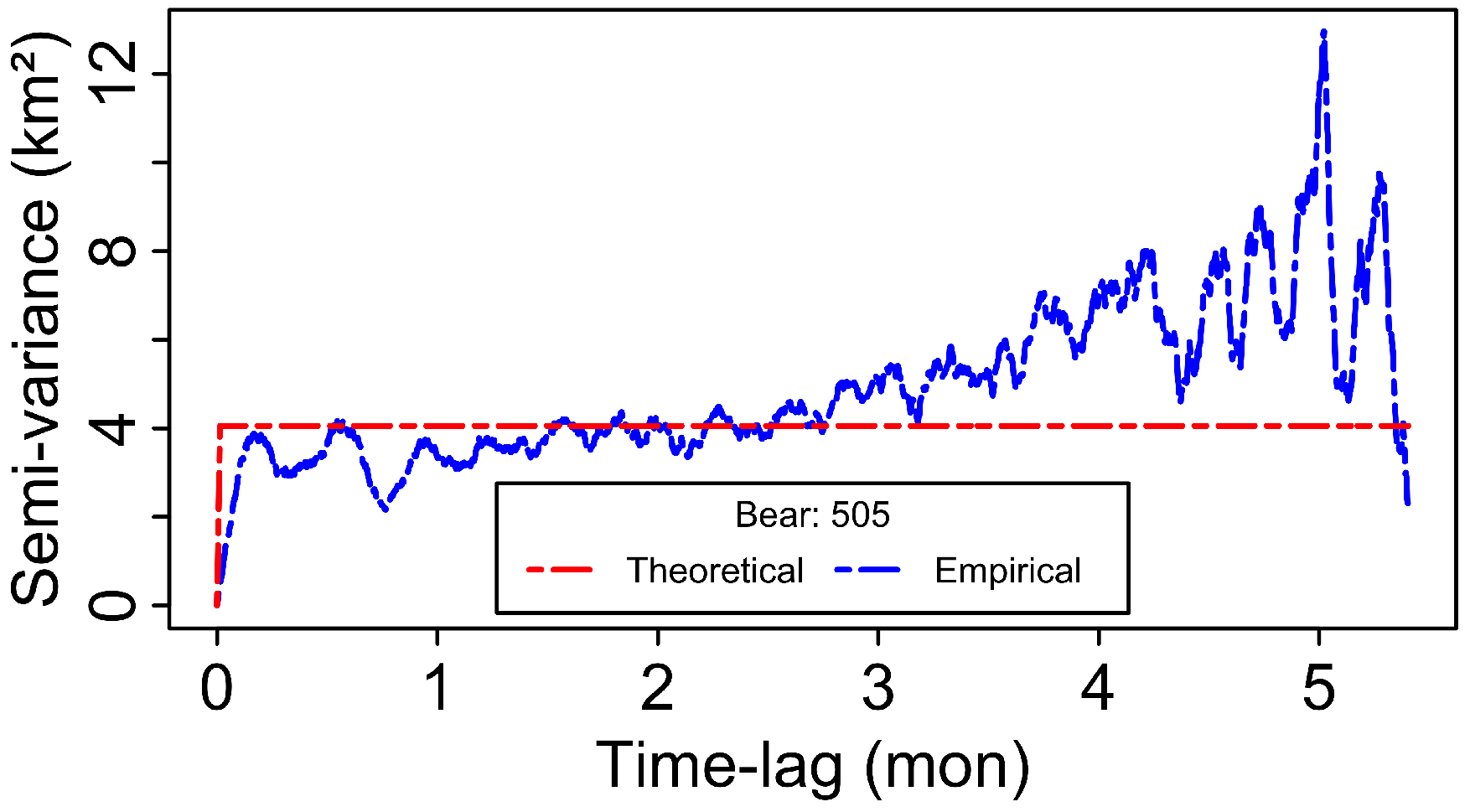}
\includegraphics[width=0.2\textwidth, height=0.15\linewidth]{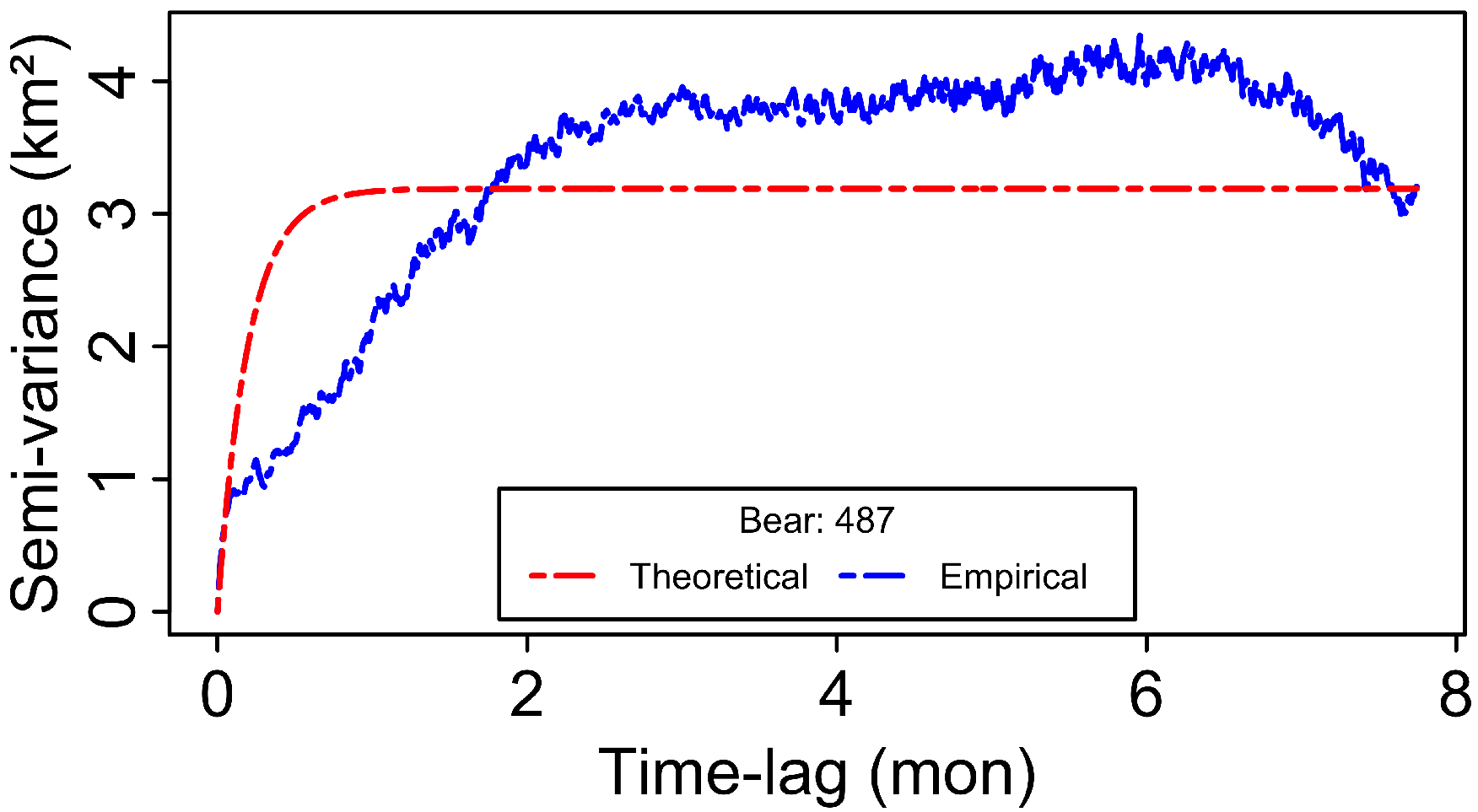}
\includegraphics[width=0.2\textwidth, height=0.15\linewidth]{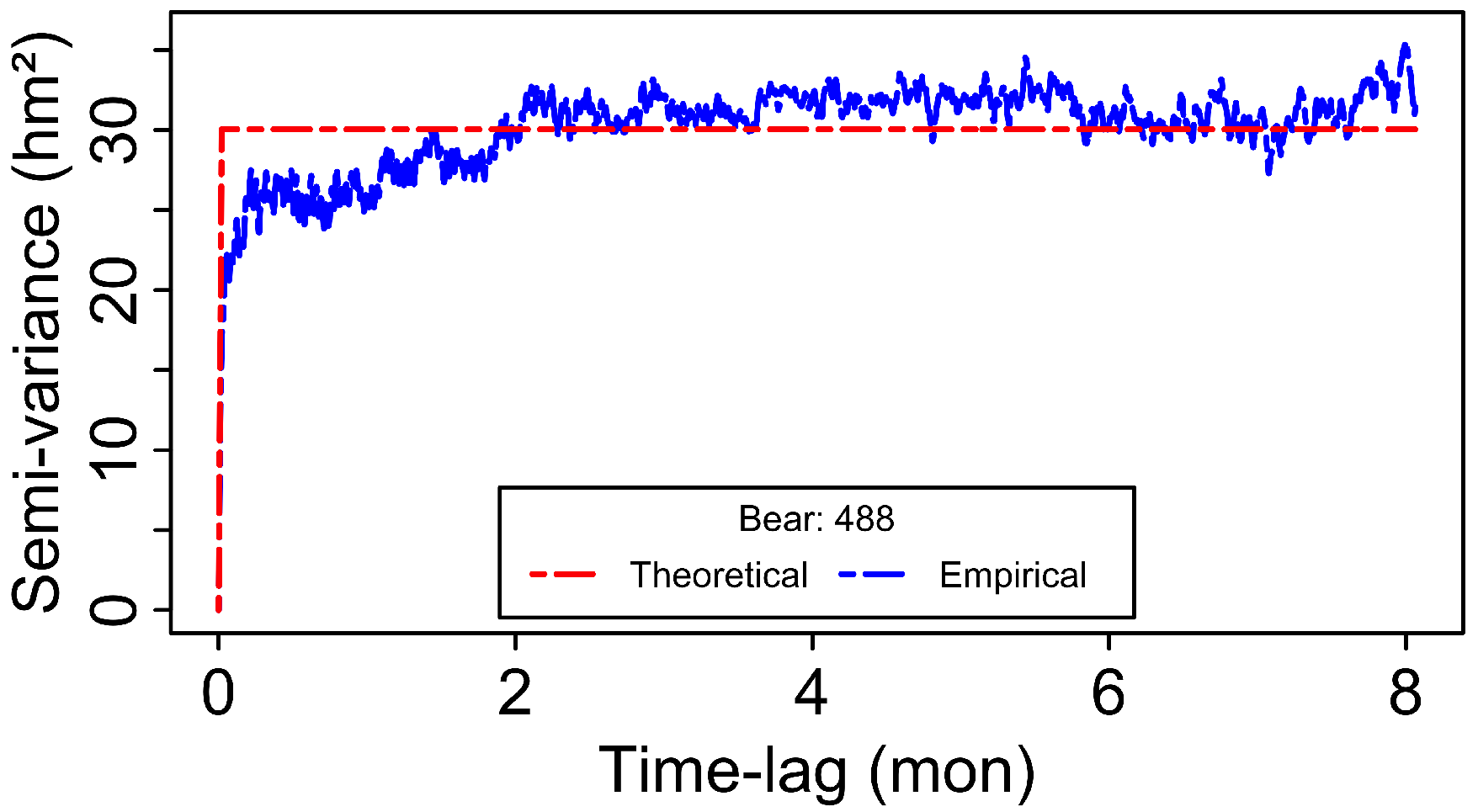}
\includegraphics[width=0.2\textwidth, height=0.15\linewidth]{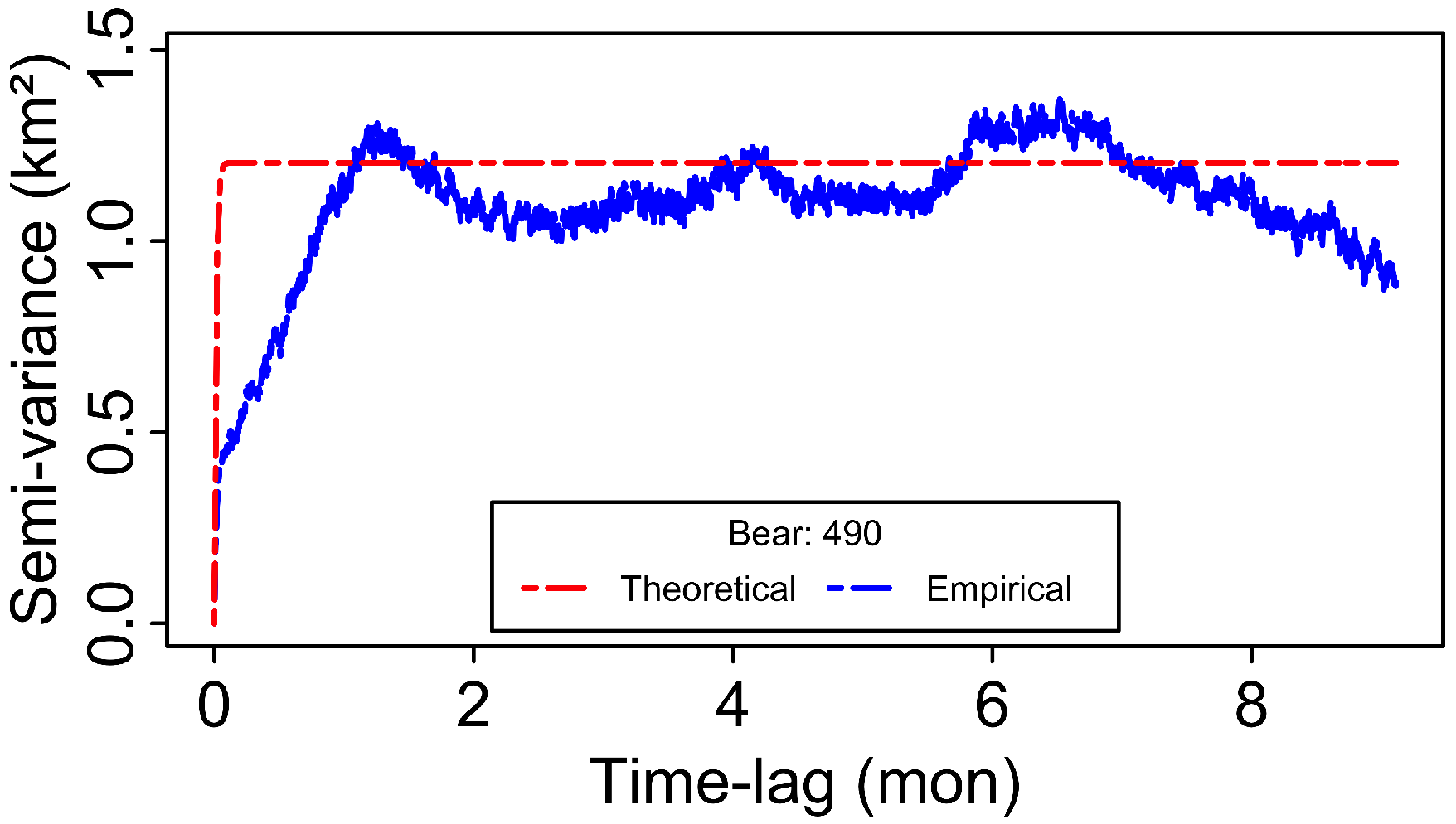}
\includegraphics[width=0.2\textwidth, height=0.15\linewidth]{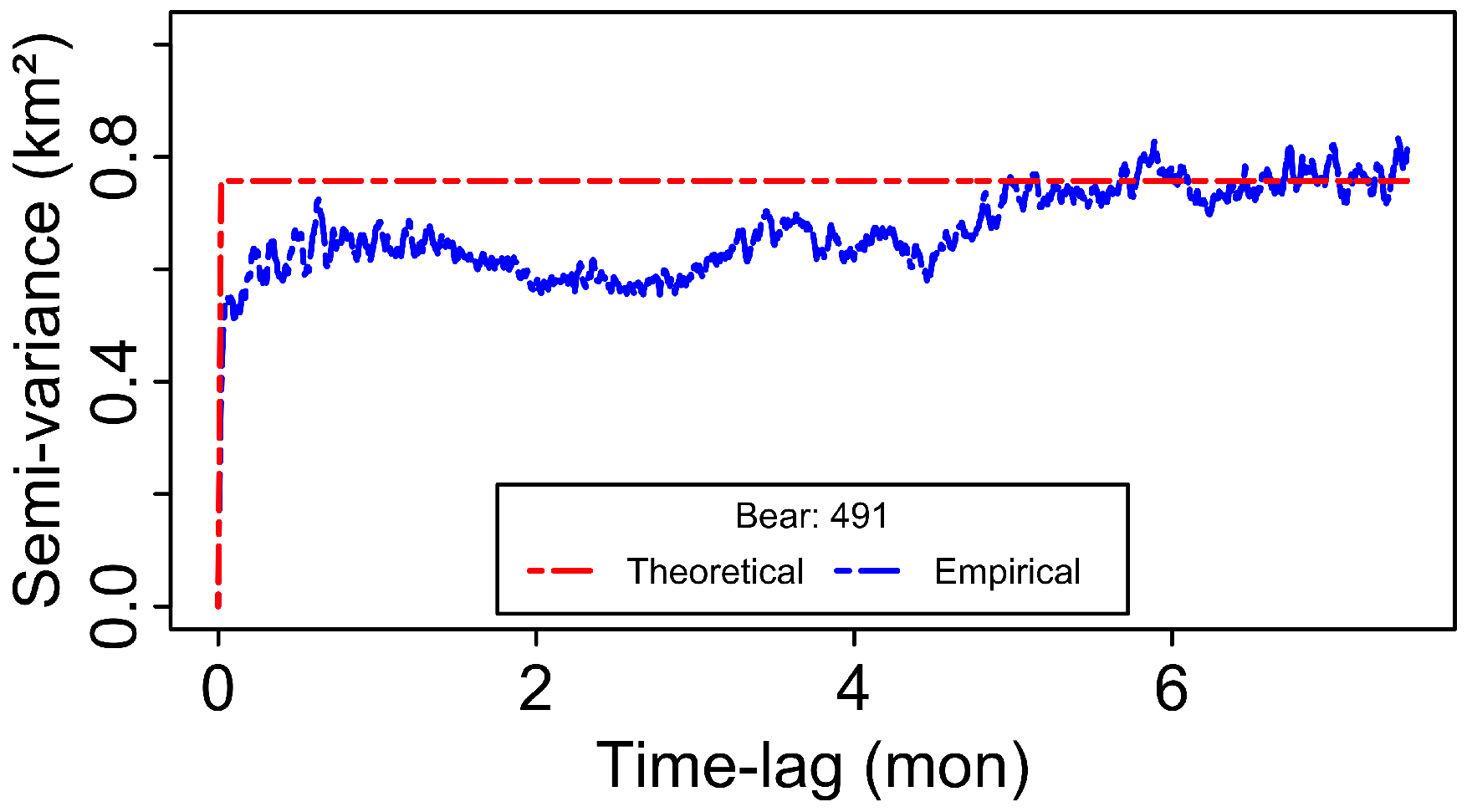}
\includegraphics[width=0.2\textwidth, height=0.15\linewidth]{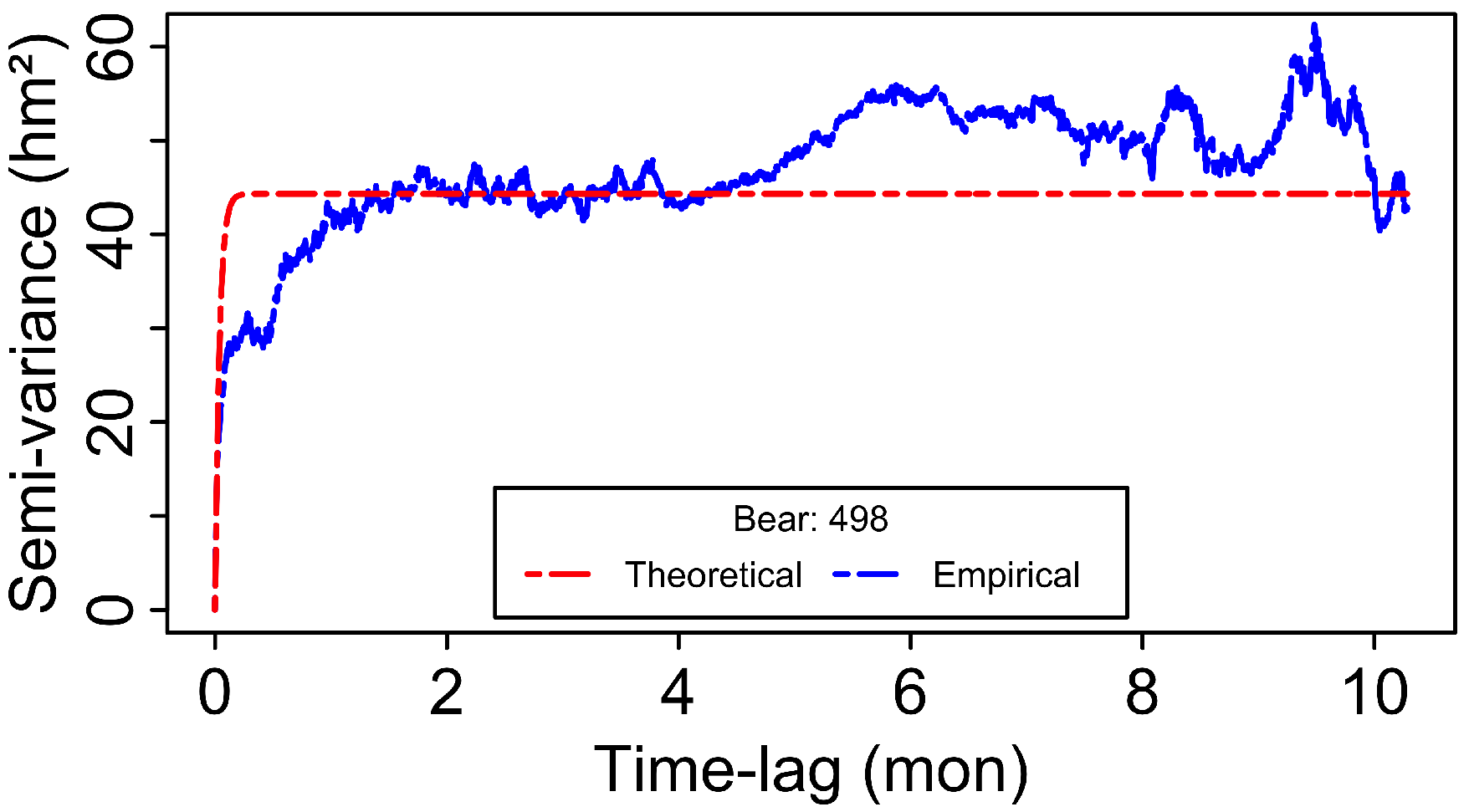}
\includegraphics[width=0.2\textwidth, height=0.15\linewidth]{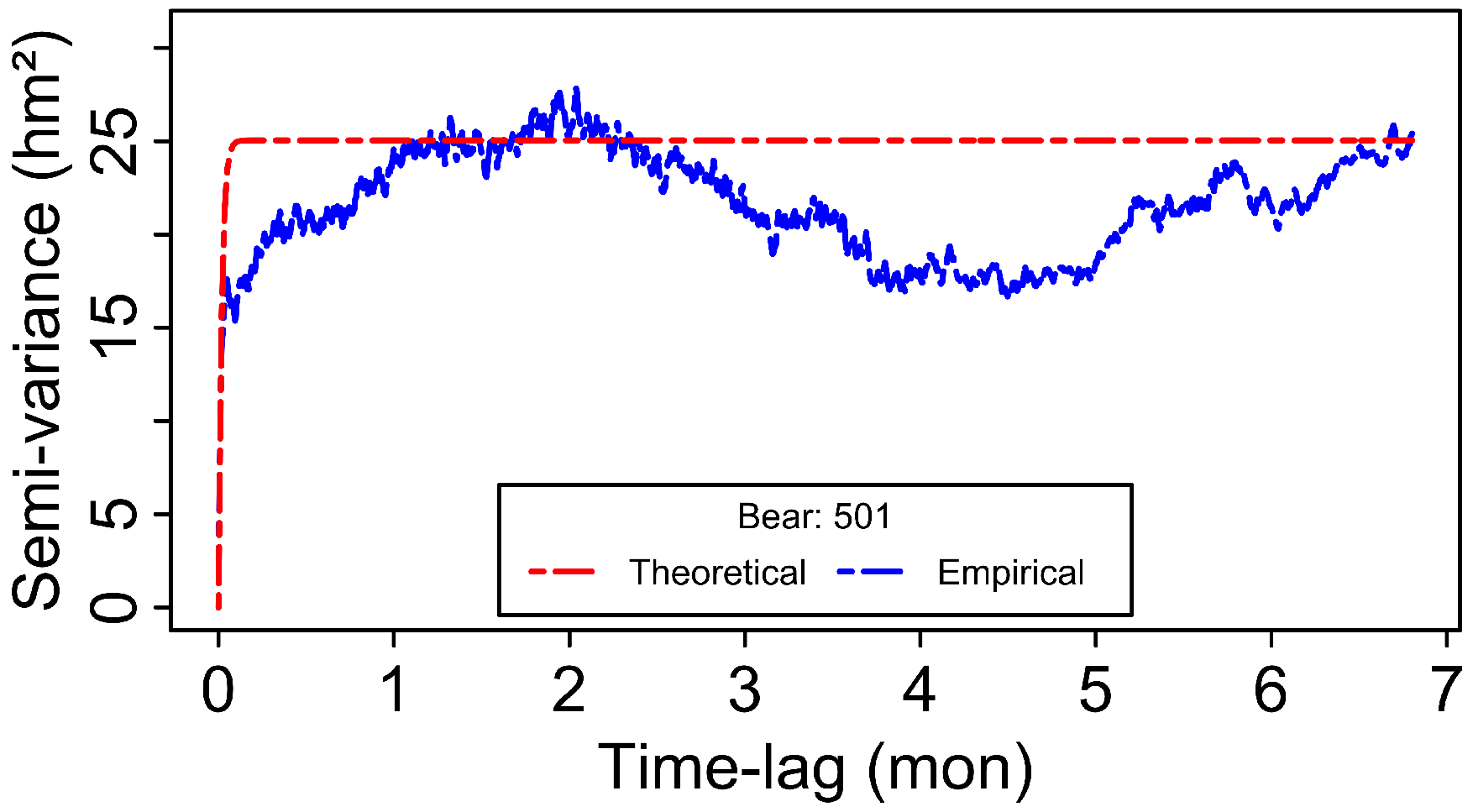}
\includegraphics[width=0.2\textwidth, height=0.15\linewidth]{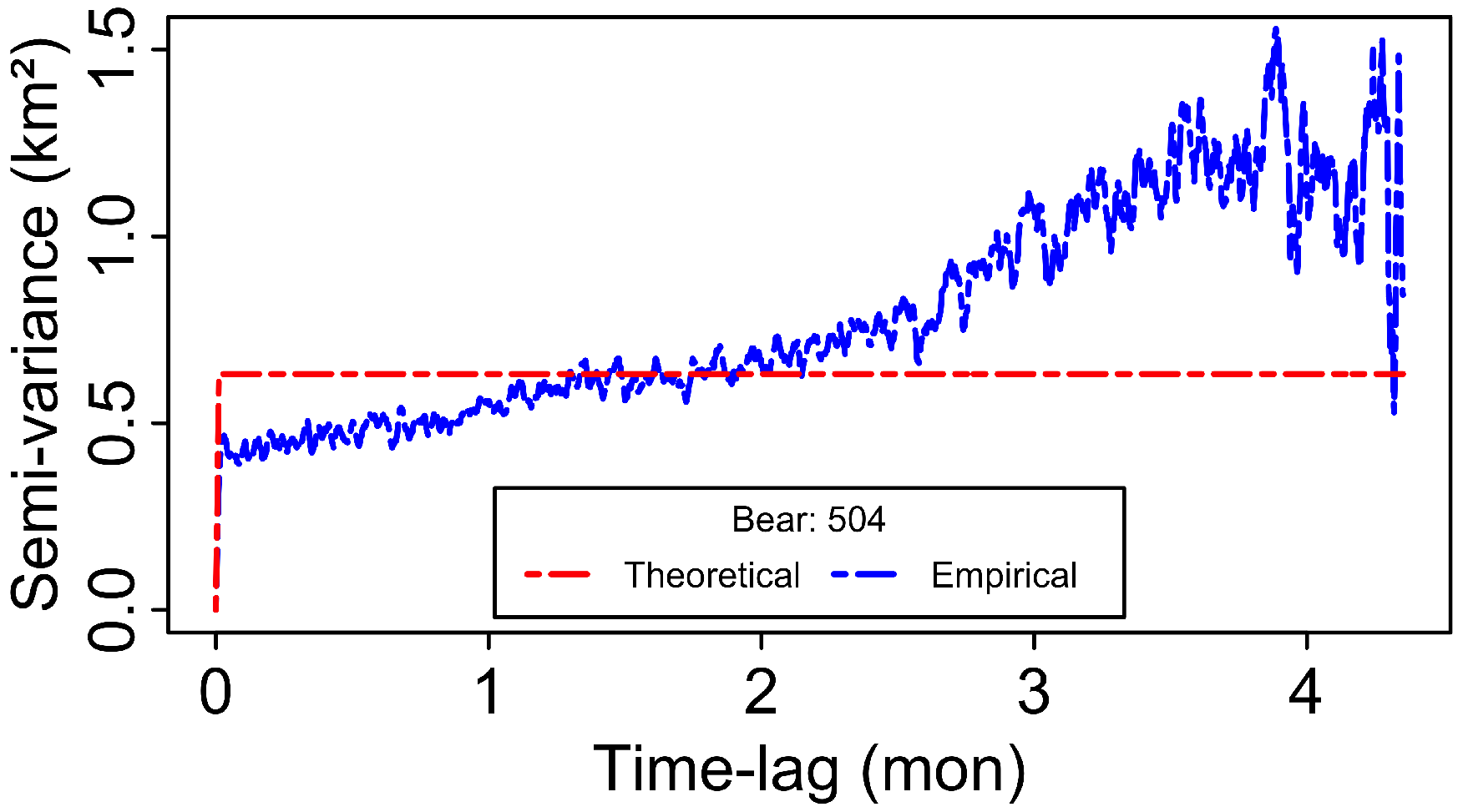}
\includegraphics[width=0.2\textwidth, height=0.15\linewidth]{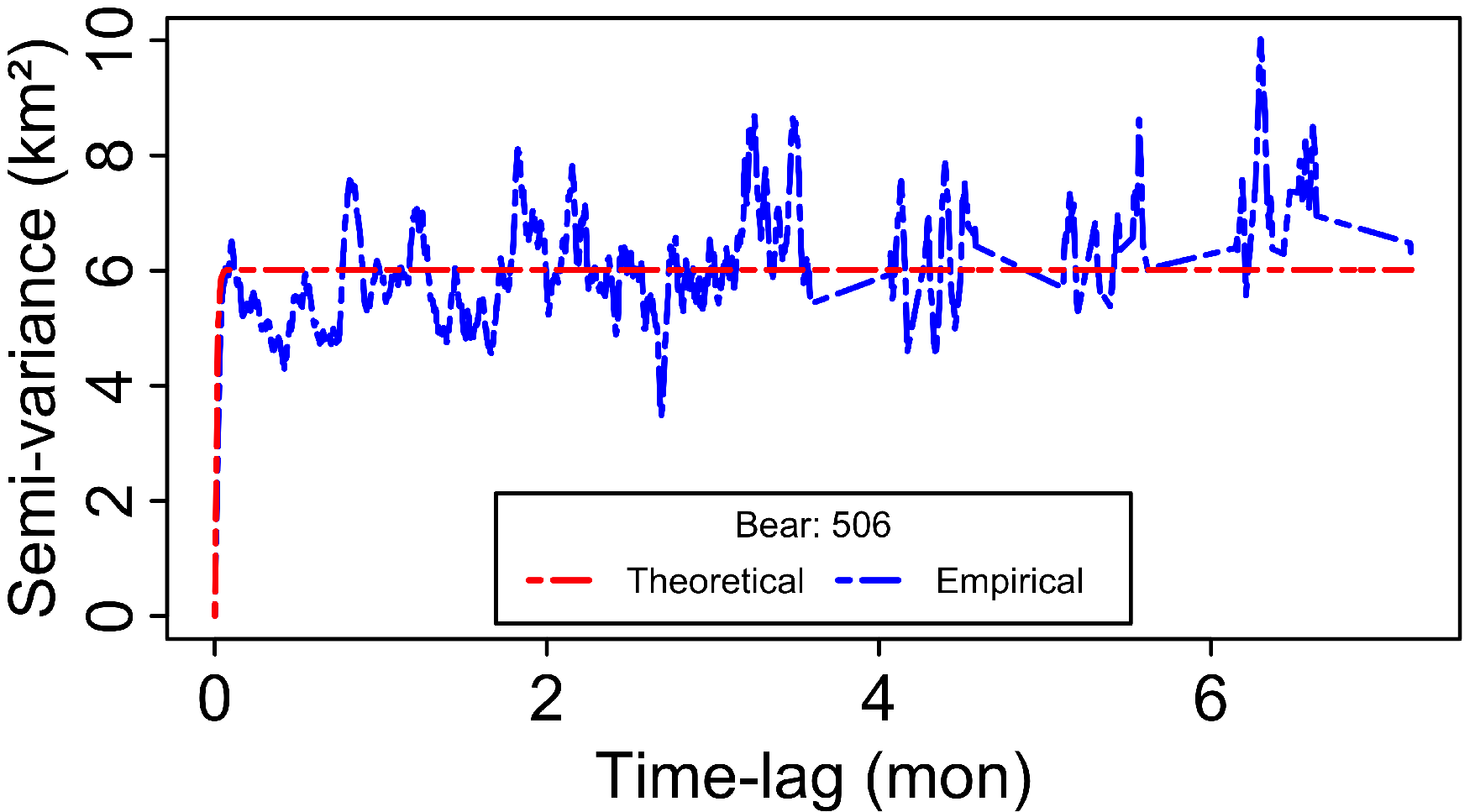}
\caption{ The fitted  theoretical semivariance models to  the empirical semivariances for the  twelve black bears in southern Alabama, USA.  The zigzag curves represent the empirical semivariances, while the other curves denote the fitted theoretical semivariance models. The $x$-axis of the plot represents time-lag in months, while the $y$-axis denotes semivariance in square kilometers (km$^{2}$) or square hectometers (hm$^{2}$)}
\label{Bearsve}
\end{figure}
It is worth noting here that black bear movement behaviour is extremely complex, as illustrated in Figure \ref{Or213}. We want to emphasize that black bears are assumed to be range residents if the empirical semivariance plots flatten out at larger lags (that is, when a range is achieved).  Based on the empirical semivariance plots in Figure \ref{Bearsve}, black bears 487, 488, 490, 491, 493, 495, 498, 501, and 506 exhibit home-range resident behaviours. We believe that the relocation data for these bears are suitable for home-range analysis. On the other hand, the empirical semivariance plots for black bears 503, 504, and 505 do not seem to flatten out, indicating that they display transient (or shifting-range) behaviours. Therefore, we excluded them from the home-range analysis.
\subsection{Home range estimation}\label{HRCA}
In this section, we present home range estimates derived from three methods: minimum convex polygon, kernel density estimation, and autocorrelated kernel density estimation.
\subsubsection{Minimum convex polygon estimation}
Minimum convex polygon estimation techniques are used to estimate the home ranges (95\%) and core home ranges (50\%) of black bears. We estimated the home ranges after removing 5\% and 50\% of relocations that are farther away from home range centroids, which are obtained by arithmetic mean of the relocations of black bears. It is widely accepted to use 95\% of relocations for estimation of home ranges \citep{powell2012home}. The data in Figure \ref{Or213} (left) show that the black bears are spatially isolated, forming two groups. As a result, we analyze the two black bear groups separately. The first group of black bears resides in Mobile County and Mount Vernon City, in southern Alabama, whereas the second group is located in Chatom, Fruitdale, and Wagarville Cities in Washington County, also in southern Alabama.  Figure \ref{Or214} presents the 95\% and 50\% home range estimates using the minimum convex polygon estimation technique.
\begin{figure}[h]
\centering
\includegraphics[width=0.4\textwidth, height=0.3\linewidth]{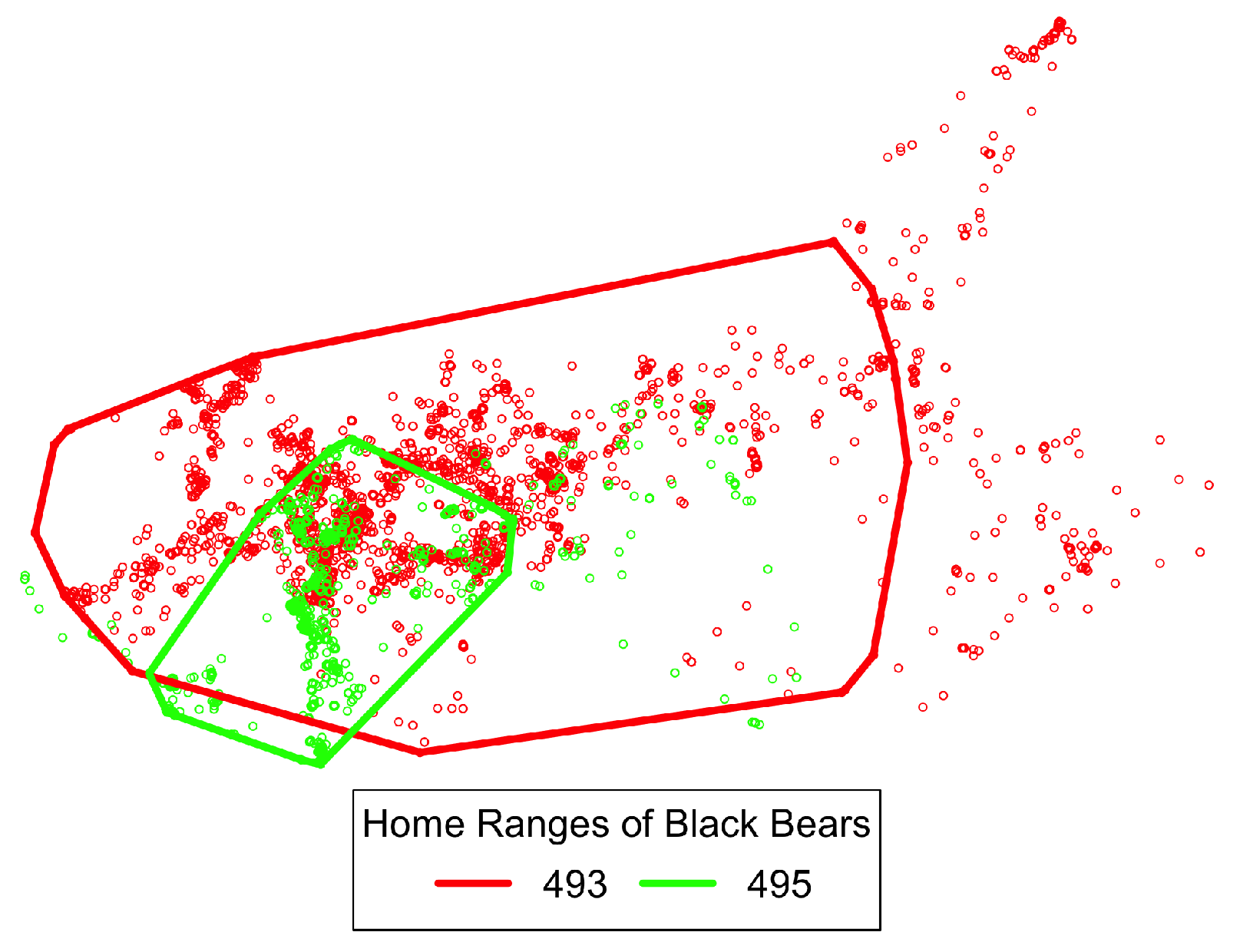}
\includegraphics[width=0.4\textwidth, height=0.3\linewidth]{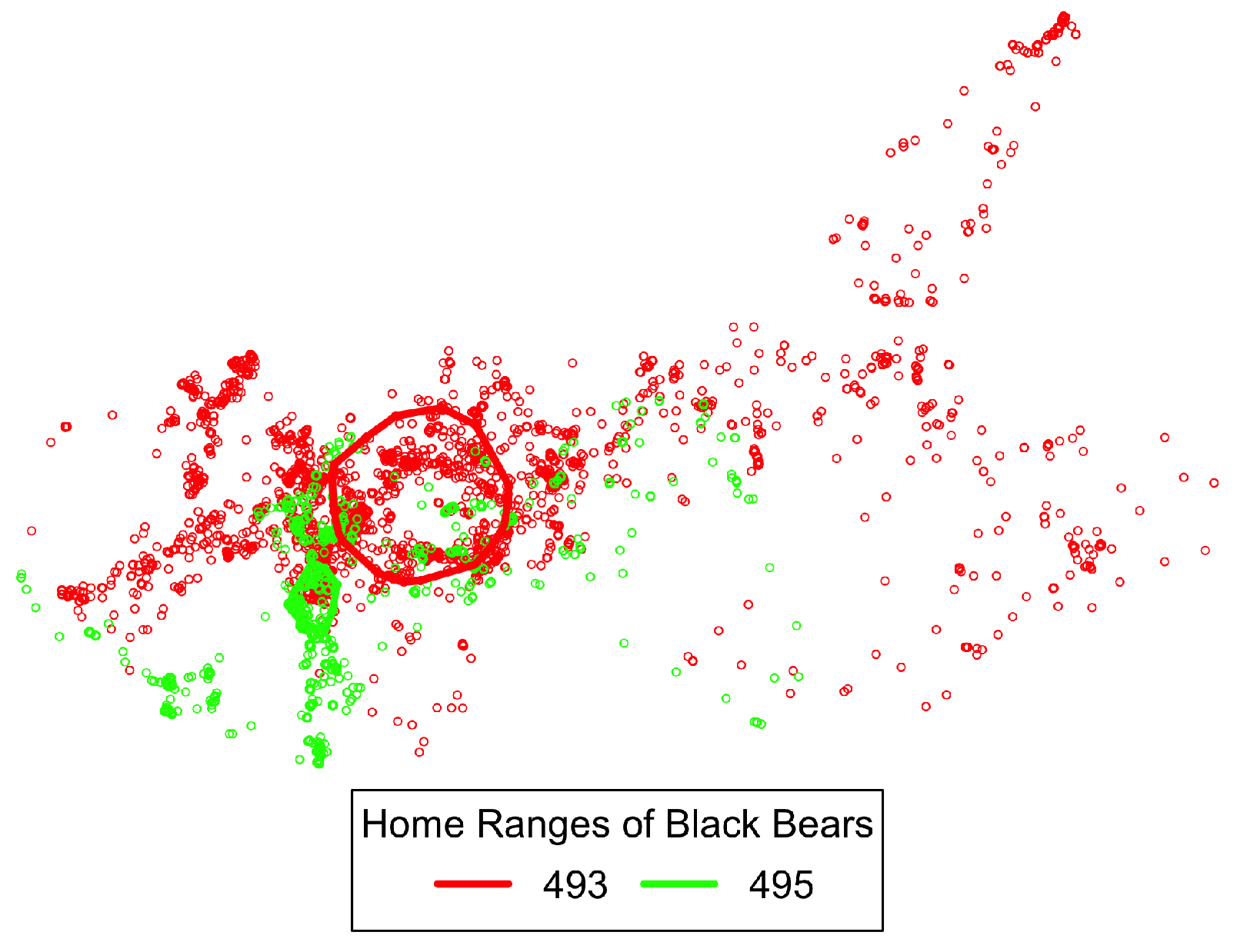}
\includegraphics[width=0.4\textwidth, height=0.5\linewidth]{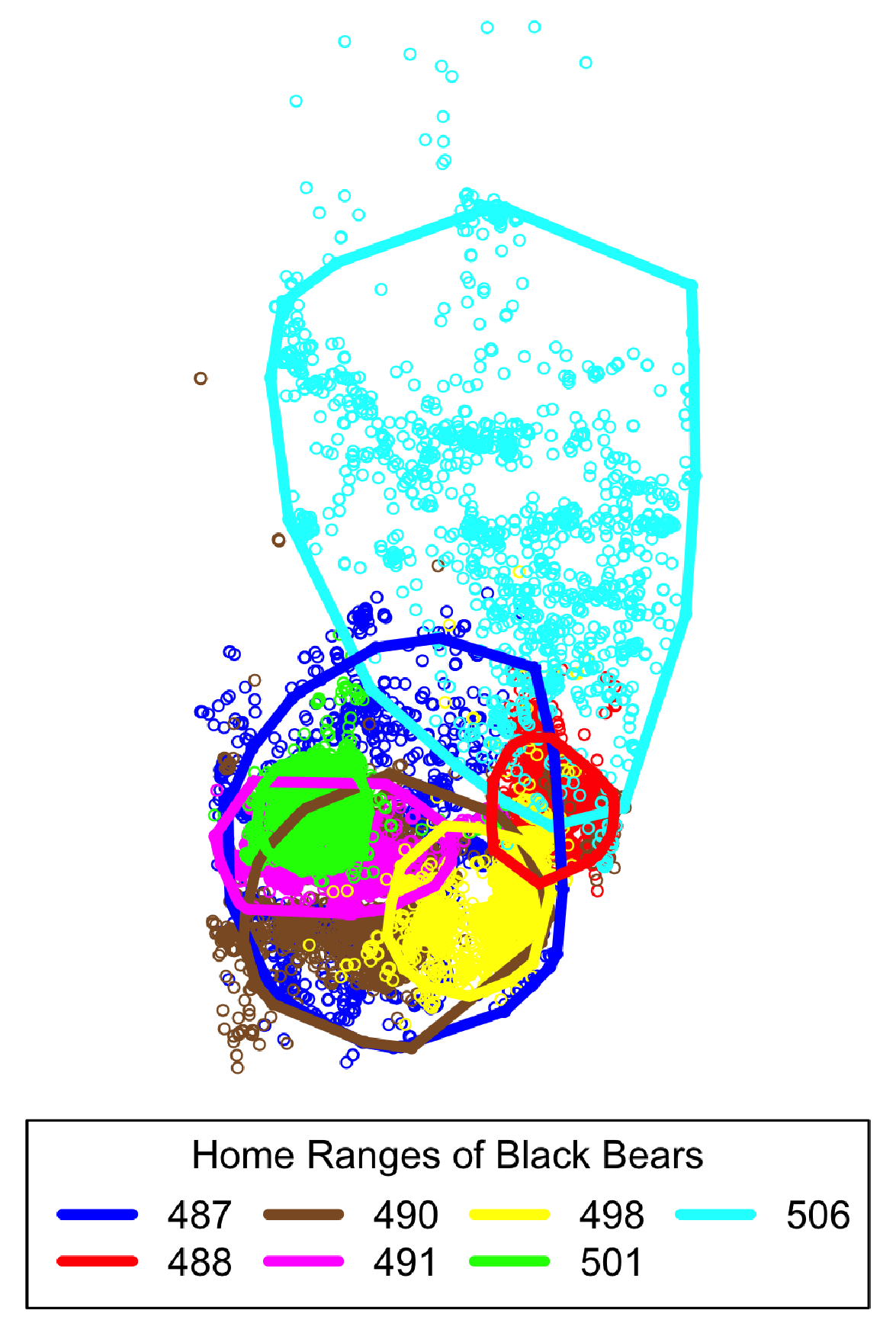}
\includegraphics[width=0.4\textwidth, height=0.5\linewidth]{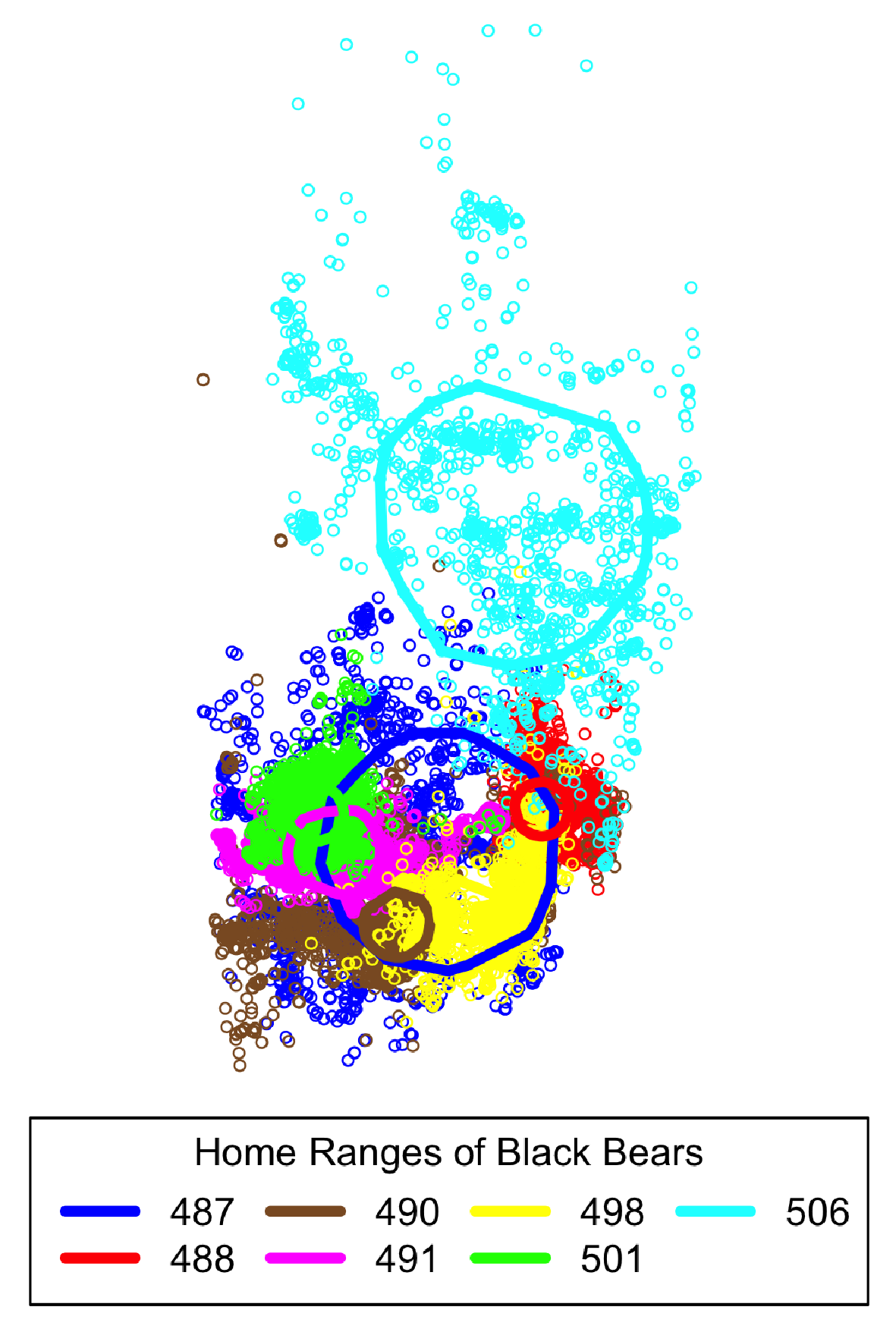}
\caption{The 95\% (left column) and 50\% (right column) home range estimates of black bears in Mount Vernon City and Mobile County (top row) and Chatom, Fruitdale, and Wagarville Cities of southern Alabama (bottom row) using the minimum convex polygon estimation. The numbers in the legends of the figure represent the black bear codes}
\label{Or214}
\end{figure}
We want to emphasize again that minimum convex polygon estimation can lead to overestimation of home ranges due to its convexity constraint.
\subsubsection{Kernel density estimation}
Bivariate normal kernels are also used for the home range estimations, and their smoothing parameters are selected in an \emph{ad hoc} manner \citep{silvermann1986density}. Figure \ref{Or21434} shows the estimated home ranges of black bears in  Mount Vernon, Chatom, Fruitdale, and Wagarville Cities  and Mobile County using bivariate normal kernel density estimators. In comparison to the estimated home ranges in Figure \ref{Or214}, the home range estimates in Figure \ref{Or21434} may be more accurate, as the non-convex nature of these estimates potentially excludes empty spaces where no relocation points of black bears are present.

\begin{figure}[h]
\centering
\includegraphics[width=0.4\textwidth, height=0.35\linewidth]{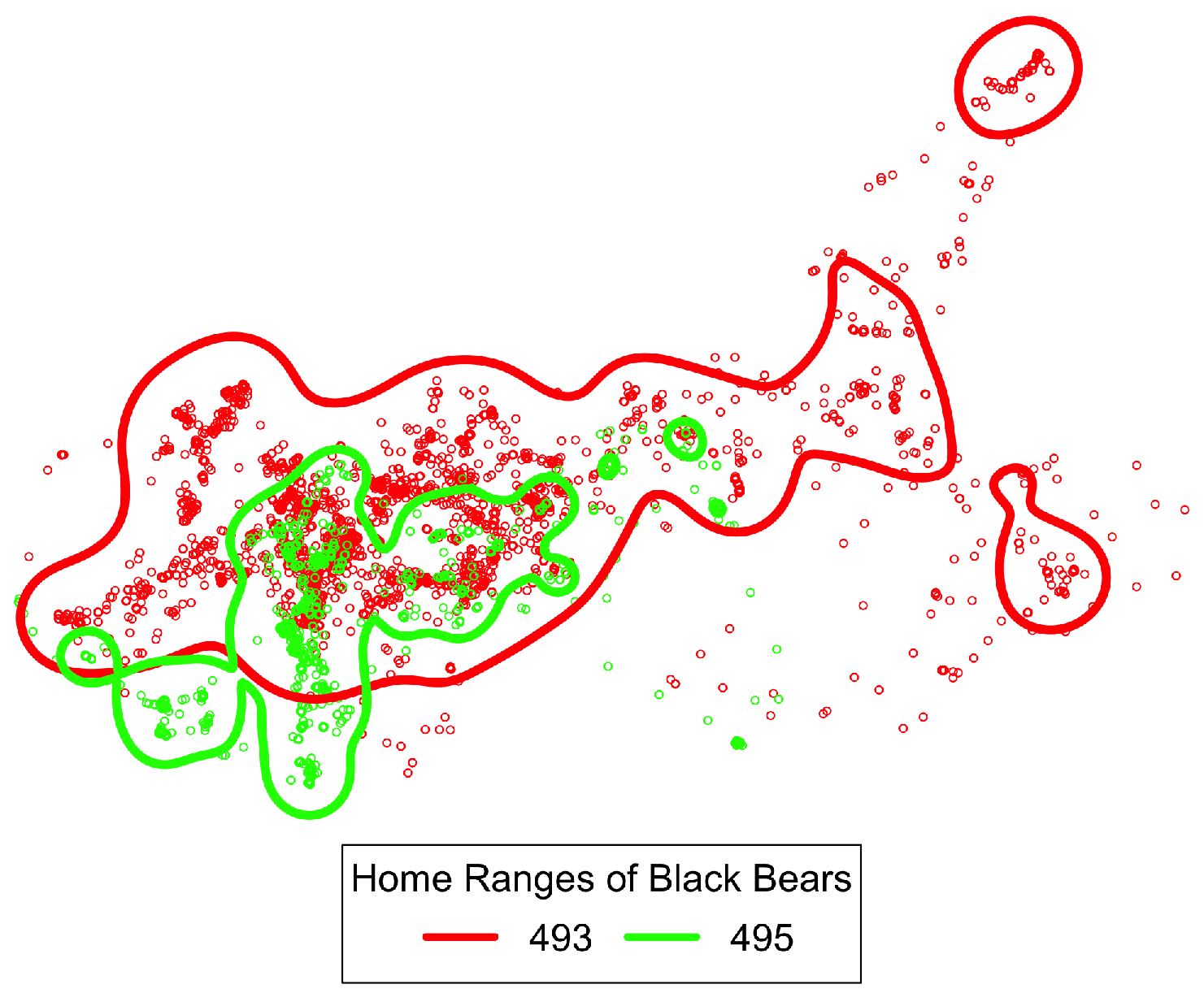}
\includegraphics[width=0.4\textwidth, height=0.35\linewidth]{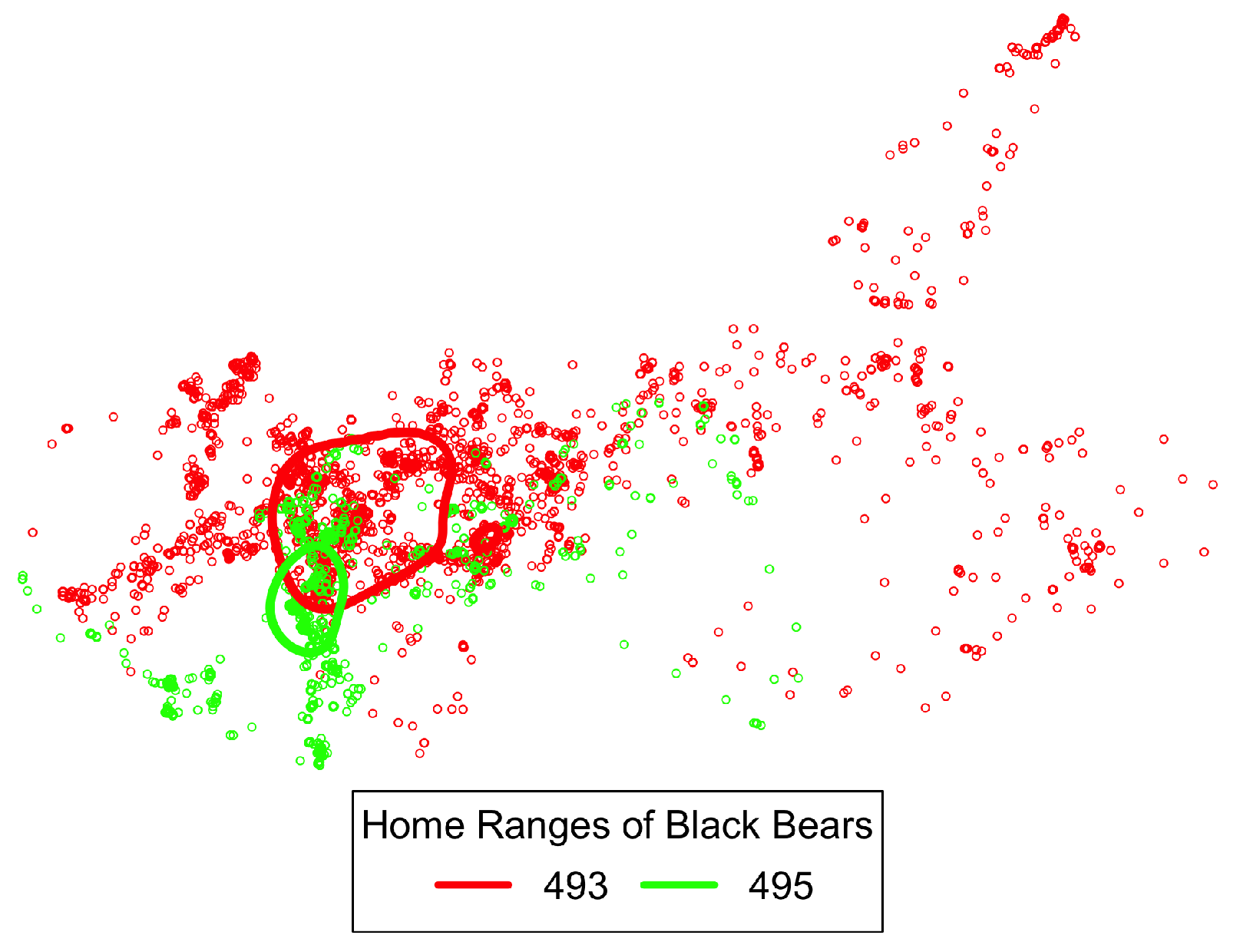}
\includegraphics[width=0.4\textwidth, height=0.55\linewidth]{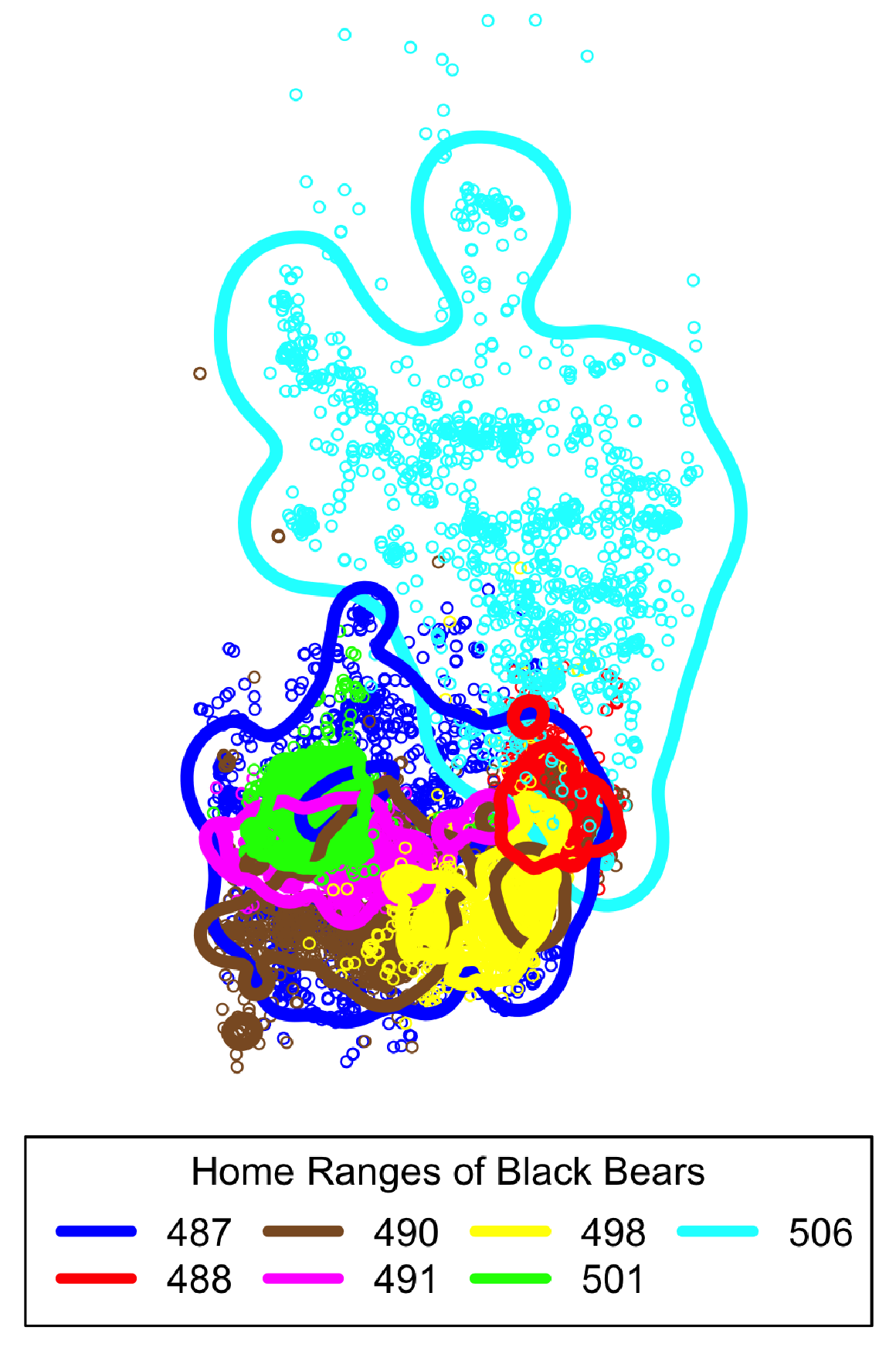}
\includegraphics[width=0.4\textwidth, height=0.55\linewidth]{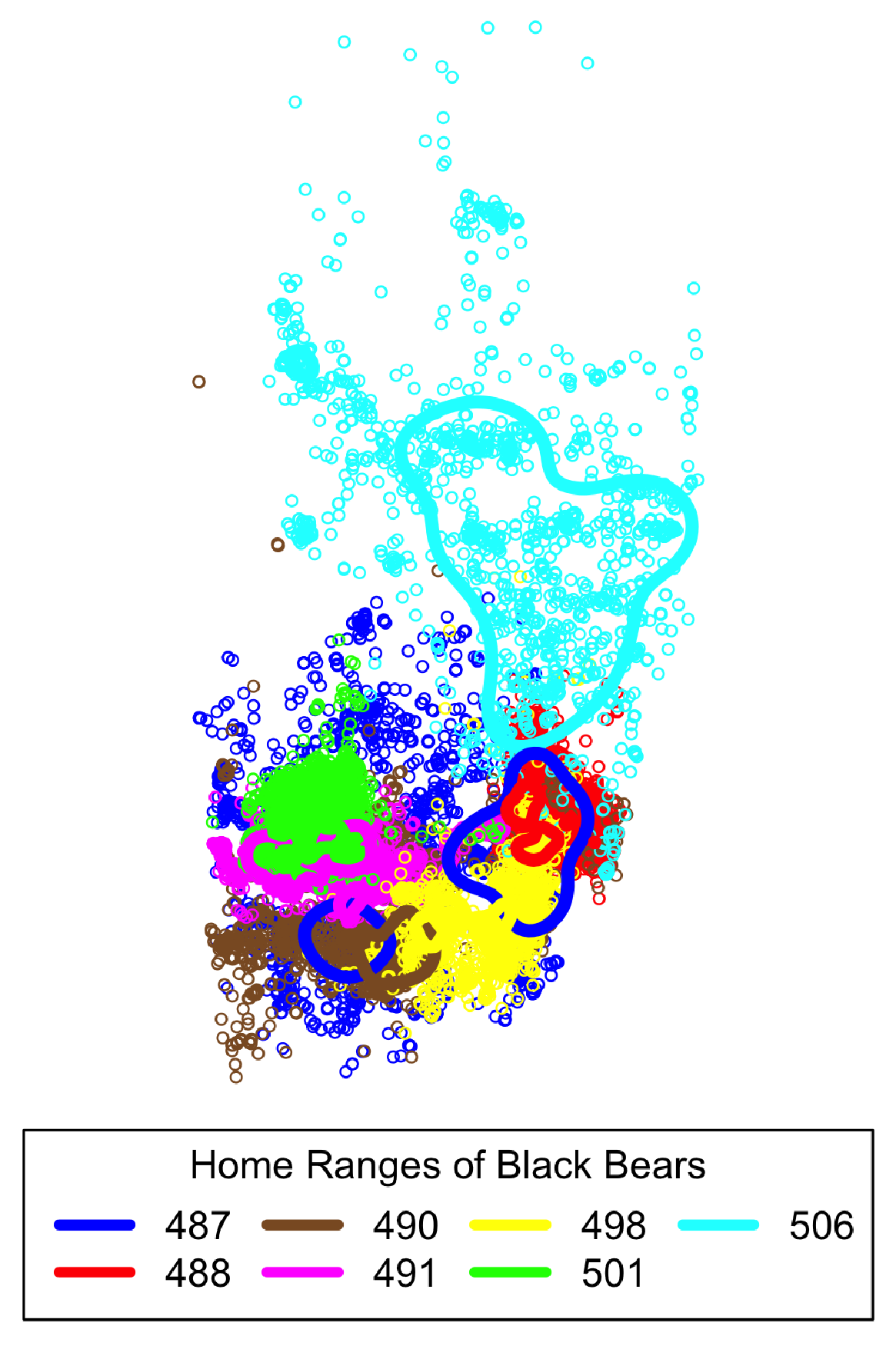}
\caption{The 95\% (left column) and 50\% (right column) home range estimates of black bears in Mount Vernon City and Mobile County (top row) and Chatom, Fruitdale, and Wagarville Cities of southern Alabama (bottom row) using bivariate normal kernel density estimators}
\label{Or21434}
\end{figure}
\subsubsection{Autocorrelated kernel density estimation}
Autocorrelated bivariate normal kernel density estimators may yield more accurate home range estimates.  In order to use autocorrelated bivariate normal kernel density estimators for home range estimation, we must first select theoretical semivariance models that best fit the empirical semivariances. When the underlying processes are \emph{iid}, the autocorrelated bivariate normal kernel density estimators become  the conventional (i.e., uncorrelated) bivariate normal kernel density estimators.\\

The semivariance models (or functions) can be used to assess whether the black bears are range residents or not.  They quantify the variations in distances between all relocation pairs at the same time lag.  When a black bear exhibits home range residence behaviour, the distance covered by the black bear approaches  an asymptote as the time lag between relocations increases. This is because black bears would normally not travel further than their home-range, even given more time. Thus, this method can detect home range residence behaviour.\\

We fit theoretical semivariance functions (\emph{iid}, \emph{OU}, \emph{OUF}) to empirical semivariances to identify optimal models using the \emph{AIC}. Visual assessment and $\Delta$\emph{AIC} aid model selection, with the optimal models  integrated into Figure  \ref{Bearsve} and used for home range estimation via autocorrelated bivariate normal kernel density estimators. Figure \ref{akdeOr21434} displays home ranges derived from these estimators, based on relocation datasets and selected semivariance function estimates. For \emph{iid} processes, these estimators reduce to conventional bivariate normal kernel density estimators. Home ranges in Figure \ref{akdeOr21434} are estimated for \emph{OU} anisotropic (487, 488, 491, 493, 495, 506), \emph{OUF} anisotropic (498, 501), and \emph{iid} anisotropic (490) processes.\\

\begin{figure}[h]
\centering
\includegraphics[width=0.4\textwidth, height=0.35\linewidth]{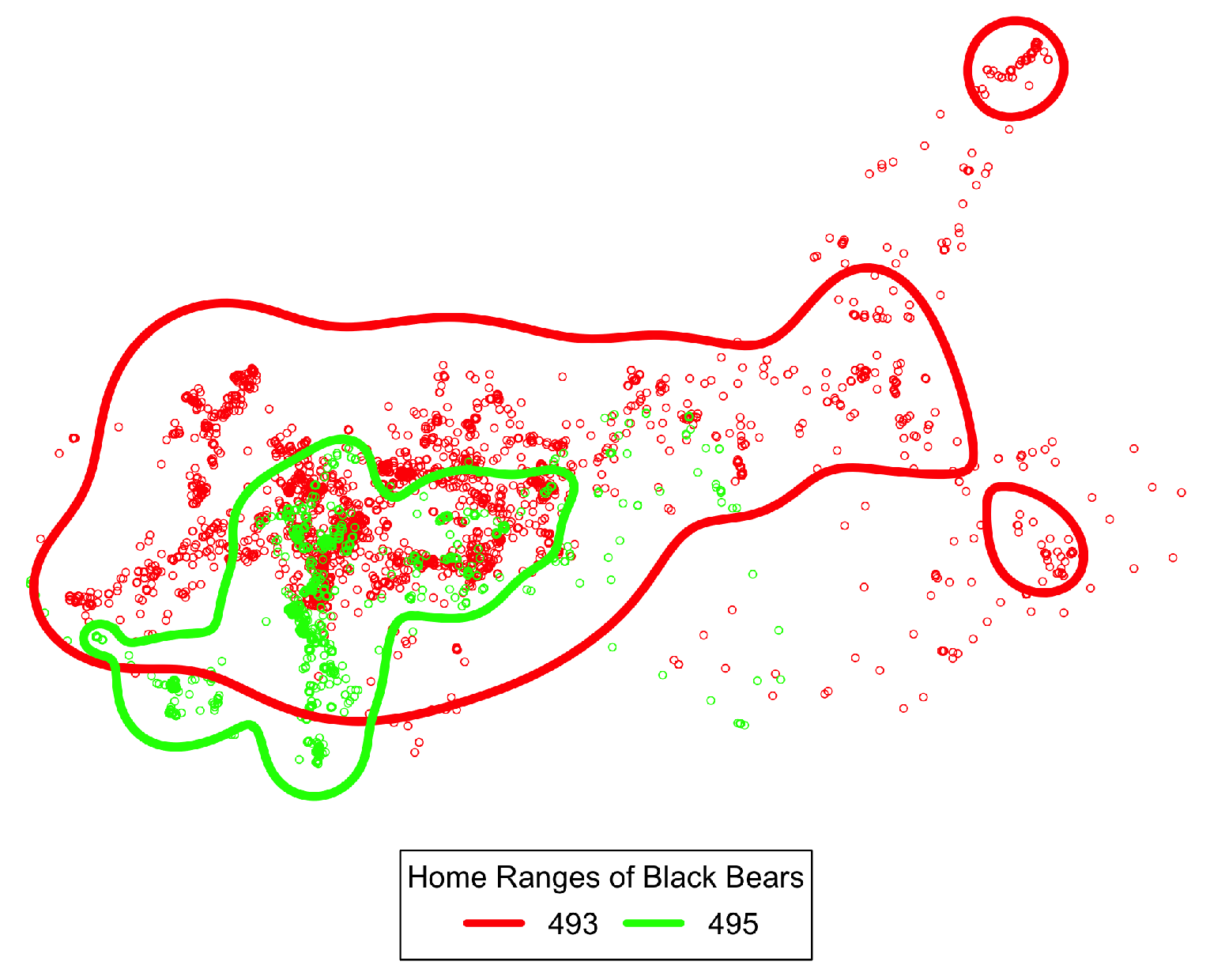}
\includegraphics[width=0.4\textwidth, height=0.35\linewidth]{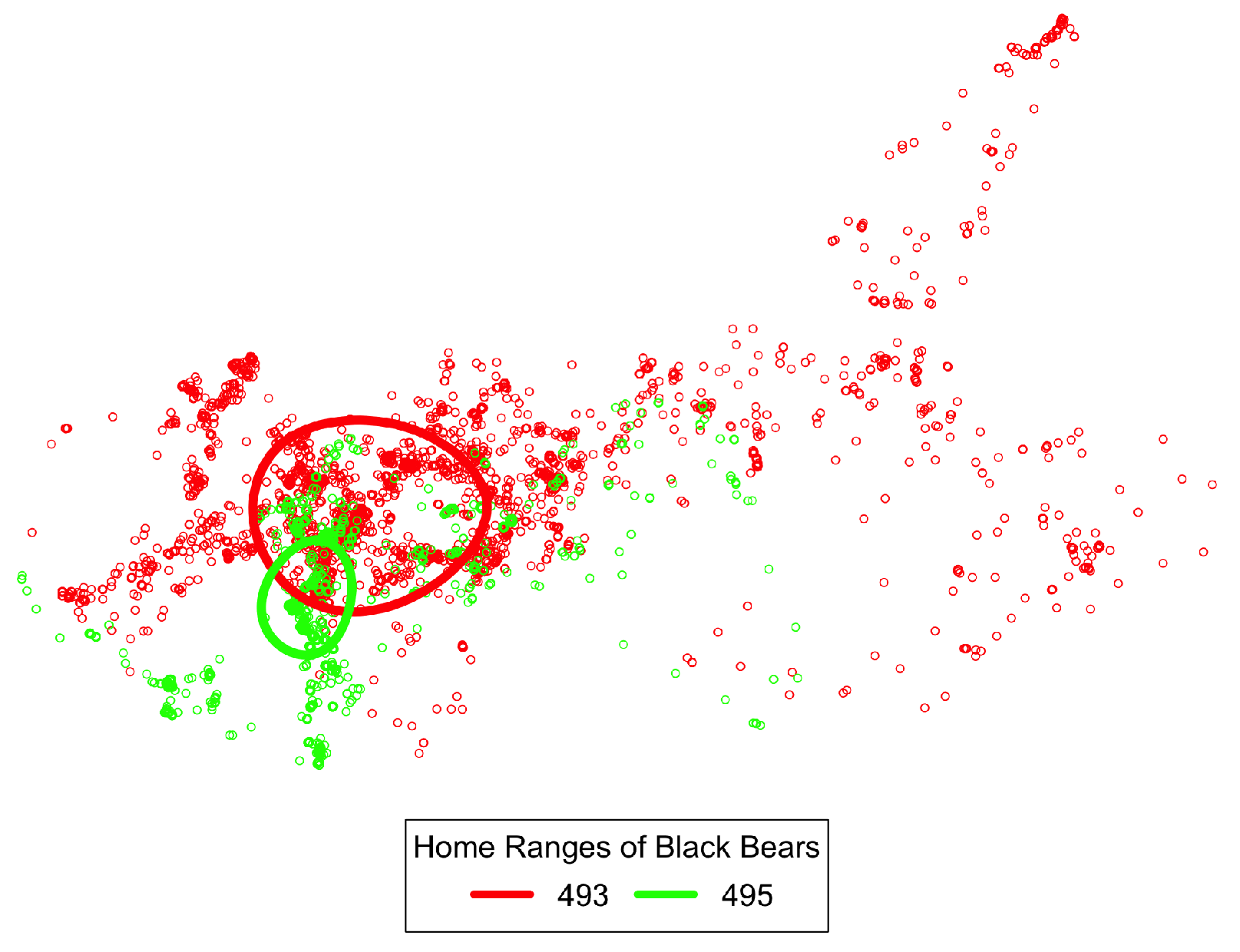}
\includegraphics[width=0.4\textwidth, height=0.55\linewidth]{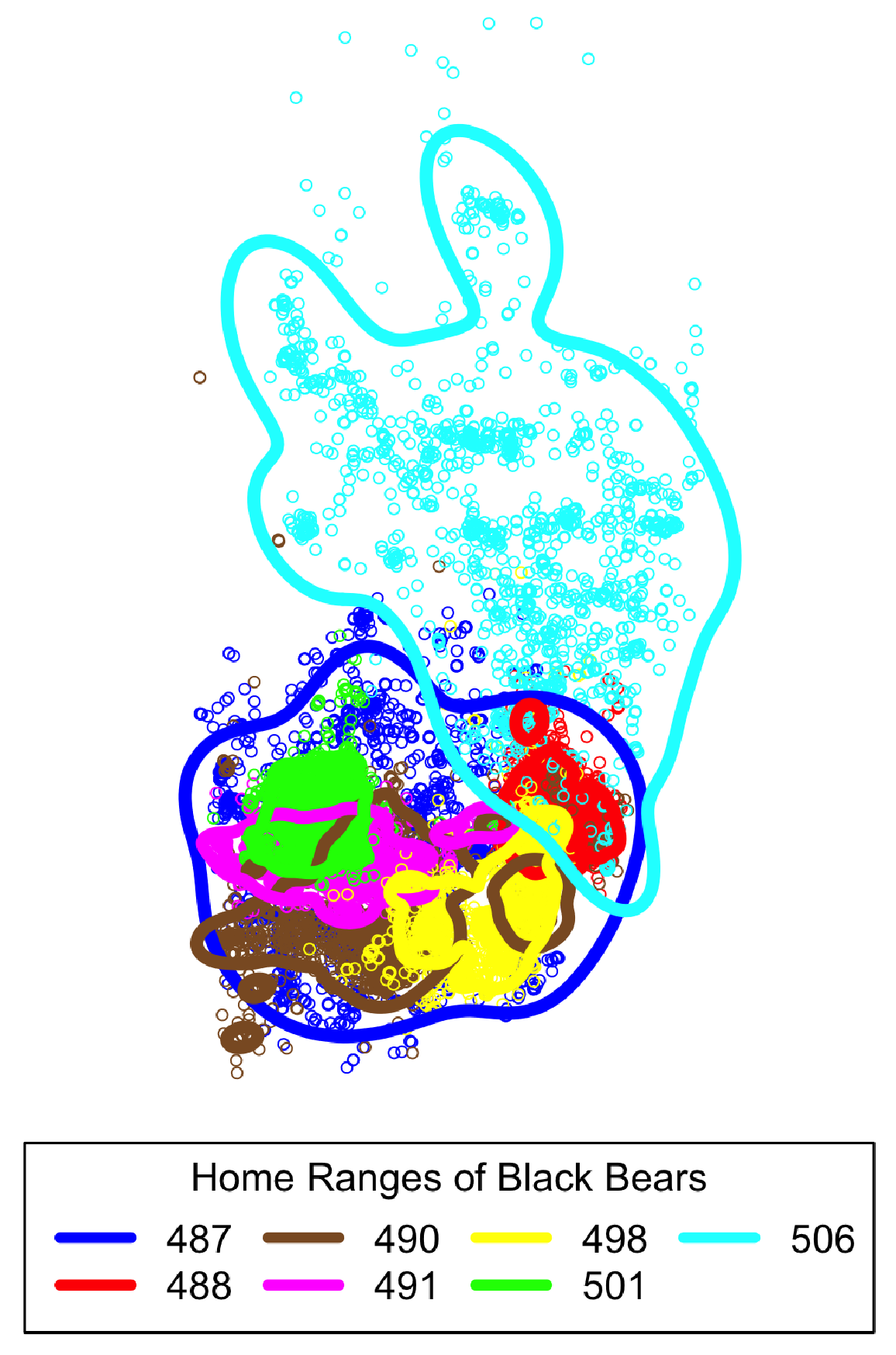}
\includegraphics[width=0.4\textwidth, height=0.55\linewidth]{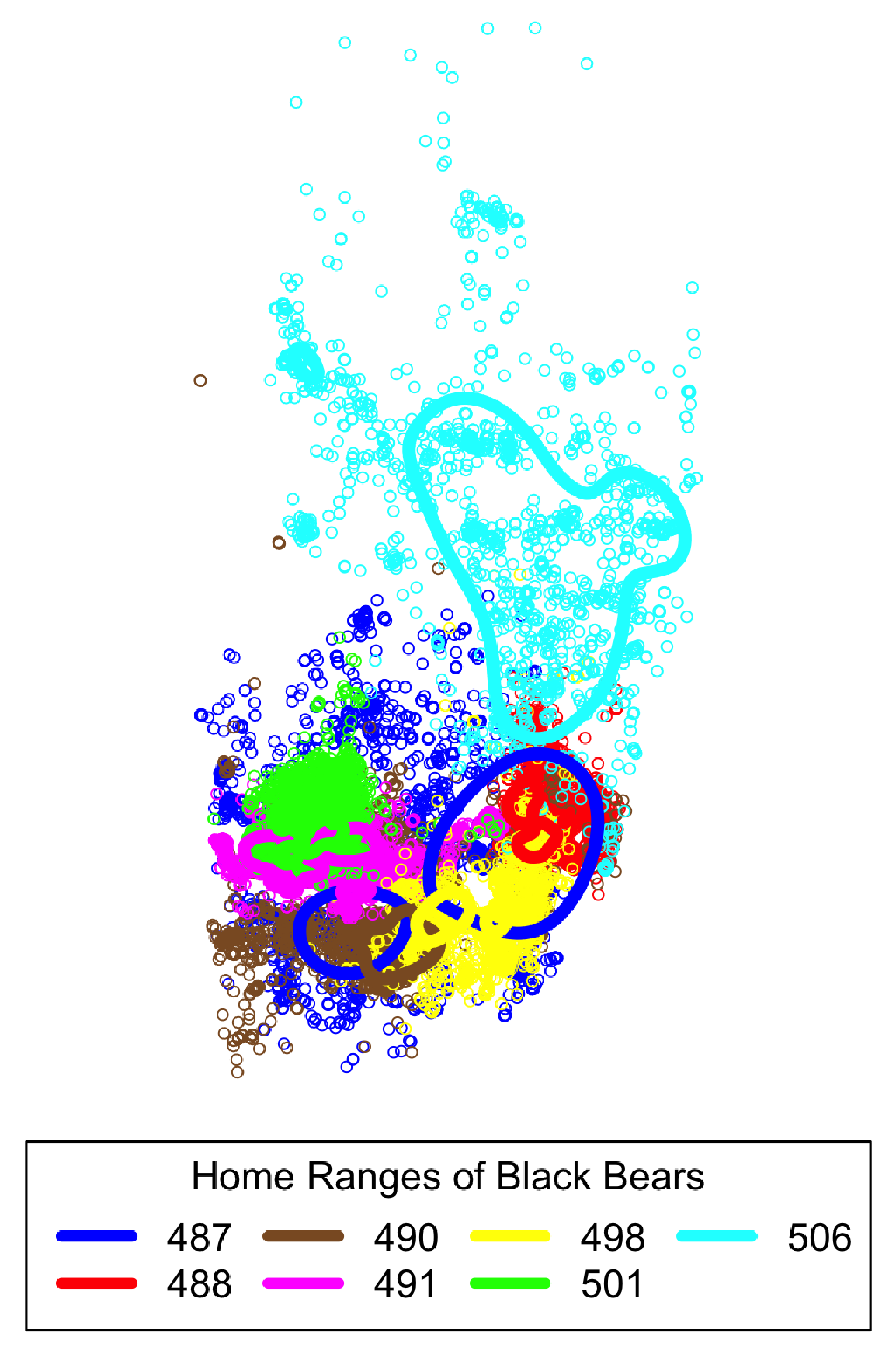}
\caption{The 95\% (left column) and 50\% (right column) home range estimates of black bears in Mobile County and Mount Vernon City (top row) and Chatom, Fruitdale, and Wagarville Cities in southern Alabama (bottom row) using autocorrelated kernel density estimators}
\label{akdeOr21434}
\end{figure}
Next, in Table \ref{tab:EstimatedModels1}, we present the estimated home range sizes (in km$^{2}$) based on the three methods described above. Notice that the kernel density estimators produce home range estimates when the underlying processes are \emph{iid} isotropic processes; see Section \ref{KDEsub}.
\begin{table}[h]
\centering
\caption{Estimated home range sizes (in square kilometers) by the minimum convex polygon (\emph{mcp}), kernel density estimator (\emph{kde}), and autocorrelated kernel density estimator (\emph{akde})}
\label{tab:EstimatedModels1}
\begin{tabular}{|c|c|c|c|c|c|c|c|c|}
\hline \multicolumn{1}{|c|}{No.} & \multicolumn{1}{|c|}{Black bear} & \multicolumn{2}{c|}{mcp} & \multicolumn{2}{c|}{kde}  & \multicolumn{2}{c|}{akde}\\\hline 
&& 95\% &50\%&95\%  &50\%&95\%&50\%\\\hline\hline
1&493 &228.885&16.115&163.272&17.120&206.559&24.226\\\hline
2&495 &47.639 &1.459&55.582&4.105&60.914&5.557\\\hline
3&487&42.700&17.012&43.684&7.446&50.570&11.409\\\hline
4&488 &5.174 &0.893&4.224 &0.668&4.076&0.726\\\hline
5&490&23.140 &1.300&15.539&1.469&14.755&1.395\\\hline
6&491 &10.051&2.432&8.065&1.847&7.090&1.686\\\hline
7&498 &8.304&1.154&6.216 &0.750&6.402&0.899\\\hline
8&501&4.001 &0.952&4.102& 0.835&3.994&0.901\\\hline
9&506&77.342&22.785&85.080& 21.166&78.082&19.823\\\hline\hline
\end{tabular}
\end{table}
The bivariate normal kernel is used both in kernel density estimation and autocorrelated kernel density estimation.
\subsection{Exploratory analysis of the pairwise interaction between black bears}
This section explores whether two black bears demonstrate independence, positive association (clustering), or negative association (avoidance/regularity/segregation). In other words, we want to study the preferences of black bears regarding their proximity to other black bears. Specifically, we intend to investigate whether a black bear is indifferent to the presence of another black bear nearby, seeks to coexist with another black bear, or actively avoids proximity to another black bear. To this end, we identified pairs of black bears with shared core home ranges in order to assess the nature of their spatial interaction.  The core home ranges are obtained by the three home range estimation methods discussed in Section \ref{HRCA}.  We only considered pairs of black bears that shared core home ranges for at least five months, as this provides sufficient data to make inferences about the nature of their spatial interactions. Table \ref{tab:PBBSHRFM} presents the pairs of black bears that share their core home ranges for at least five months. As the table shows, black bear 487 spatially interacts (static interaction) with most of the other black bears in its group. Similarly, Figure \ref{Or214} (bottom left) supports this conclusion, as its 95\% estimated home range encompasses the estimated home ranges of nearly all other black bears within the same group.  The permutations in Table \ref{tab:PBBSHRFM}  arise because the estimators \( L_{\text{inhom}}^{ij}(r) \) and \( L_{\text{inhom}}^{ji}(r) \) may not be identical due to the inclusion of edge effects. A similar rationale applies to \( J_{\text{inhom}}^{ij}(r) \). Therefore, the objective is to explore whether, for instance, black bear 487 interacts with 488 when black bear 488 is nearby, and conversely, whether black bear 488 interacts with 487 when black bear 487 is nearby. This explains why both \(\left( 487, 488\right)\) and \(\left(488, 487\right)\) are included in the table for the analysis of spatial interactions between pairs of black bears sharing their core home ranges using the summary statistics.

\begin{table}[h]
\centering
\caption{Pairs of black bears that share their core home ranges for at least five months}
\label{tab:PBBSHRFM}
\begin{tabular}{|cccc|}
\hline 
(487, 488)& (487, 490)&(487, 491)&(491, 501)\\
(488, 487) &(490, 487) &(491, 487)&(501, 491)\\\hline
\end{tabular}
\end{table}
A closer look at the relocation data plots of the black bears in Figure  \ref{Or213}  shows that their relocations are not evenly distributed across the study area. That is,  the point patterns corresponding to each mark are not homogeneous. Following that, we investigate whether there is attraction, repulsion, or no interaction between pairs of black bears using the inhomogeneous cross-type summary functions $L$ and $J$  presented in Section \ref{Summary2080}.  The spatial intensities must be known to estimate the summary functions $L$ and $J$. They are, however, unknown, and we must estimate them from the observed multitype point patterns. To accomplish this, we employ the parametric modelling described in Equation \eqref{LinearModel}. That is, we adopt the recommendation provided by \cite{baddeley2015spatial} to utilize the estimated spatial intensities derived from the observed point pattern.\\

To determine the type of spatial interaction between pairs of black bears, we compare the inhomogeneous cross-type $L$- and $J$-function estimates, obtained using border edge correction, from the observed spatial point patterns with those derived from simulated spatial point patterns generated by randomly shifting the observed spatial point patterns. When considering the null hypothesis of component or type independence, random shifting of the pattern and intensity function of one black bear in relation to the other black bears keeps the marginal structures unchanged and only affects the interactions between the two black bears. We are interested in running Monte Carlo tests to assess whether  the degrees of deviation of the estimated summary functions of the observed point pattern from the estimated summary functions of the simulated spatial point patterns are significant. To do so, we use two-sided hypotheses that test no spatial interaction between pairs of black bears  against spatial interaction alternatives; the procedures are detailed in Section \ref{MCET2080}.  We generate 2500 Monte Carlo simulations under the null hypothesis to construct the envelopes and test the hypotheses.\\

We present a sample of the results here, namely for one pair of black bears, and the remaining results can be found in the Online Resource. Figure \ref{Result495to503} shows the pointwise envelopes for the inhomogeneous cross-type $L$- and $J$-functions.  Both summary statistics indicate that the presence of black bear 487 near black bear 488 does not bother the latter, and vice versa.\\

\begin{figure}[H]
\centering
\includegraphics[width=0.35\textwidth, height=0.2\linewidth]{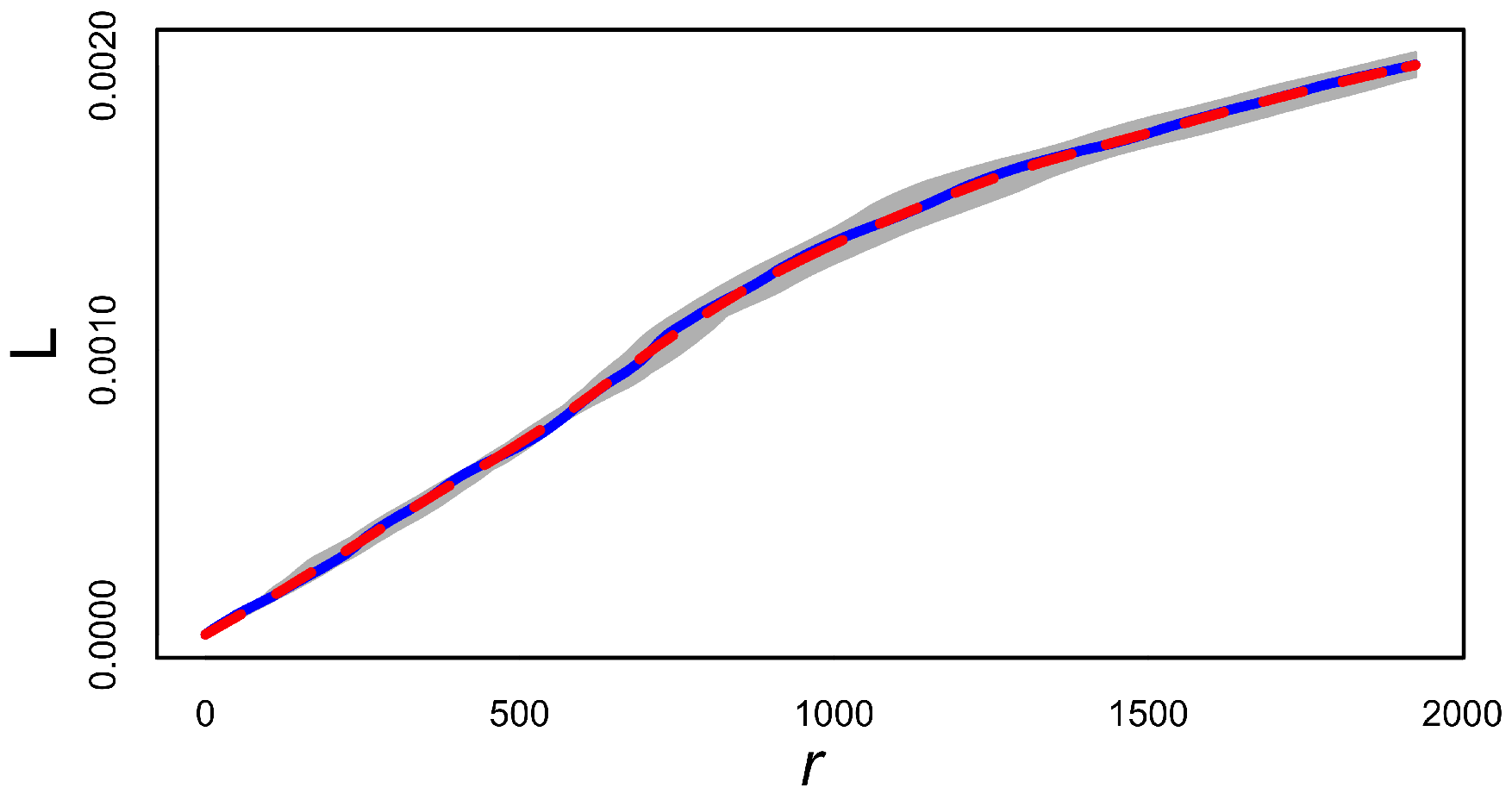}
\includegraphics[width=0.35\textwidth, height=0.2\linewidth]{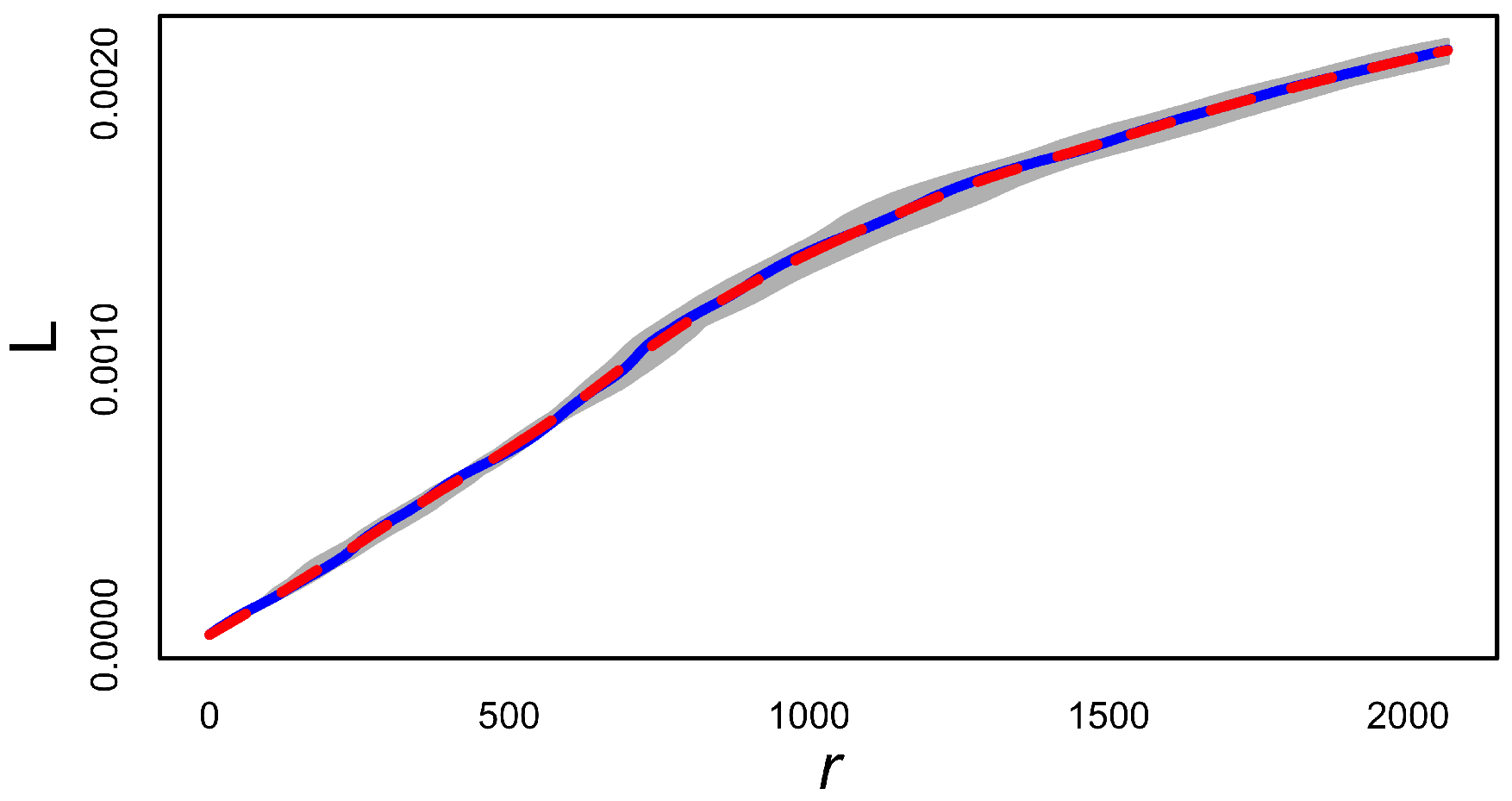}
\includegraphics[width=0.35\textwidth, height=0.2\linewidth]{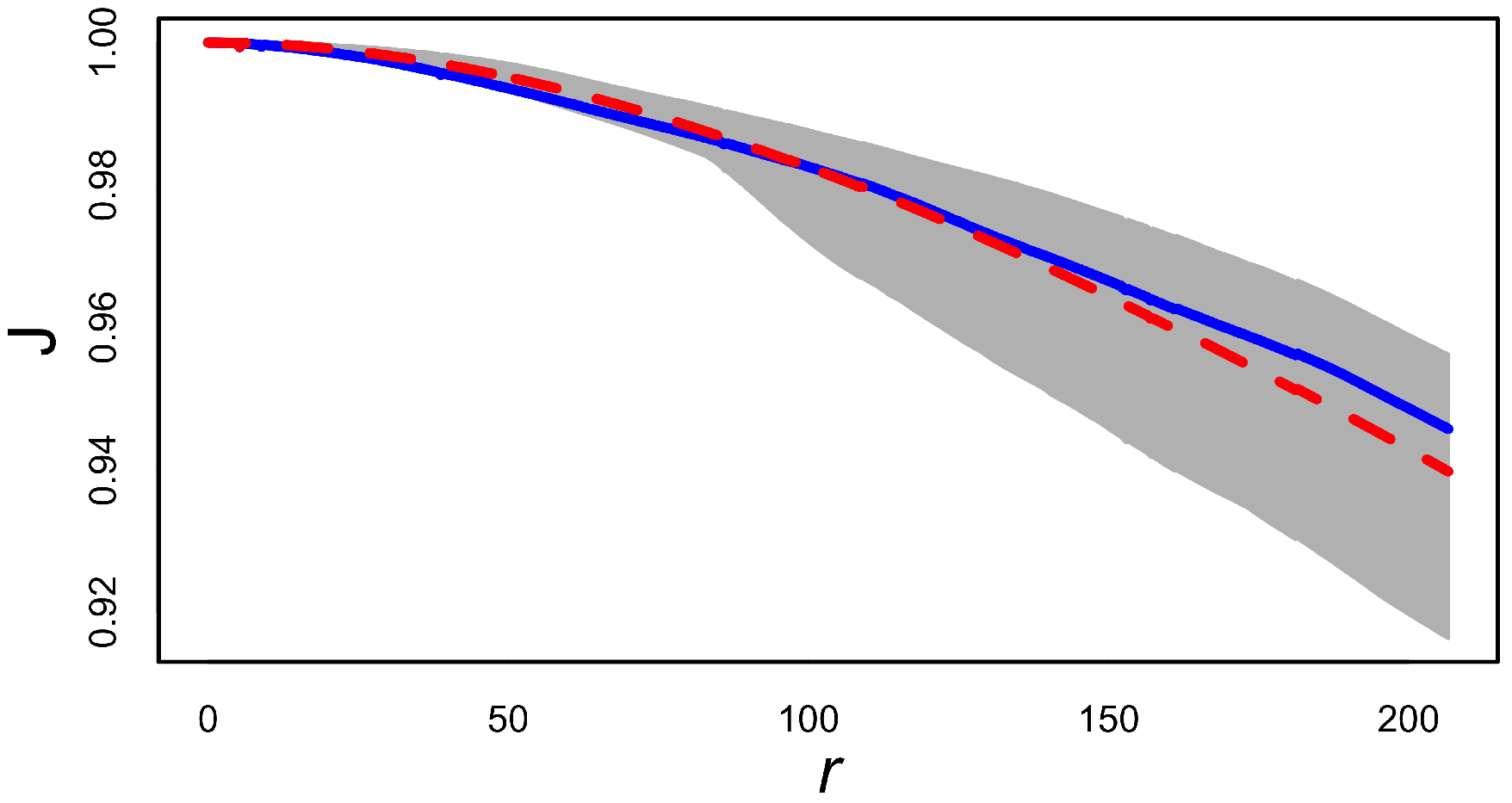}
\includegraphics[width=0.35\textwidth, height=0.2\linewidth]{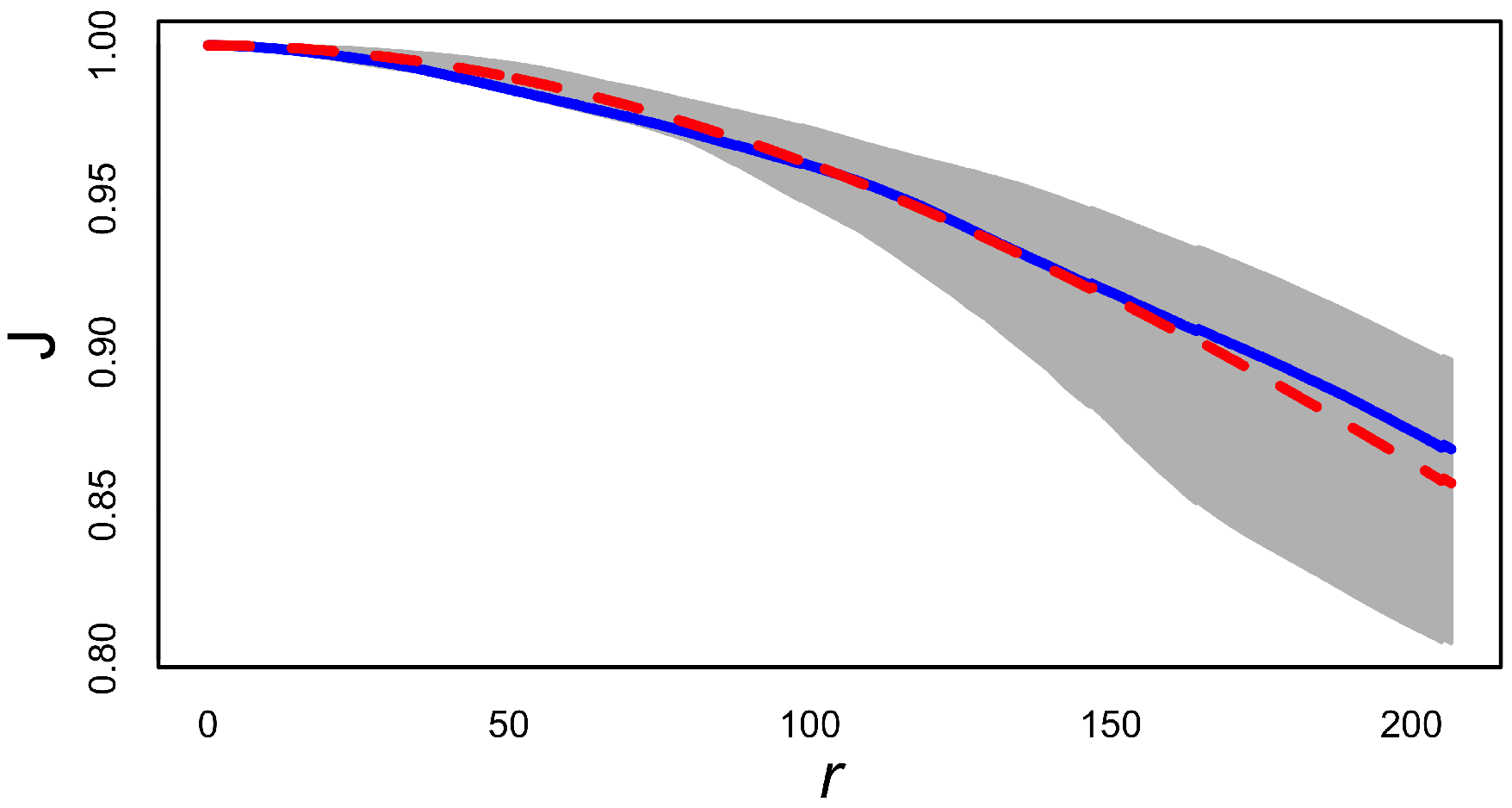}
\caption{Simulation envelopes for estimated inhomogeneous cross-type $L$- function (first row) and $J$-function (second row) from black bear 487  to 488  (first column)  and  from 488 to 487 (second column) based on 2500 simulations. The solid curve depicts estimates of the summary function for observed data, while the broken lines represent means of Monte Carlo simulations}
\label{Result495to503}
\end{figure}

Two-sided \textit{MAD} and \textit{DCLF} tests are used to test for lack of spatial interaction between the black bears 487 and 488.  Table \ref{tab:SummaryStatistic} presents the $p$-values for two-sided \textit{MAD} and \textit{DCLF} tests based on the inhomogeneous cross-type $L$-and $J$-functions. The $p$-values, which correspond to the interactions (as can be learned from the envelopes) shown in the Figure \ref{Result495to503}  from top left to bottom right (row-wise) for each test, are based on the range of $r$ values, specifically over the ranges $\left[0, 1925.200\right]$ and $\left[0, 206.570\right]$ in meters.

\begin{table}[H]
\centering
\caption{Two-sided \textit{MAD} and \textit{DCLF} tests for lack of spatial interaction based on  inhomogeneous cross-type $L$-and $J$-functions, with corresponding $p$-values}
\label{tab:SummaryStatistic}
\begin{tabular}{|c|c|c|c|c|c|}
\hline \multicolumn{1}{|c|}{Summary statistic}& \multicolumn{2}{c|}{MAD test} & \multicolumn{2}{c|}{DCLF test} \\\hline\hline
& 487 to 488&488 to 487&487 to 488 &488 to 487\\\hline\hline
L&0.985&0.986&0.990&0.993\\\hline
J&0.679&0.797&0.682& 0.769\\\hline\hline
\end{tabular}
\end{table}
Based on the $p$-values in Table \ref{tab:SummaryStatistic}, we conclude that there is no reason to reject the null hypothesis of independent components of the types/black bears 487 and 488 at 0.05 significance level. We would like to emphasize that both the envelope analysis and the hypothesis testing approaches indicate that black bears 487 and 488 are indifferent to living in close proximity to each other.

\section{Discussion and Conclusion}\label{Discussion}
The purpose of this study is to model the home range of wild animal species and the spatial interactions that may exist between pairs of the same-species wild animals. Understanding the home ranges of  wild animal species  and how they partition the areas they share can help with effective management of the wild animal species.   To illustrate the methodology, we use black bear relocation data from southern Alabama, USA.  \\

The data set at hand consists of spatial data for each of the twelve black bears that are tracked using GPS collars between the years 2015 and 2017. The black bear locations are tracked approximately hourly by GPS, and the data collection duration ranges from  197.041 to 414.334 days. The data set for a given black bear can be viewed in two ways. To begin, the data set can be viewed as a collection of spatial locations collected over a specific time period. Here, we ignore any autocorrelation that might be inherent in the data. In this case, the data set can be used to model black bear home ranges.  The data set can also be used to model the spatial interaction between pairs of black bears. The data set for a given black bear can also be viewed as a collection of the time series of the locations of the black bear, which is also called \emph{black bear relocation data}. In this instance, we can also use the data set to model black bear home ranges, taking the autocorrelation information that may exist in the black bear relocation data into account. \\

Our modelling strategy is to treat the data set, i.e., the spatial locations of the black bears, as a collection of realizations from inhomogeneous multitype point processes. The minimum convex polygon techniques and bivariate normal kernel density estimators are used to estimate home ranges. Outliers in the dataset can affect the minimum convex polygon estimation method, leading to the inclusion of empty spaces where there are no black bear relocation points. In this regard, the home range estimates in Figure \ref{Or21434} using bivariate normal kernels may be more accurate than those in Figure \ref{Or214}.   Autocorrelated bivariate normal kernel density estimators, as detailed in \cite{fleming2015rigorous}, use autocorrelation information from the relocation data to estimate home ranges. We reported the estimated home range sizes using the minimum convex polygon method, kernel density estimators, and autocorrelated kernel density estimators in Table \ref{tab:EstimatedModels1}.   The differences in the estimated home range sizes may be due to the modelling approaches, specifically whether the autocorrelation information from the data set is useful or not, and there could also be other factors as well. Accurate home range estimation is critical in ecology, as it might provide insight into the underlying ecological processes and suggest a promising course of action for future research. For example, it can aid in the development of management plans to adequately protect black bears.\\

The sizes of the estimated home ranges of black bears vary significantly, as shown in Figures \ref{Or214} to \ref{akdeOr21434} and Table \ref{tab:EstimatedModels1}, influenced by factors such as habitat quality, food availability, population density, etc. Some home ranges are notably larger, with certain ones entirely lying within others.  Larger home ranges may belong to male black bears, which typically cover more extensive areas while searching for mates and resources \citep{gautrelet2024first}. Despite significant overlap, black bears generally avoid direct interactions through strategies like temporal avoidance to reduce conflict. The smaller estimated home ranges likely belong to female black bears, whose ranges are typically more compact and may overlap with those of their offspring or related individuals, reflecting familial spatial use \citep{carroll2024evaluating}. Conversely, male home ranges often encompass those of several females, consistent with their polygynous mating system \citep{vidal2024negative}. The overlapping and nested home ranges suggest shared resources or movement corridors, highlighting the complex social and ecological (or spatial) dynamics of black bear.\\

\cite{fronville2024performance} have assessed three types of interactions among moving individuals, namely avoidance, attraction, and neutral movement, using dynamic interaction methods applied to data from a predator-prey system simulation.  Intraspecies interaction modelling of black bears can help understand how they partition the shared space in which they live. The shared space of the black bears can be partitioned through attraction/clustering, repulsion/avoidance/regular patterns, or random/no interaction. To study how black bears interact in this context, inhomogeneous cross-type summary functions can be employed. These functions aid in evaluating whether the black bear interactions exhibit positive association/clustering, negative association/repulsion, or random/no interaction. Simulation-based statistical methods can be utilized to evaluate the nature of interactions between pairs of black bears. To this end, when utilized appropriately, pointwise envelopes can be employed for testing the null hypothesis of no interaction. However, it carries a higher risk of misinterpretation due to the issue of multiple testing \citep{baddeley2015spatial}. As a result, in addition to a pointwise envelope, a global envelope is used to test the null hypothesis. In general, our results suggest there may not be spatial interaction between pairs of black bears in our dataset. That is, black bears do not exhibit any preference for or aversion to living in close proximity to each other. Although the study suggests a lack of spatial interaction between pairs of black bears, caution is advised before drawing any firm conclusions. The generalization of the findings to other black bear populations may depend on various factors. A larger data set is primarily needed for increased test power, better estimation, and generalizability. Furthermore, we need to keep in mind that the results of the Monte Carlo tests can be affected by a number of factors, including the type of edge correction used, the estimated intensity function, the bandwidth selected, and the range of interaction considered. The robustness of the results in relation to these choices should also be investigated. Nonetheless, we believe that this paper provides a sound foundation for black bear interaction modelling (or modelling the same-species wild animal spatial interactions).\\

Our study conducted a comparison of home range estimation methods in black bears, revealing variations in the sizes of their home range estimates.  The variations in home range sizes may suggest gender diversity within the studied black bear population and could also be influenced by the roles or responsibilities of the bears within the population, as well as the availability of resources and the landscape in the area \citep{braunstein2020black}.  We observed differences in the movement patterns of black bears tracked using GPS collars. Our findings suggest that black bears display an overall indifference towards living in close proximity or dispersed areas. These findings contribute to our understanding of black bear ecology and have significant implications for their management and conservation. By using these insights, we can develop effective strategies for the management and conservation of black bear populations.

\section{Future work}\label{FutureWork}
We put forward three approaches to further investigate the home range and spatial interactions of wild animal species, focusing on black bears in particular.\\

We suggest an adaptive kernel density estimator to extend the exploration of the home range of wild animal species. Let $Z = \{(\z_i, t_i) \}_{i=1}^n$ be the set of spatiotemporal relocation data of a wild animal species, where $\z_i$ represents the spatial location at time $t_i$. Assume that $\Lambda = \{\lambda(\z, t) \}$ is the spatiotemporal intensity function of the home range of a wild animal species, varying with both location $\z$ and time $t$, influenced by dynamic environmental covariates $E(t)$. The home range (H) of the wild animal species can be estimated by an adaptive kernel density estimator that incorporates these dynamic covariates. Define the adaptive kernel density estimator $\hat{p}_{\Lambda, \;E}(\z, t)$ as:
\[\hat{p}_{\Lambda, \;E}(\z, t) = \frac{1}{n} \sum_{i=1}^n \frac{1}{|\Lambda(\z_i, t_i)|^{1/2}} K \left( \Lambda(\z_i, t_i)^{-1/2} (\z - \z_i) \right) \cdot g\left(E(t_i), E(t)\right), \]
where $K$ is a bivariate kernel function, $\Lambda(\z_i, t_i)$ is the adaptive bandwidth matrix at point $(\z_i, t_i)$, and $g(E(t_i), E(t))\ge 0$ is a function representing the influence of environmental covariates at times $t_i$ and $t$.
\begin{theorem}
Under the assumption that the movement patterns of a wild animal species are influenced by dynamically changing environmental covariates, the adaptive kernel density estimator $\hat{p}_{\Lambda, E}(\z, t)$ provides a more accurate estimation of the home range $H$ of a wild animal species compared to traditional kernel density estimators that do not account for temporal changes in environmental covariates.
\end{theorem}
\begin{proof}[Proof sketch]
The spatiotemporal intensity function $\Lambda(z, t)$ is adapted to reflect changes in environmental covariates $E(t)$, which influence the movement and habitat preference of a wild animal species. By integrating the function $g(E(t_i), E(t))$, the estimator $\hat{p}_{\Lambda, E}(z, t)$ adjusts the weight of each point based on the similarity of environmental conditions at times $t_i$ and $t$. The adaptive bandwidth matrix $\Lambda(\z_i, t_i)$ ensures that the kernel density estimator is sensitive to local variations in spatiotemporal intensity, providing a more accurate and localized estimate of the home range. Traditional kernel density estimators assume a static environment, leading to potential biases in the presence of dynamic environmental factors. The adaptive estimator accounts for these dynamics, reducing bias and improving accuracy.
\end{proof}

We propose spatial point process models with covariates and hierarchical Bayesian models to further investigate the spatial interaction of a wild animal species.  Spatial point process models with covariates can be used to incorporate environmental covariates into spatial point process models, such as generalized additive models (GAMs) or point process models with spatial covariates. These models enhance the understanding of species interactions and habitat preferences by accounting for habitat heterogeneity and environmental gradients, offering precise insights into how different factors influence spatial patterns. For instance, GAMs can model the intensity of point processes based on environmental variables like vegetation type and elevation, thus capturing complex, non-linear effects. On the other hand,  hierarchical Bayesian models provide a flexible framework that can incorporate various sources of uncertainty and hierarchical structures within the data. These models can simultaneously account for individual and group-level processes, integrate prior information, and offer comprehensive posterior distributions of parameters. An example application of a hierarchical Bayesian model is its use to capture individual variability in movement patterns while also estimating broader population-level trends.

\section*{Supplementary Information}
The supplementary material is available in the appendix section.
\section*{Acknowledgments}
Statistical computations are conducted using Auburn University and Alabama Supercomputer High Performance Computing Centers, centrally managed resources available to faculty, staff, students, and collaborators. We thank them for their support during our research.
\section*{Statements and Declarations}

\begin{itemize}[left=-9pt]
\item[]\textbf{Funding} The Alabama Department of Conservation and Natural Resources provided funding and technical assistance for data collection on bears.
\item[]\textbf{Competing interest} The authors declare no competing interest.
\item[]\textbf{Data availability} The bear data is available upon request from the authors.
\item[]\textbf{Code availability} The R code to obtain the results in the work can be found at https://github.com/harmee2020/FekaduBayisa.
\item[]\textbf{Author contribution} Conceptualization: Fekadu L. Bayisa, Elvan Ceyhan, and Todd D. Steury; Methodology: Fekadu L. Bayisa and Elvan Ceyhan; Data Collection: Christopher L. Seals, Hannah J. Leeper, and Todd D. Steury; Formal Data Analysis and Investigation: Fekadu L. Bayisa; Writing: Original Draft Preparation: Fekadu L. Bayisa; Writing: Review and Editing: Elvan Ceyhan; Todd D. Steury; Funding Acquisition: Todd D. Steury; Supervision: Elvan Ceyhan and Todd D. Steury.
\end{itemize}

\bibliography{sn-bibliography}
\clearpage 

\appendix
\section*{Supplementary Information}
This section presents additional information about the contents of the article.

\captionsetup[algorithm]{labelfont=rm,labelsep=period}
\captionsetup[figure]{labelfont={bf},name={Fig.},labelsep = space}

\renewenvironment{table}[1][]%
{\tableorg[#1]%
\tablebodyfont%
\renewcommand\footnotetext[2][]{{\removelastskip\vskip3pt%
\let\tablebodyfont\tablefootnotefont%
\hskip0pt\if!##1!\else{\smash{$^{##1}$}}\fi##2\par}}%
}{\endtableorg}

\def\P{\mathbb P}
\def\E{\mathbb E}
\def\R{\mathbb R}
\def\Var{\mathbb Var}
\def\Cov{\mathbb Cov}
\def\du{du}
\def\P{\mathbb P}
\def\R{\mathbb R}
\def\1{\mathbf 1}
\def\x{\mathbf x}
\def\z{\mathbf z}
\def\y{\mathbf y}
\def\E{\mathbb E}
\def\W{\mathbb W}
\def\w{\mathbf w}

\DeclareFontFamily{U}{BOONDOX-calo}{\skewchar\font=45 }
\DeclareFontShape{U}{BOONDOX-calo}{m}{n}{ <-> s*[1.05] BOONDOX-r-calo}{}
\DeclareFontShape{U}{BOONDOX-calo}{b}{n}{ <-> s*[1.05] BOONDOX-b-calo}{}

\parindent=0pt

\setcounter{secnumdepth}{3}

\setcitestyle{aysep={}} 

\hypersetup{colorlinks = true, linkcolor = blue, filecolor = blue, urlcolor = blue}

\section{Stochastic movement processes}
This section provides an overview of the statistical methods used to estimate the home ranges of wild animal species. Towards this end, let $ \w = \left\lbrace \z_{t_{i}} =\left(x_{t_{i}}, y_{t_{i}}\right)\mid t_{i} \ge 0, \; i = 1, 2, \ldots, n\right\rbrace  \subset \W$ denote the relocation data of a wild animal in a region $ \W\subset \mathbb{R}^{2}$ at times $t_{i}$, $i = 1, 2, \ldots, n$. Here, $x_{t_{i}}$ and $y_{t_{i}}$ represent the Longitude and Latitude coordinates of the wild animal, respectively, in the Global Positioning System (GPS) at times $t_{i}$.\\

Relocation data $\w$ for wild animal species can be viewed as repeated observations of the locations of individual wild animals. The relocation data are usually collected with a constant sampling rate, resulting in a fixed time interval between each collection. Current approaches for identifying multiple movement behaviours or movement modes in such data sets proceed locally through each path \citep{mcclintock2012general}. The relocation data for wild animal species can be taken as a realization of  a stochastic process, a sequence of time-indexed random variables $\mathbf{z}\left(t\right)$ with potential time correlation. The mean location $\boldsymbol\mu\left( t\right)$ of a nonstationary stochastic movement process $\mathbf{z}\left(t\right)$ can indicate shifts over time, reflecting movement behaviours like movement within-home ranges. In such cases, analyzing the autocorrelation function, or equivalently, the \emph{semivariance function} of the stochastic movement process can provide a comprehensive understanding of the movement behaviours.

\subsection{Semivariance function}
Let $\mathbf{z}(t_1)$ and $\mathbf{z}(t_2)$ denote the spatial locations of a wild  animal species at times $t_1$ and $t_2$. The semivariance function $\gamma(t_1, t_2)$ measures spatial distance variability between locations $\mathbf{z}(t_1)$ and $\mathbf{z}(t_2)$. That is, the semivariance function quantifies average dissimilarity between two locations of a wild animal species across all time lags ($\tau = t_2 - t_1$) in the dataset. The plot of estimated semivariance against time lag reveals empirical semivariance, offering insights into the mix of processes represented in the relocation dataset. The semivariance function holds the majority of information found in the autocorrelation function, signifying their partial equivalence. As per \cite{fleming2014fine}, avoiding direct estimates of mean and variance benefits the semivariance function, providing unbiased estimators—unlike the autocorrelation function.  Estimating the semivariance function with confidence intervals contains information for all possible time lags in the dataset. \\

Traditional time-series methods for semivariance functions assume stationarity, but ecological systems with daily, seasonal, or annual cycles violate this assumption. Analyzing movement data of wild animals in such systems necessitates a nonstationary approach. In nonstationary processes of wild animal movement, the semivariance between locations at times $t_{1}$ and $t_{2}$ can be influenced by both the lag $\tau = t_{2}-t_{1}$ and the absolute times. We view nonstationarity as a nuisance factor, and to mitigate its impact, we calculate the average time $\bar{t} = \left( t_{1} + t_{2}\right)/2$ for each location pair. By averaging over dependence on $\bar{t}$, we modify the semivariance function to depend solely on the time lag, yielding a time-averaged semivariance function analogous to the time-averaged autocorrelation function. Reliable semivariance estimates are limited to the lag range $t_{d} < \tau \ll T$ in the analysis of individual wild animal movement, where $t_{d} $ is the sampling time step and $T$ is the sampling duration. The semivariance between two sample locations of a wild animal tends to increase with distance but plateaus beyond a certain point, reaching the variance around the average value and ceasing further increase. The time-averaged semivariance function is averaged over pairs of locations with a specified time lag $ \tau$ to quantify wild animal movement across a broader range of timescales. Matheron's method of moments estimator is a widely used formula for semivariance computation \cite{cressie2015statistics}, and its implementation as an algorithm depends on the data configuration. For evenly sampled data with a scalar lag $\tau$, semivariances are computed at integral multiples of the sampling interval. Utilizing the nonstationary approach, the method-of-moments estimator for semivariance in evenly sampled data is expressed as:
\begin{align}\label{Semivariance}
\hat{\gamma}\left(\tau\right)  = \displaystyle\frac{1}{2n\left(\tau\right) }\sum_{\bar{t}}\left[ \mathbf{z}\left(\frac{\bar{t} + \tau}{2}\right)-\mathbf{z}\left(\frac{\bar{t} - \tau}{2}\right)\right] ^{2}.
\end{align}
Here, $n\left(\tau\right)$ denotes the number of wild animal location pairs with a lag of $\tau$, and $\hat{\gamma}\left(\tau\right)$ is computed by summing over the time-average values $\bar{t}$ for the given lag $\tau$. In other words, $\hat{\gamma}\left(\tau\right)$ represents the average squared distance between two wild animal locations observed at different times with a time lag $\tau$. Averaging over time values $\bar{t}$ at a specific lag $\tau$ helps compute the semivariance estimator, addressing nonstationarity. For more details, refer to \cite{fleming2014fine}. With increasing lag, fewer wild animal locations are available for semivariance estimation. Thus, the most reliable semivariance estimates emerge from shorter lags in evenly sampled relocation data, underscoring the significance of treating fine-scale features at smaller lags as equally or more critical than larger-scale features at larger lags.\\

The empirical semivariance, derived from equations \eqref{Semivariance}, can be plotted against the time lag between relocations. This provides an unbiased visualization of the autocorrelation structure in the relocation data. As per \cite{fleming2014fine}, analyzing semivariance behaviour near the origin, its shape for intermediate lags, and long-lag patterns can help diagnose model fitting issues. It is important to note that semivariance typically exhibits a linear increase over immediate time lags, signifying autocorrelation in relocation data. On the other hand,  the semivariance tends to flatten out after a long-lag, namely the \emph{range}, if it exists. The long-lag behaviour of semivariances reveals information about the use of space by the  wild animal species. In range-resident wild animals, the semivariance of their relocation data is expected to plateau, reaching an asymptote proportional to their home range.  If the semivariance does not approach an asymptote with increasing time lag, the relocation data may be unsuitable for home range analysis. This could suggest insufficient tracking duration for the wild animal or that the wild animal is shifting its range. Therefore, if there is no evidence of constrained space use (or no semivariance asymptote is reached), any home range analysis may be inappropriate \citep{calabrese2016ctmm}.\\

Wild animal locations closer in time show higher similarity than those farther apart. Directional persistence in animal motion leads to autocorrelated velocities, implying that an animal's direction and speed at one point correlate with those at other points. According to \cite{calabrese2016ctmm}, position autocorrelation, velocity autocorrelation, and range residency are valuable for classifying continuous-time stochastic processes or movement models.

\subsection{Independent and identically distributed process}
This approach assumes that wild animal locations, as well as velocities, are uncorrelated, a simplification commonly used in traditional home range estimation. Despite its unrealistic nature, this \emph{independent and identically distributed (IID)} process is included in this work for the sake of completeness in the modelling of the home ranges.

\subsection{Brownian process}
The Brownian process extends the discrete-time random walk to continuous time, serving as a fundamental model for movement. It exhibits position autocorrelation but lacks velocity autocorrelation, meaning  the speed and direction of a wild animal are not correlated across adjacent times. In this model, wild animals engage in a simple random search within an infinite and uniform resource distribution. Notably, the model lacks drift and attraction components, rendering it less flexible in capturing diverse movement patterns \citep{turchin1998quantitative}.

\subsection{Ornstein-Uhlenbeck process}
The Ornstein-Uhlenbeck process, an extension of the Brownian process, includes attraction to a point or mean. This mean-reverting model exhibits a tendency to drift towards its long-term mean, with increased attraction as movement deviates from the center. The Ornstein-Uhlenbeck process is suitable for modelling data lacking directional persistence but demonstrating confined space use. Parameters like home range crossing time and position variance are associated with this model, with home range crossing time representing the timescale for a wild animal to cover its home range, as detailed in \cite{fleming2014fine}.
\subsection{Ornstein-Uhlenbeck with foraging process}
The \emph{Ornstein-Uhlenbeck with foraging process} can be a suitable model for analyzing relocation data featuring correlated velocities and restricted use of space. It can be applicable to diverse datasets with fine sampling, revealing velocity autocorrelation and prolonged durations for range residence. This model involves three key parameters: position autocorrelation, velocity autocorrelation, and position variance \citep{fleming2014fine, calabrese2016ctmm}. If the optimal model is the \emph{Ornstein-Uhlenbeck process}, the home range estimation provides values for home ranges and crossing times. Conversely, if the selected model is the \emph{Ornstein-Uhlenbeck with foraging process}, the estimation includes home ranges, crossing times, velocity autocorrelation timescales, and average distance traveled for each wild animal. Wild animal relocation data can show either isotropic or anisotropic behaviour. In an isotropic process, the relocation lacks directional dependence, whereas in an anisotropic process, the relocation of the wild animal varies based on the direction of interest.\\

To sum up, empirical semivariance plots, derived from equation \eqref{Semivariance}, provide insights into wild animal movement behavior. They aid in assessing theoretical semivariance models for IID, Ornstein-Uhlenbeck, and Ornstein-Uhlenbeck with foraging processes. These theoretical models are fitted to empirical data using maximum likelihood, and their comparison via the Akaike information criterion (AIC) determines the most suitable model. The chosen model can subsequently be employed for estimating home ranges through autocorrelated kernel density estimation.
\section{Summary functions for pairs of types} 
The estimator $\hat{K}^{ij}_{inhom}\left(r\right)$ of the inhomogeneous cross-type $K$-function $K^{ij}_{inhom}\left(r\right)$ can be given by 
\begin{align*}
\hat{K}^{ij}_{inhom}\left(r\right) = \displaystyle\frac{1}{\abs{\W}}\sum_{\mathbf{v}\in\mathbfcal{Z}_{i}}\sum_{\mathbf{z}\in\mathbfcal{Z}_{j}}\frac{e\left(\mathbf{v}, \mathbf{z}\right) 1_{\left\lbrace \norm{\mathbf{z}-\mathbf{v}}\le r\right\rbrace}}{\lambda_{i}\left(\mathbf{v} \right)\lambda_{j}\left(\mathbf{z} \right)},
\end{align*} 
where $\mathbfcal{Z}_{i}\subset\W$ is the marginal point process (or sub-process of points) of type (or wild animal) $i$ with intensity $\lambda_{i}\left(\cdot\right)$, $\mathbfcal{Z}_{j}\subset\W$ is the sub-process of points of type (or wild animal) $j$ with intensity $\lambda_{j}\left(\cdot\right)$, $e\left(\mathbf{v}, \mathbf{z}\right)$ is an edge correction weight, and $\abs{\W}$ is the area of the study region $\W$. Edge correction methods can be found in \cite{baddeley2015spatial}.
\section{Spatial interaction modelling results}
This section presents the estimated inhomogeneous cross-type $L$- and $J$-functions. The plots depict the nature of spatial interactions between two black bears, which can be attraction, repulsion/avoidance, or no interaction. If the plot of the estimated inhomogeneous cross-type $L$- function from the observed data (the blue-colored curve in the figures below) entirely lies within the envelope, then there is no spatial interaction between the pair of black bears. If the curve of the estimated cross-type $L$-function wanders outside of the envelope from above, then there is an attraction between the pair of black bears. Otherwise, there is a repulsion/avoidance between the pairs of black bears. If, on the other hand, the estimated cross-type $J$-function from the observed data wanders outside of the envelope from below, then there is an attraction between the pairs of black bears. If the estimated cross-type $J$-function from the observed data wanders outside of the envelope from above, the pair of black bears have a repulsive/avoidance interaction. Figures \ref{Result5} to \ref{Result3} show the spatial interactions that may exist between the stated black bears. However, this conclusion must be supported by the statistical test of the interactions.\\

We also conducted a hypothesis test of the null hypothesis, which states that there is no spatial interaction between pairs of black bears. To test the hypothesis, we used statistical methods such as the maximum absolute deviation test (MAD test) and the Diggle-Cressie-Loosmore-Ford test (DCLF test).  Tables from \ref{tab:487To491} to \ref{tab:491To501} depict the hypothesis testing results of the study.  For a given inhomogeneous cross-type function and statistical test method, the $p$-values in each cell of the tables correspond to the rows of the figure panels.  In general, the statistical analyses using the inhomogeneous cross-type $J$- and $L$-function indicate that black bears do not show a preference for living in close proximity or at a distance from one another.

\begin{figure}[H]
\centering
\includegraphics[width=0.35\textwidth, height=0.2\linewidth]{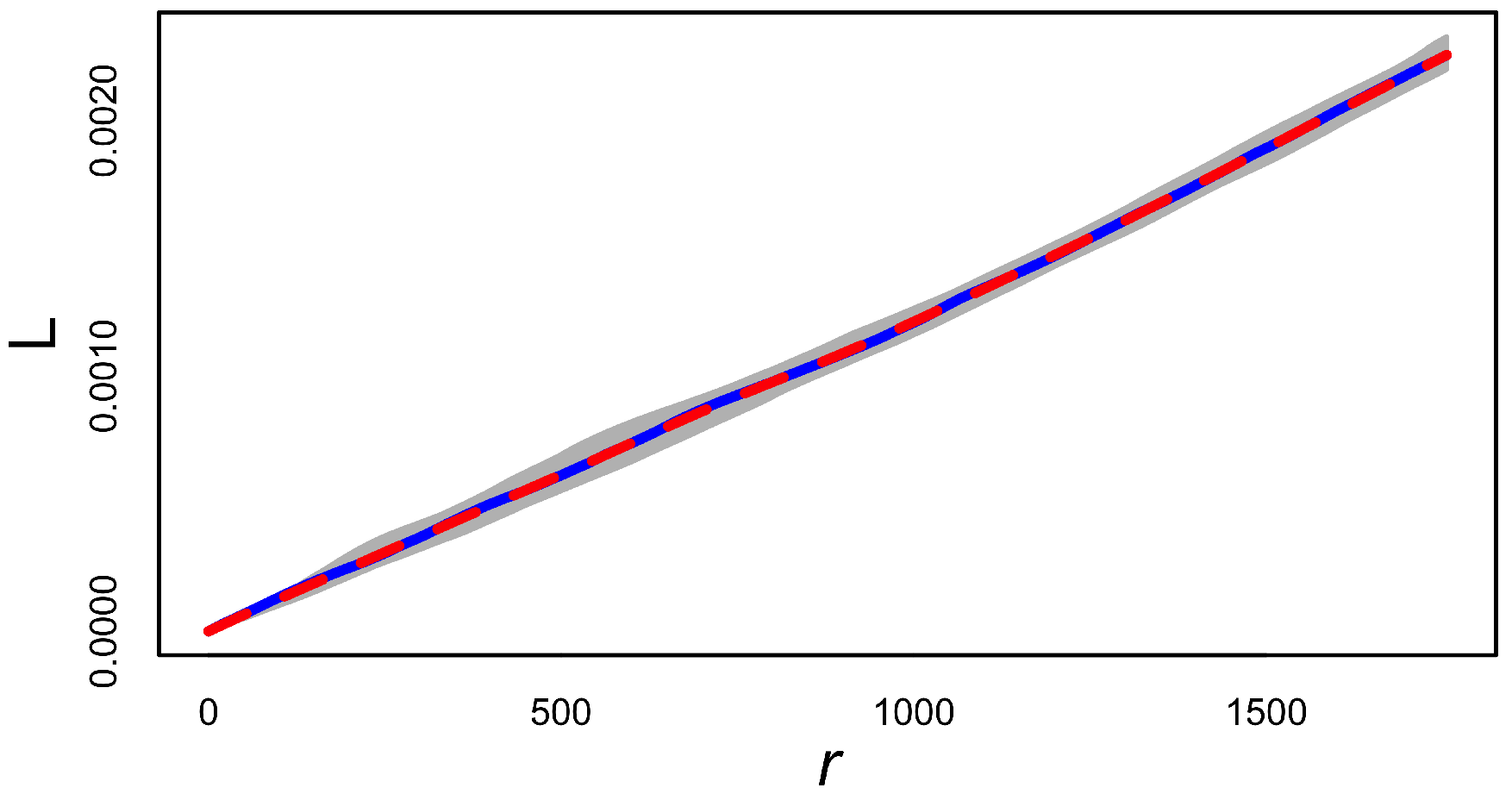}
\includegraphics[width=0.35\textwidth, height=0.2\linewidth]{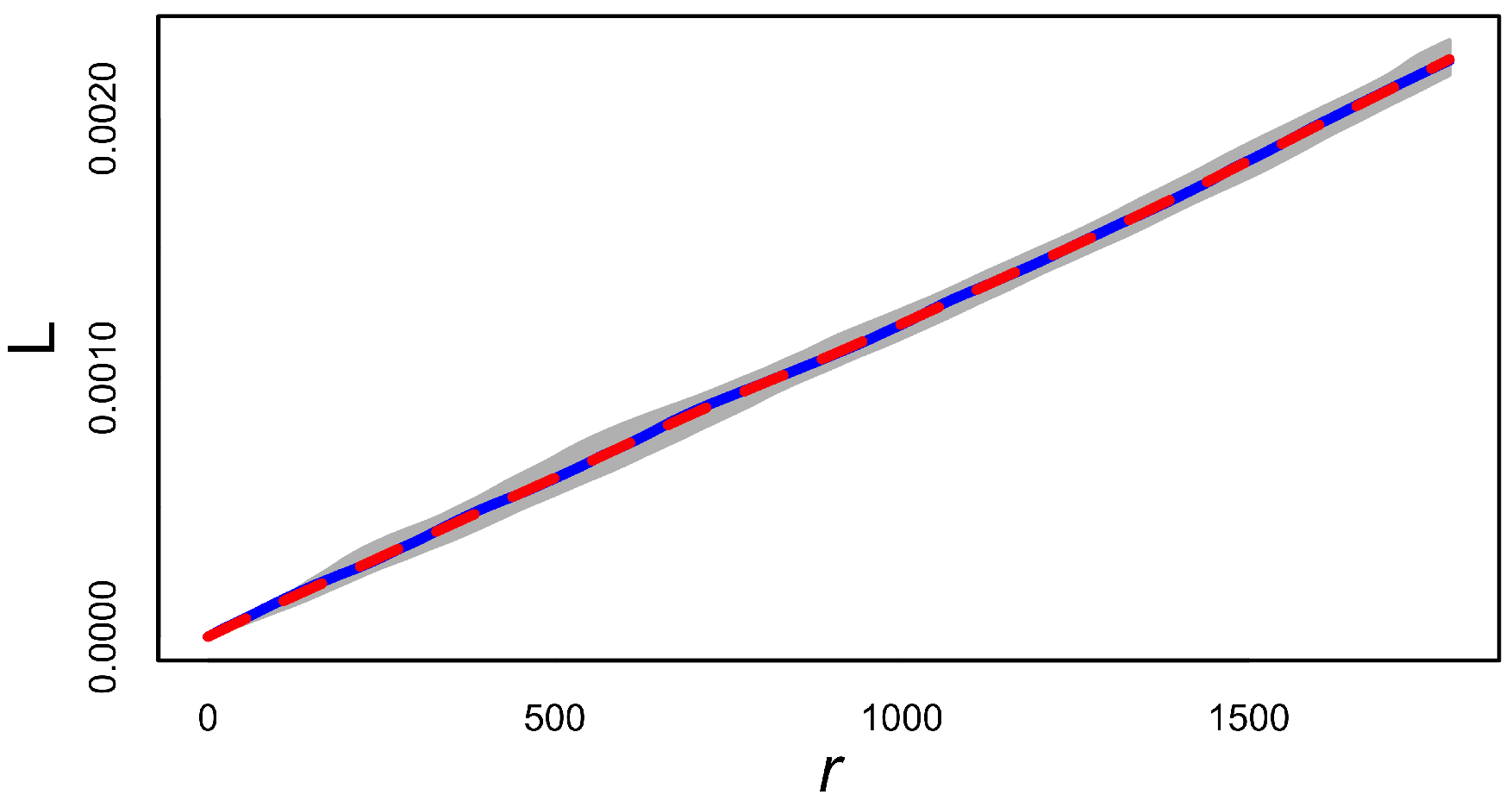}
\includegraphics[width=0.35\textwidth, height=0.2\linewidth]{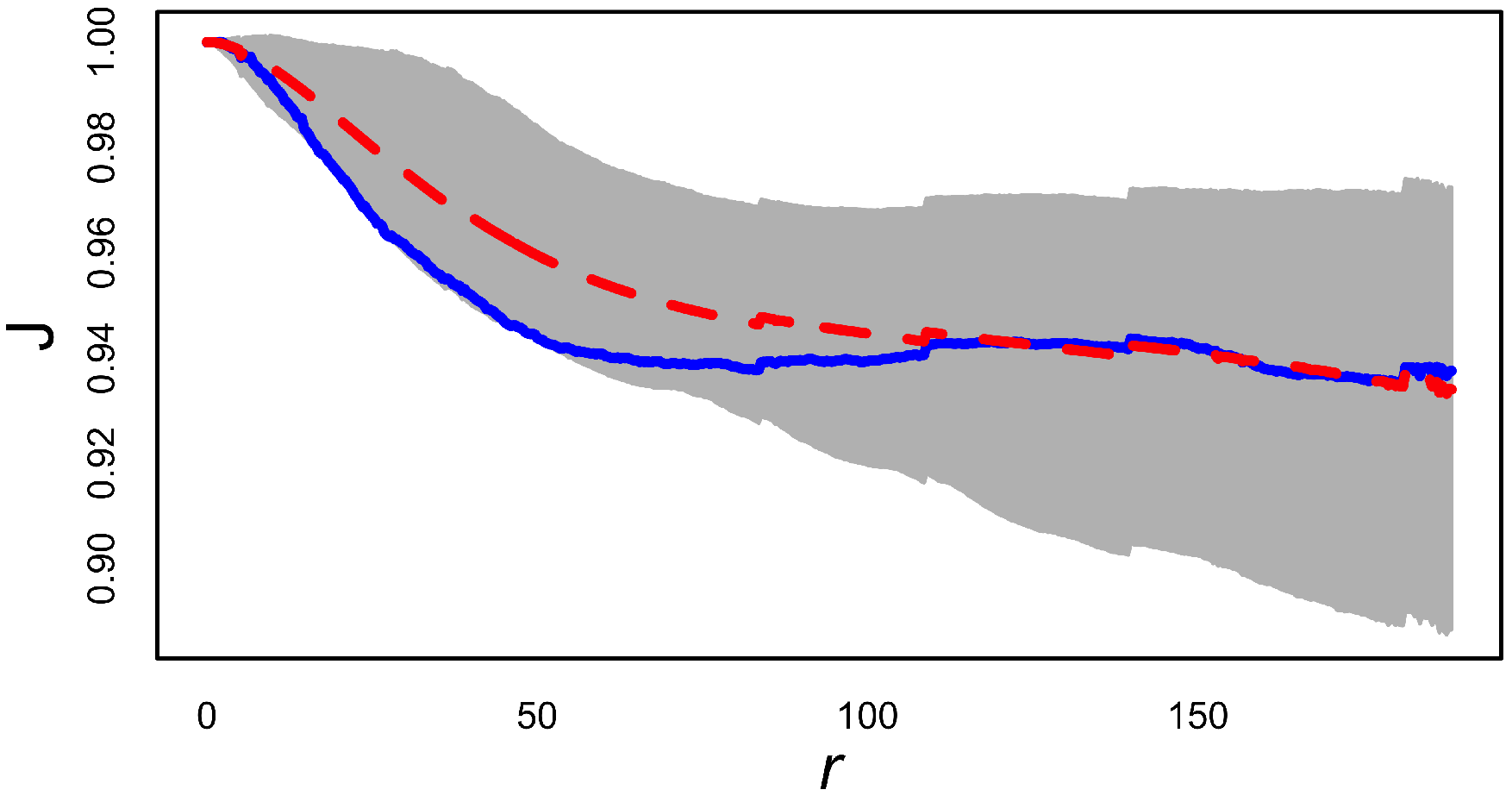}
\includegraphics[width=0.35\textwidth, height=0.2\linewidth]{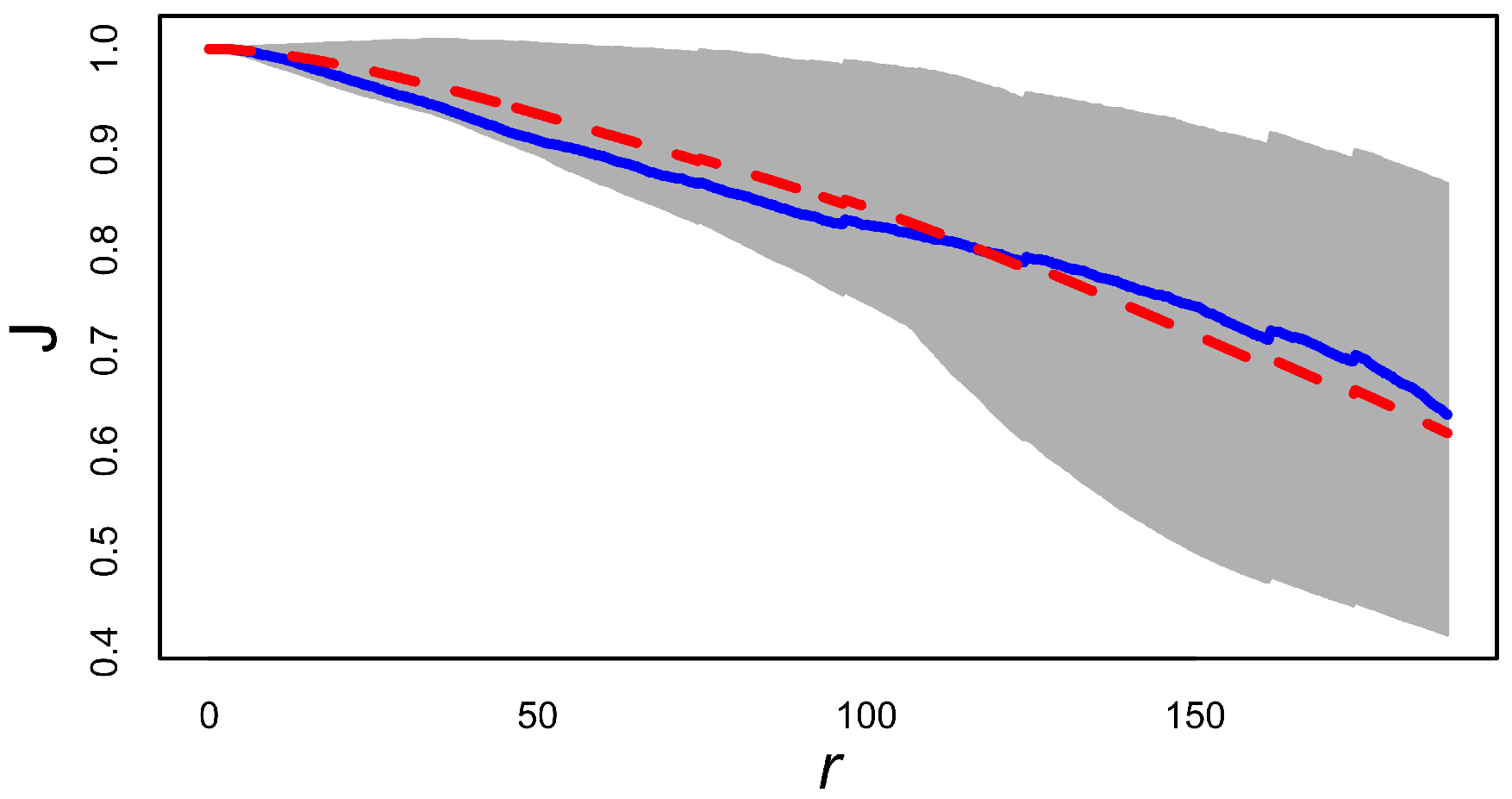}
\begin{table}[H]
\centering
\caption{Two-sided \textit{MAD} and \textit{DCLF} tests for lack of spatial interaction}
\label{tab:487To491}
\begin{tabular}{|c|c|c|c|c|c|}
\hline \multicolumn{1}{|c|}{Summary statistic}& \multicolumn{2}{c|}{MAD test} & \multicolumn{2}{c|}{DCLF test} \\\hline\hline
& 487  to 491 &491 to 487&487  to 491  &491 to 487\\\hline\hline
L&0.994&0.994&0.997&0.998\\\hline
J&0.701&0.946&0.677& 0.904\\\hline\hline
\end{tabular}
\end{table}
\caption{Simulation envelopes for estimated inhomogeneous cross-type L- function (first row) and J-functions (second row) from black bear 487  to 491 (first column)  and  from 491 to 487 (second column) based on 2500 simulations. The table exhibits the $p$-values of the statistical test methods  for inhomogeneous cross-type $L$-and $J$-functions.}
\label{Result5}
\end{figure}

\begin{figure}[H]
\centering
\includegraphics[width=0.35\textwidth, height=0.2\linewidth]{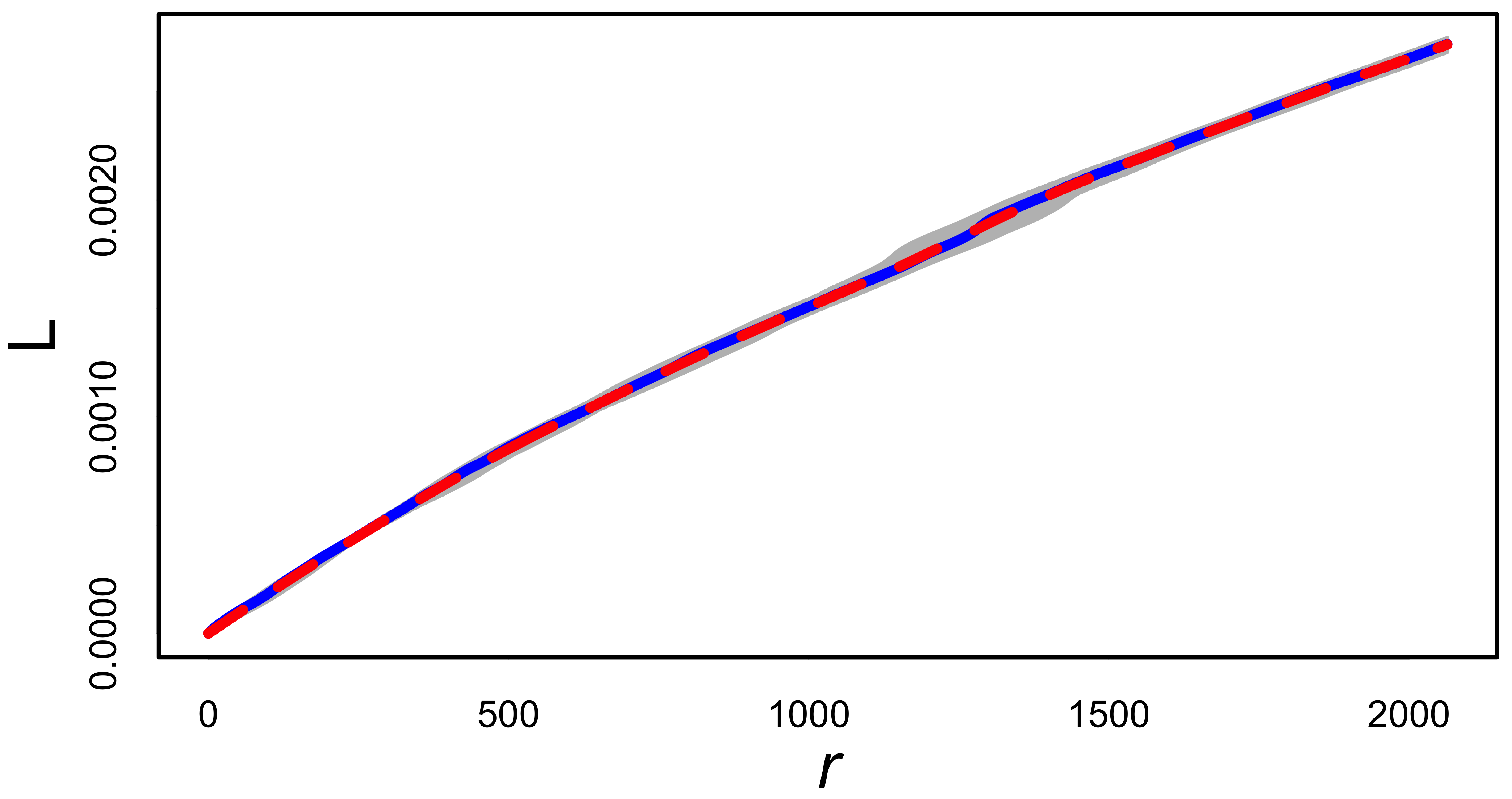}
\includegraphics[width=0.35\textwidth, height=0.2\linewidth]{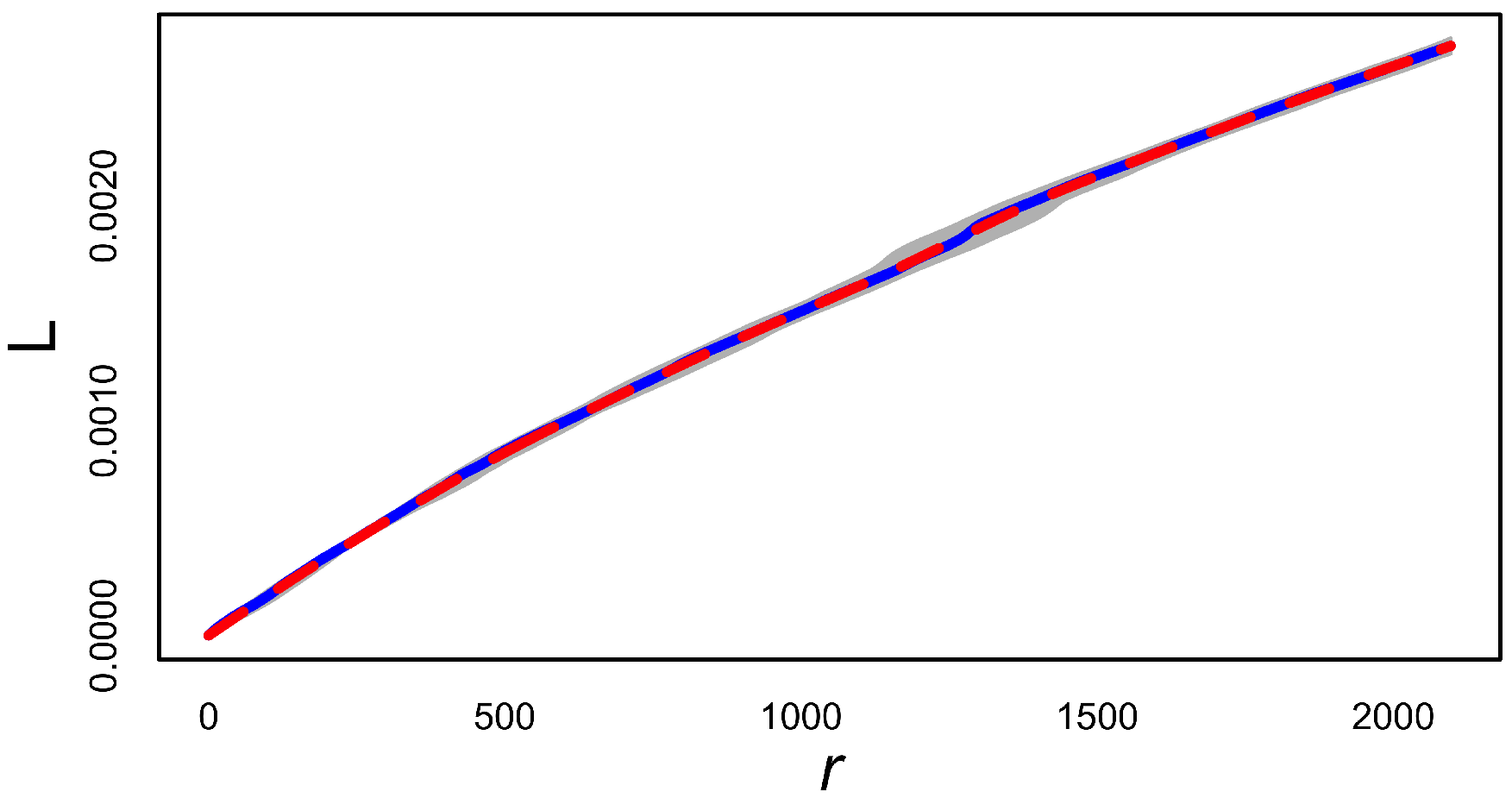}
\includegraphics[width=0.35\textwidth, height=0.2\linewidth]{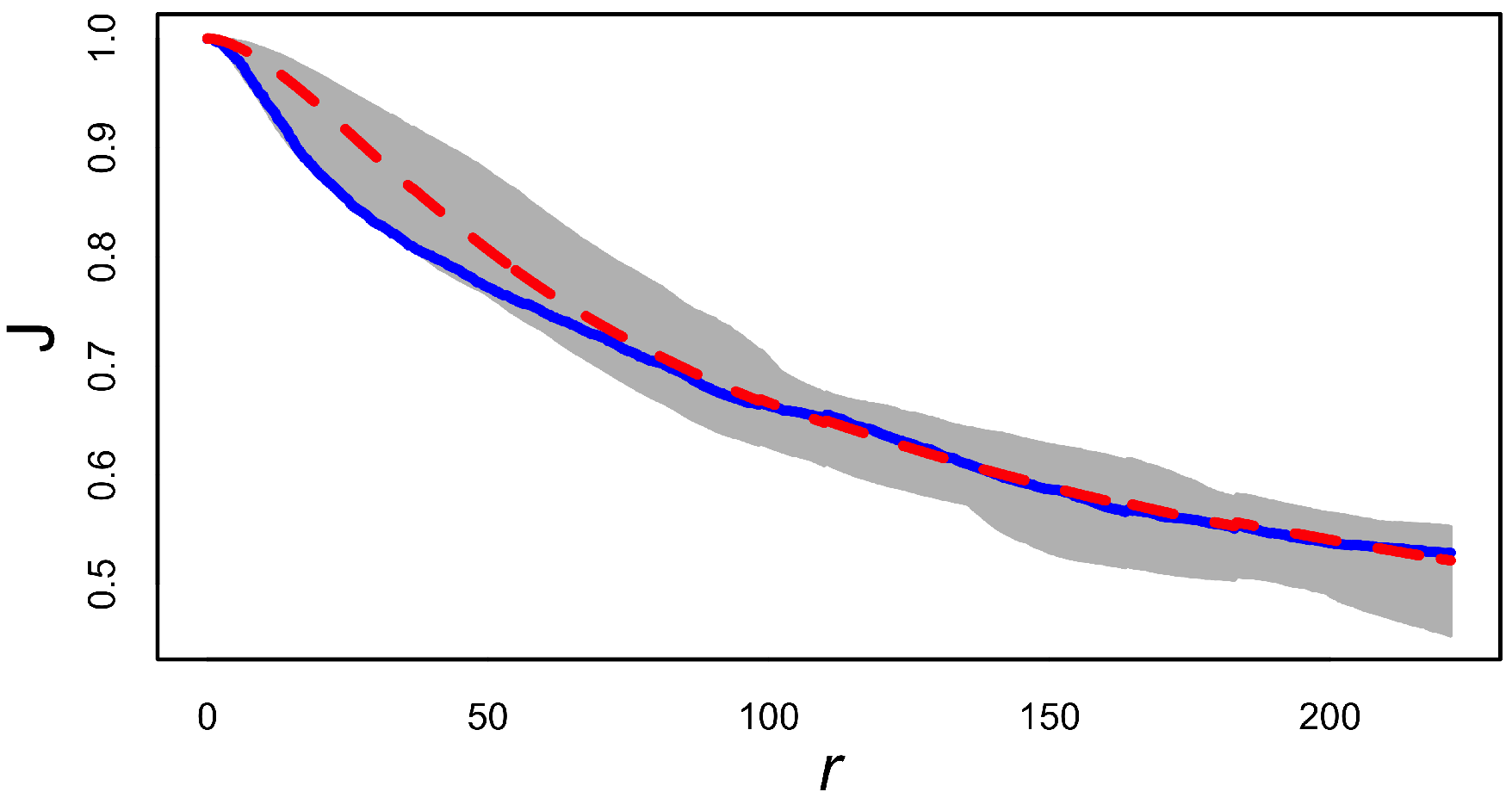}
\includegraphics[width=0.35\textwidth, height=0.2\linewidth]{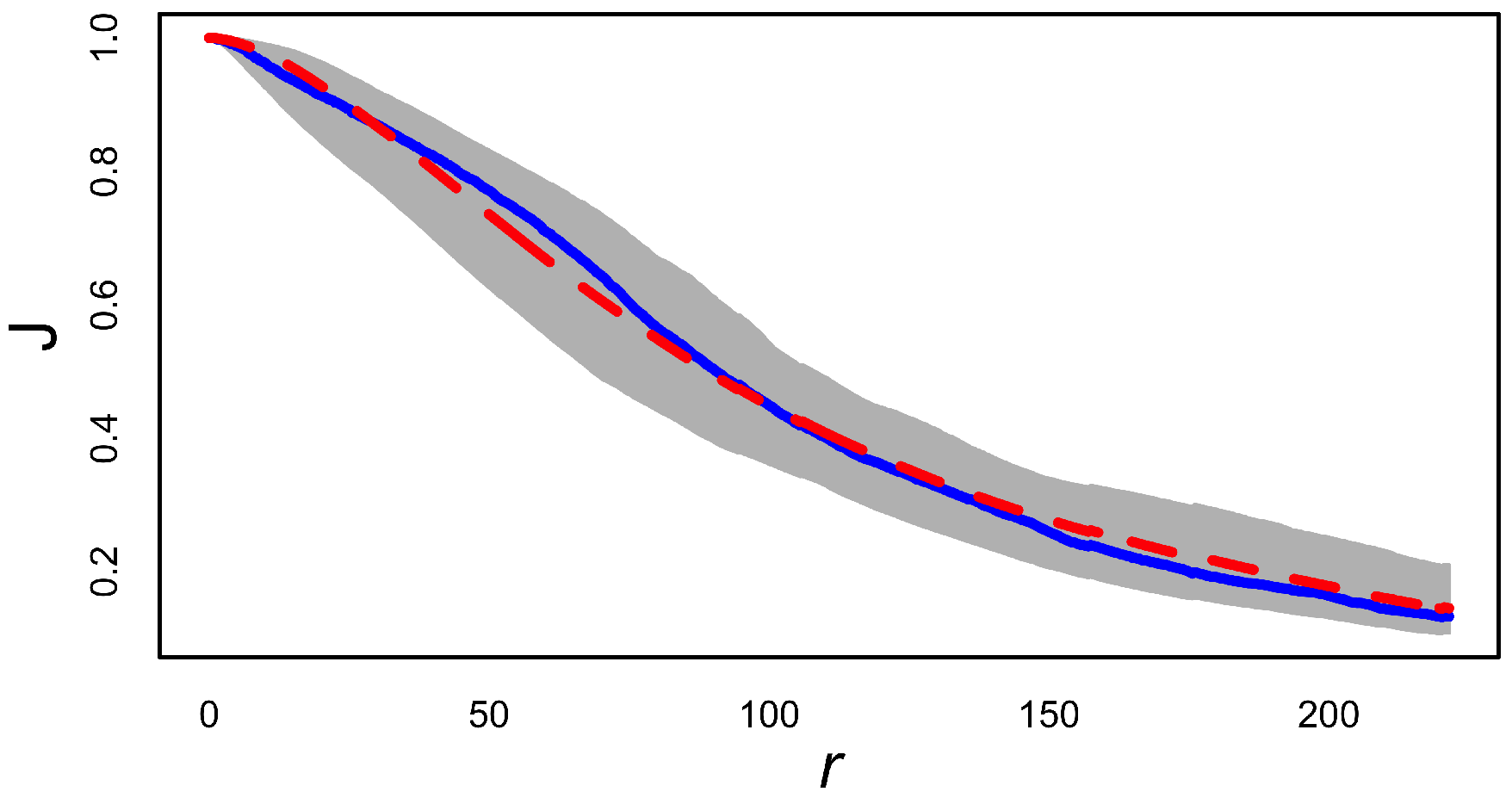}
\begin{table}[H]
\centering
\caption{Two-sided \textit{MAD} and \textit{DCLF} tests for lack of spatial interaction}
\label{tab:487To490}
\begin{tabular}{|c|c|c|c|c|c|}
\hline \multicolumn{1}{|c|}{Summary statistic}& \multicolumn{2}{c|}{MAD test} & \multicolumn{2}{c|}{DCLF test} \\\hline\hline
& 487  to 490  &490 to 487& 487  to 490  &490 to 487\\\hline\hline
L&0.987&0.981&0.992&0.992\\\hline
J&0.052&0.867&0.196& 0.914\\\hline\hline
\end{tabular}
\end{table}
\caption{Simulation envelopes for estimated inhomogeneous cross-type L- function (first row) and J-functions (second row) from black bear 487  to 490 (first column)  and  from 490 to 487 (second column) based on 2500 simulations. The table demonstrates the $p$-values of the statistical test methods  for inhomogeneous cross-type $L$-and $J$-functions.}
\label{Result4}
\end{figure}

\begin{figure}[H]
\centering
\includegraphics[width=0.45\textwidth, height=0.3\linewidth]{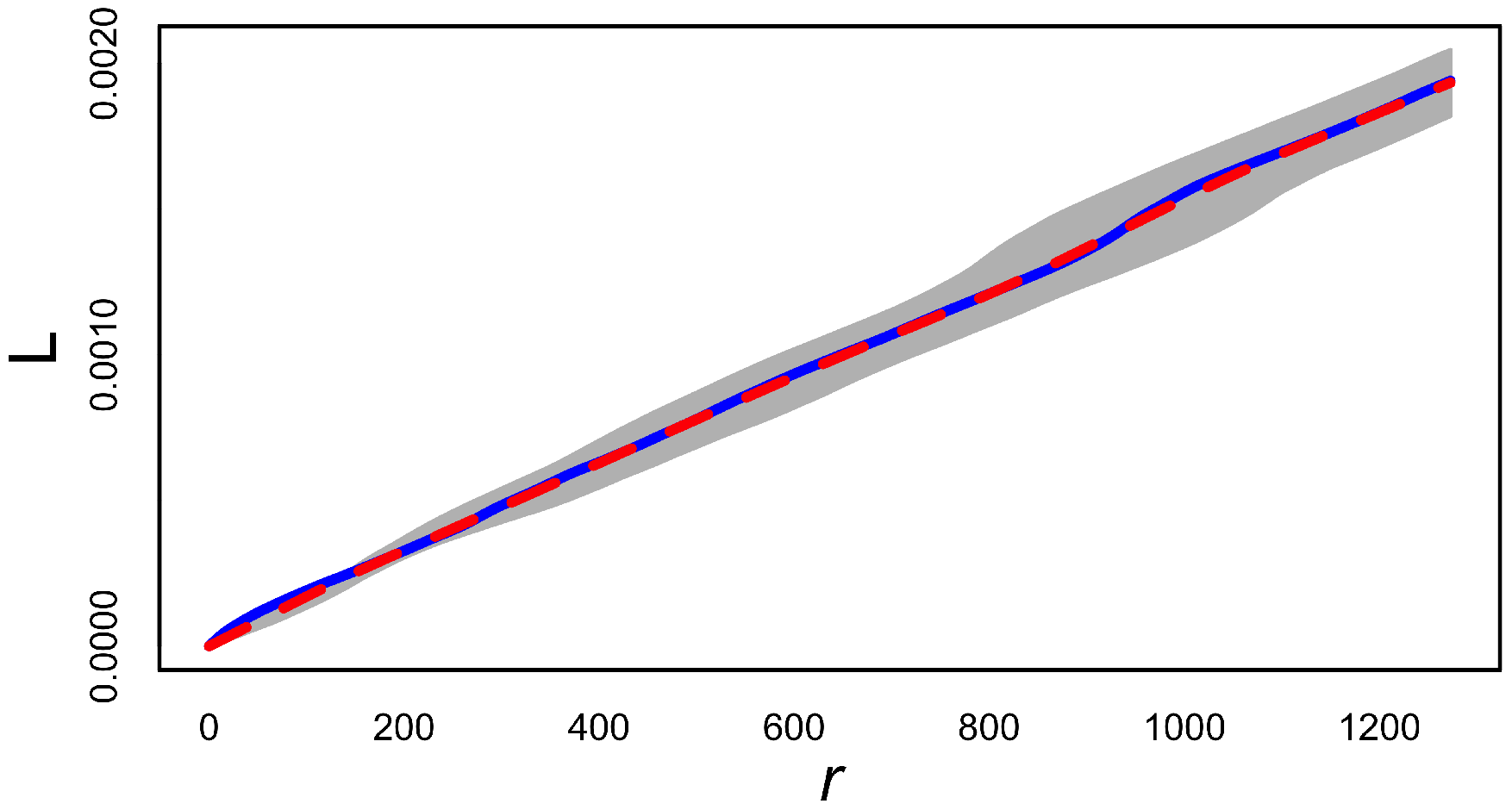}
\includegraphics[width=0.45\textwidth, height=0.3\linewidth]{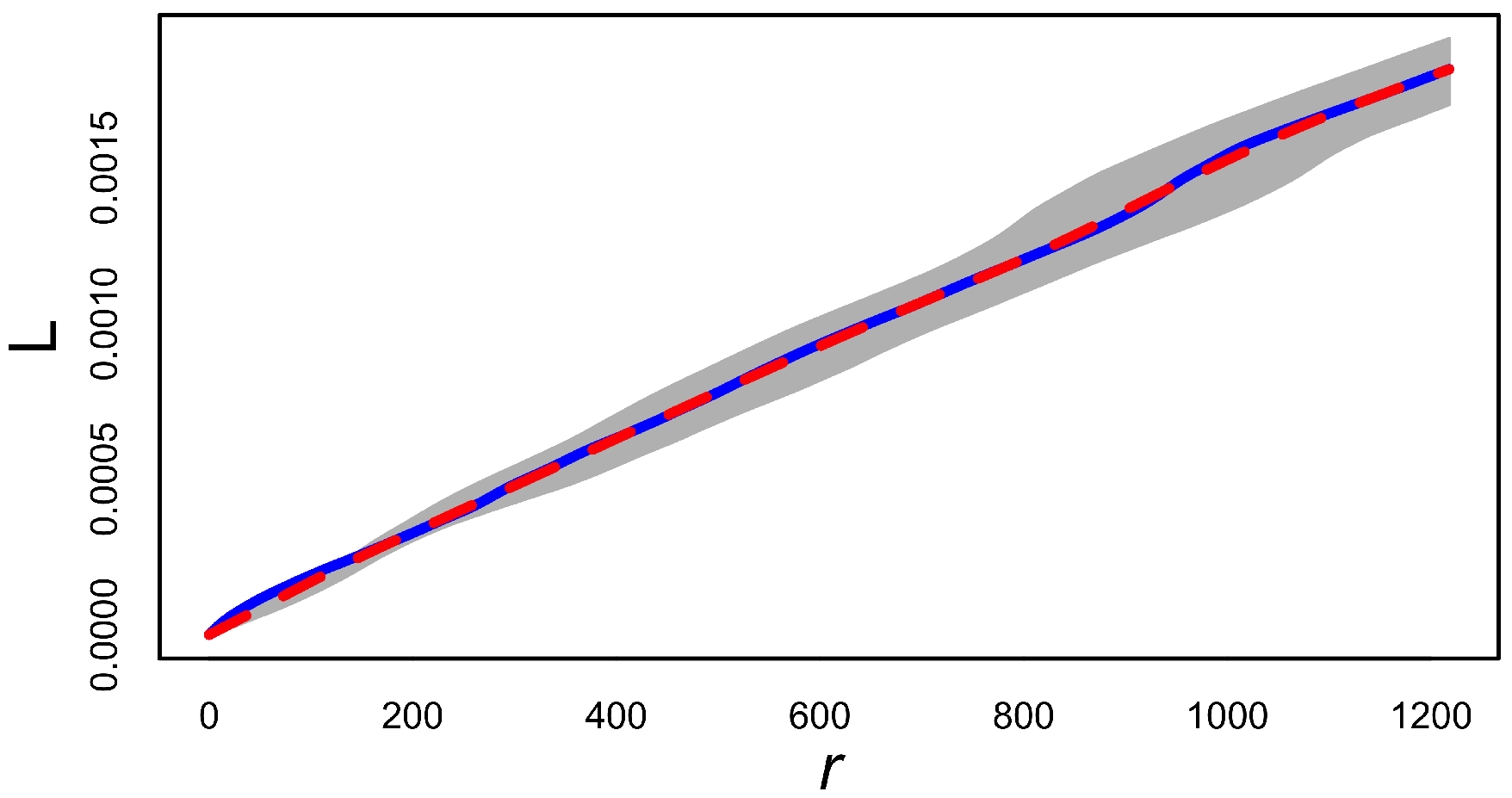}
\includegraphics[width=0.45\textwidth, height=0.3\linewidth]{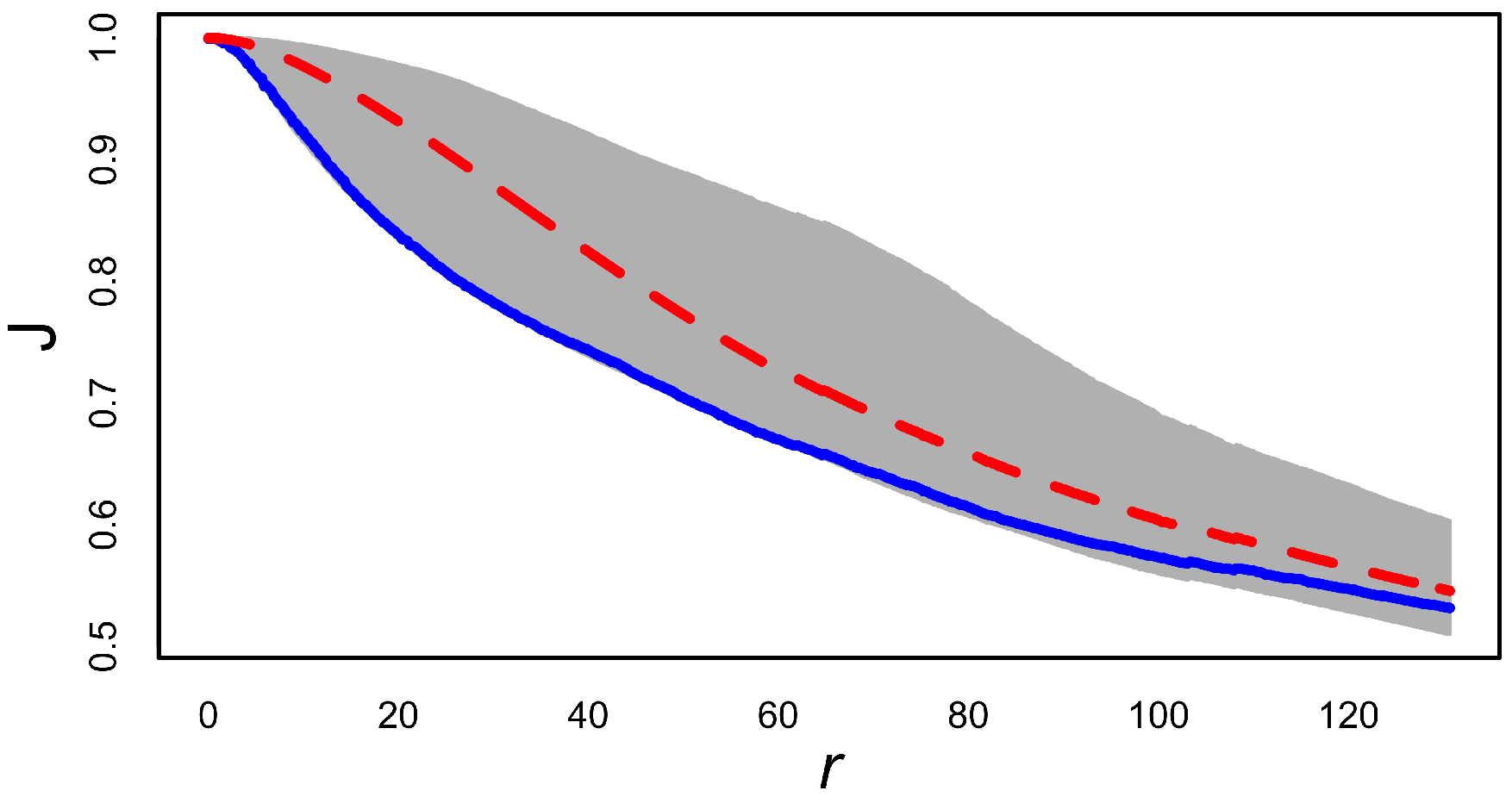}
\includegraphics[width=0.45\textwidth, height=0.3\linewidth]{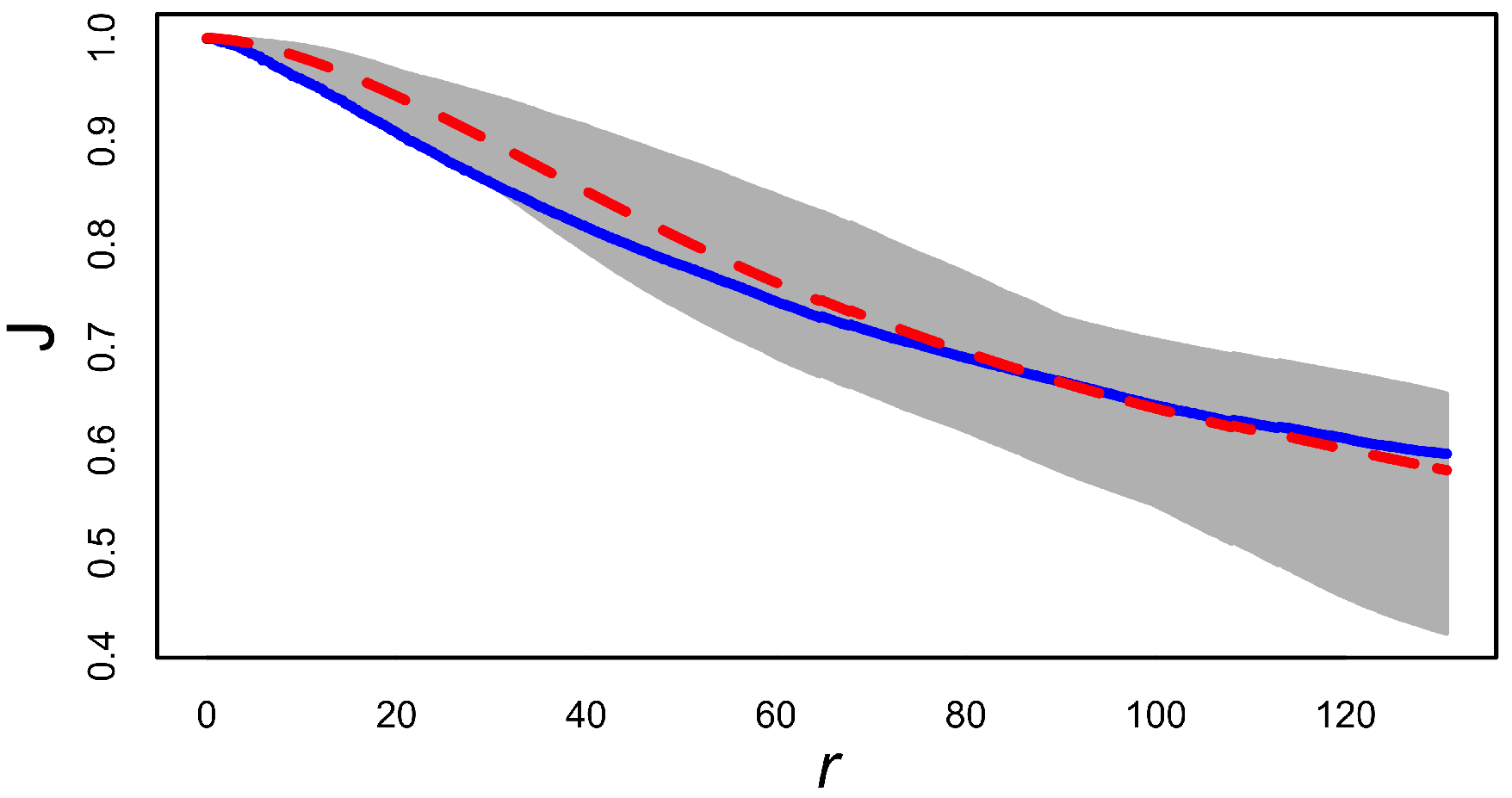}
\begin{table}[H]
\centering
\caption{Two-sided \textit{MAD} and \textit{DCLF} tests for lack of spatial interaction}
\label{tab:491To501}
\begin{tabular}{|c|c|c|c|c|c|}
\hline \multicolumn{1}{|c|}{Summary statistic}& \multicolumn{2}{c|}{MAD test} & \multicolumn{2}{c|}{DCLF test} \\\hline\hline
& 491  to 501   &501 to 491 & 491  to 501   &501 to 491 \\\hline\hline
L&0.928&0.922&0.981&0.972\\\hline
J&0.029&0.527&0.042& 0.630\\\hline\hline
\end{tabular}
\end{table}
\caption{Simulation envelopes for estimated inhomogeneous cross-type L- function (first row) and J-functions (second row ) from black bear 491  to 501 (first column)  and  from 501 to 491  (second column) based on 2500 simulations. The table illustrates the $p$-values of the statistical test methods  for inhomogeneous cross-type $L$-and $J$-functions.}
\label{Result3}
\end{figure}

\end{document}